\documentclass[a4paper]{article}

\usepackage{amsmath, amssymb}
\usepackage{subfigure}
\usepackage{booktabs} 
\usepackage{floatrow} 
\usepackage{graphicx} 
\usepackage{xcolor}
\usepackage{rotating}
\usepackage[authoryear,round]{natbib}

\newcommand{\keywords}[1]{{\bf Keywords:} {#1}}

\newcommand{\JEL}[1]{{\bf JEL:} {#1}}

\newtheorem{thm}{theorem}[section]

\newtheorem{proposition}[thm]{Proposition}

\newtheorem{lemma}[thm]{Lemma}
\newenvironment{proof}{{\bf Proof. }}{\hfill$\Box$}
\newenvironment{MyProof}[1]{{\bf Proof {#1}}}{\hfill$\Box$}

\newcommand{\MakeTitle}{\maketitle\newcommand{\and}{$\cdot$ }}

\title{Seasonal Stochastic Volatility and the Samuelson Effect in Agricultural Futures Markets
    \thanks{
		We would like to thank
        Carol Alexander, Roy Cerqueti, David Criens, Ennio Fedrizzi, Ren\'{e} Flacke, Yevhen Havrylenko,
        Kristofer J\"{u}rgensen, R\"{u}diger Kiesel, Fran\c{c}ois Le Grand, Cassio Neri,
        and participants at the
        TU Munich 2017 Innovations in Insurance, Risk- and Asset Management Conference,
        the Nantes 2018 Commodity Markets Winter Workshop,
        the Duisburg-Essen Energy and Finance Seminar,
        the 35th Annual Conference of the French Finance Association in Paris,
        and the Commodity and Energy Markets Association Annual Meeting 2018 at La Sapienza in Rome,
        for helpful and stimulating comments, discussions and suggestions.
        Part of this paper was written while Lorenz Schneider was KPMG Visiting Professor at the Technische Universit\"{a}t M\"{u}nchen.
        A first version of this paper was initially circulated in June 2015 under the title
        ``Seasonal Stochastic Volatility and Correlation Together with the Samuelson Effect in Commodity Futures Markets''.	
	}
}

\author{
	Lorenz Schneider
	\thanks{EMLYON Business School, \texttt{schneider@em-lyon.com}.}
    \quad \quad
	Bertrand Tavin
	\thanks{EMLYON Business School, \texttt{tavin@em-lyon.com}.}
}

\date{\today}

\begin{document}

\MakeTitle

\begin{abstract} %
We introduce a multi-factor stochastic volatility model for commodities that incorporates seasonality and the Samuelson effect.
Conditions on the seasonal term under which the corresponding volatility factor is well-defined are given,
and five different specifications of the seasonality pattern are proposed.
We calculate the joint characteristic function of two futures prices for different maturities in the risk-neutral measure.
The model is then presented under the physical measure, and its state-space representation is derived, in order to estimate the parameters
with the Kalman filter for time series of corn, cotton, soybean, sugar and wheat futures from 2007 to 2017.
The seasonal model significantly outperforms the nested non-seasonal model in all five markets,
and we show which seasonality patterns are particularly well-suited in each case.
We also confirm the importance of correctly modelling the Samuelson effect in order to account for futures with different maturities.
Our results are clearly confirmed in a robustness check carried out with an alternative dataset of constant maturity futures for the same agricultural markets.

\bigskip


\bigskip





\JEL{C63 \and C51 \and C52 \and G13}
\end{abstract}


\section{Introduction}
\label{s:Introduction}

Seasonality is a well-known empirical feature of many commodities markets.
In the energy sector, among fossil fuels, natural gas futures curves,
and among refined products, gasoline, heating oil and fuel oil futures curves all typically display seasonality.
In the agricultural sector, almost all futures curves show seasonality due to harvest times and the seasons of the year.

It is important to distinguish from the outset between two types of seasonality:
seasonality of futures prices and seasonality of volatility of futures prices.

Regarding seasonality of prices, consider agricultural commodities such as corn, soybeans and wheat.
These tend to be in low supply in the months preceding the harvest, and in high supply after the harvest in summer.
This typically leads to relatively high futures prices of contracts with delivery months in late winter or spring,
and low prices of futures contracts with delivery months in the summer or early fall.
Therefore, when the prices of these contracts are plotted as a function of their maturity, they tend to rise and fall with the maturity in some seasonal way.
In other words, the futures curve shows seasonality.
However, the price of an individual futures contract with a given maturity should not rise and fall over time in any kind of seasonal way:
indeed, such a behaviour would lead to easy arbitrage opportunities.

Regarding seasonality of volatility, the situation is different in the sense that now an individual futures contract, with fixed maturity,
tends to go through phases of relatively high or low volatility according to a seasonal pattern.
To take again the example of agricultural commodities, the weather in the months leading up to the harvest has a direct impact on
its quality and quantity, and futures prices can fluctuate strongly as forecasts for the new crop change.
In contrast to this, weather patterns in winter tend to be of minor consequence for the harvest, and futures prices tend to fluctuate less strongly.

It follows from these empirical observations that for commodity models, seasonality is usually only an issue for the volatility, but not for the futures price itself.
Mathematically, individual futures prices are modelled as martingales in the risk-neutral measure, and martingales do not have a tendency to rise or fall in a pre-determined way.
\citet{Clark2014} and \citet{RoncoroniFusaiCummins2015} give general discussions and numerous examples of seasonality in various commodities markets.

Traditionally, there are two approaches to modelling the prices of futures contacts: futures-based models and spot-based models.
An advantage of futures-based models is that since the futures price curve is an input of the model,
any arbitrage-free shape of the initial futures curve can be accommodated, including any type of seasonality.
In contrast, a first step for spot-based models is to make them fit the initial futures curve, which uses up model parameters and doesn't necessarily
lead to satisfactory results.

\citet{Sorensen2002} studies the modelling of seasonality in corn, soybean and wheat futures markets.
Analysis of a large data set of CBOT futures prices data from 1972 to 1997 confirms clearly that futures prices exhibit a seasonality.
Another feature that is suggested by the data is seasonal behaviour of the futures price volatilities.
In this vein, \citet{RichterSorensen2002} propose a model for the spot price of soybeans based on seasonal stochastic volatility.
\citet{GemanNguyen2005} also introduce a spot-based model for soybean prices with seasonality both for the price level and the (possibly stochastic) volatility level.
\citet{BackProkopczukRudolf2013} analyze data from corn, soybean, heating oil and natural gas markets and compare various spot-based models with deterministic seasonal volatility.
They conclude that a volatility with seasonality is an important feature when valuing options on futures in these markets.
\citet{ArismendiBackProkopczukPaschkeRudolf2016} also study a futures-based model with seasonal stochastic volatility, which is based on the \citet{Heston1993} stochastic volatility model
with deterministic, seasonal mean-reversion level in the square-root process followed by the variance.
\citet{SchmitzWangKimn2013} study calendar spread options in agricultural grain markets relying on a joint Heston model for the two underlying futures contracts.
These two contracts share the same variance process, which has a constant mean-reversion level, and therefore does not display seasonality.
In the context of interest rates, the \citet{CoxIngersollRoss1985} (CIR) model has been extended to time-dependent parameters by \citet{Maghsoodi1996},
and the \citet{Heston1993} model with time-dependent parameters, including the correlation between the spot price and its variance, has been studied by \citet{BenhamouGobetMiri2010}.
Let us also note that in the context of electricity markets, \citet{LuciaSchwartz2002} give a detailed justification of the choice of seasonality function,
as do \citet{GemanRoncoroni2006}.

In parallel to his remark about seasonality in futures prices, \citet{Sorensen2002} confirms the \citet{Samuelson1965} hypothesis that
``the variations of distant maturity futures are lower than nearby futures prices.''
We call this pattern the \textit{Samuelson effect.} 
Popular futures-based models that incorporate this effect are those of \citet{ClewlowStrickland1999_1,ClewlowStrickland1999_2}.
The volatility functions used in these models are deterministic.
\citet{SchneiderTavin2018} extend the multi-factor model of \citet{ClewlowStrickland1999_2} to incorporate stochastic volatility.
Not only is stochastic volatility an incontestible empirical feature of prices in futures markets,
its inclusion also allows to calibrate the model to option volatility smiles and skews typically seen in futures option markets.
In agricultural markets, a reflection of stochastic volatility is the introduction in 2011 and 2012 of three volatility indices on the CBOE/CBOT:
the Corn Volatility Index (CIV), Soybean Volatility Index (SIV), and the Wheat Volatility Index (WIV).

In this paper, we extend the model introduced in \citet{SchneiderTavin2018} to incorporate seasonal trends in the stochastic variance processes.
To achieve this, we begin in Section \ref{s:SeasonalStochasticVolatility} by studying the mathematical conditions to impose on the seasonality function
to guarantee that the generalized Cox-Ingersoll-Ross (CIR) process of \citet{CoxIngersollRoss1985} retains important features,
such as existence and uniqueness of a strong solution, and positivity.
These conditions appear to be not only interesting from a theoretical point of view, but also useful in practice:
different markets may need to be modelled with different seasonality patterns for the volatility, and we therefore propose five seasonality patterns.

We then introduce the model with seasonal volatility in Section \ref{s:SeasonalStochasticVolatilityModelForCommodityFutures} in the risk-neutral measure and show how,
by an extension of the results in \citet{SchneiderTavin2018}, the joint characteristic function (c.f.) of the log-returns of two futures prices can be obtained.
It turns out that the Riccati ordinary differential equation (ODE) for the first function $A$ is not affected,
and only the integral ODE for the second function $B$ depends on the time-dependent, deterministic mean-reversion level and is altered.
Therefore, the same closed-form solution for $A$ as in \citet{SchneiderTavin2018} can be used.
With the joint c.f. at hand, European options and calendar spread options on futures can be priced easily and rapidly.
We conclude the section by introducing market prices of risk and presenting the model in the physical measure.

In Section \ref{s:StateSpaceRepresentation}, we give a state-space representation of the model in order for it to be used together with the Kalman filter.
When the observed time-series are futures prices or returns, the model is conditionally Gaussian and fits into the classical Kalman filter setup,
and the log-likelihood function for a given set of model parameters is readily evaluated.

In Section \ref{s:ModelEstimation} we present our data for corn, cotton, soybean, sugar and wheat futures contracts.
Since our focus is also on the Samuelson effect, we include all available liquid maturities in our sample,
which span roughly two years from the nearby contract to the last one.
We then estimate five seasonal versions and the non-seasonal version of our model in these agricultural markets.
We see that the seasonal model significantly outperforms the nested non-seasonal model in all five markets,
and show which seasonality patterns are particularly well-suited for each market.
Also, the importance of correctly modelling the Samuelson effect in order to account for futures with different maturities is clearly confirmed by our results.

In order to test the robustness of our results, we have carried out a second parameter estimation based on constant-maturity returns series calculated from the same futures data.
This approach has been proposed by \citet{Galai1979} for building an index for call options, and by \citet{AlexanderKorovilas2013} in the context of VIX futures.
The results obtained with this alternative dataset clearly confirm that our estimations are robust
with respect to which type of time series - concatenated or constant-maturity - we use.

Section \ref{s:Conclusion} concludes,
Appendices \ref{a:ProofCIR}, \ref{a:ProofCF} and \ref{a:Transforms} contain proofs and auxiliary results,
and Appendix \ref{a:altdata_results} contains our estimation results for the constant-maturity futures series.


\section{Seasonal Stochastic Volatility}
\label{s:SeasonalStochasticVolatility}

\subsection{The CIR Process with Time-Dependent Drift}
\label{ss:TheCIRProcessWithTimeDependentDrift}

To our knowledge, \citet{HullWhite1990} were the first to consider extending the \citet{CoxIngersollRoss1985} (CIR) interest rate model to time-dependent coefficients.
They conclude that in this general case, it is no longer possible to obtain European bond option prices analytically.
\citet{Maghsoodi1996} also studies the ``extended'' CIR process in which the parameters $\kappa, \theta$ and $\sigma$ are time-dependent and finds, under certain conditions,
the unique strong solution to the SDE describing the evolution of the process.

In the context of the \citet{Heston1993} stochastic volatility model, the CIR process represents the variance process of a stock price or foreign-exchange rate.
\citet{BenhamouGobetMiri2010} study the ``time dependent Heston model'' and derive analytical formulas approximating European option prices.
In their setup, the mean-reversion parameter $\kappa$ is constant, but the parameters $\theta, \sigma$ and $\rho$ (giving the correlation between the stock price, or foreign-exchange rate, and
its variance) are all allowed to vary with time $t$.

In the model introduced here, we only let the mean-reversion level given by $\theta$ depend on time,
while the other parameters $\kappa > 0$ and $\sigma > 0$ (and later also $\rho$) remain constant.

Let $(\Omega, \mathcal{A}, {\mathbb P}, \mathcal{F})$ be a filtered probability space, and let $B = (B_t)_{t \geq 0}$ be a Brownian motion on this space.
Let $\mathcal{T} = \{ t_i, i = 1, ... \}$ be a set of times having only finitely many points in every bounded interval,
and let $\mathcal{Z} = \{ 0 \leq t_1 < t_2 < ... < t_i < ... \}$ be the partition of $\mathbb{R}_0^+$ defined by $\mathcal{T}$.
Finally, let the seasonality function $\theta: \mathbb{R}_0^+ \to \mathbb{R}^+$ be piecewise continuous with respect to $\mathcal{Z}$,
and assume that it is bounded from below and above by positive constants $\theta_{min}$ and $\theta_{max}$.

We will compare two processes $v$ (seasonal) and $\tilde{v}$ (non-seasonal), which are given, respectively, by the SDEs
\begin{align}
\label{CIR_SDE_TimeDependentTheta}
dv(t) &= \kappa \left( \theta(t) - v(t) \right) dt + \sigma \sqrt{v(t)} dB(t),
\\
\label{CIR_SDE_ConstantTheta}
d \tilde{v}(t) &= \kappa \left( \theta_{min} - \tilde{v}(t) \right) dt + \sigma \sqrt{\tilde{v}(t)} dB(t),
\end{align}
with identical parameters $\kappa > 0, \sigma > 0$ and initial conditions $0 < \tilde{v}(0) = \tilde{v}_0 \leq v(0) = v_0$.

It is well known that \eqref{CIR_SDE_ConstantTheta} has a unique strong solution.
The following result describes the solution to \eqref{CIR_SDE_TimeDependentTheta}.
\begin{proposition}
\label{Prop:CIR_SDE_Solution}
Assume that the seasonality function $\theta$ is piecewise continuous w.r.t. the partition $\mathcal{Z}$ of $\mathbb{R}_0^+$,
and bounded by positive constants $\theta_{min}$ and $\theta_{max}$, i.e. for all $t \geq 0, 0 < \theta_{min} \leq \theta(t) \leq \theta_{max}$.
Let the processes $v$ and $\tilde{v}$ be given by \eqref{CIR_SDE_TimeDependentTheta} and \eqref{CIR_SDE_ConstantTheta}, respectively.
Then:
\begin{enumerate}
\item
The process \eqref{CIR_SDE_TimeDependentTheta} has a unique strong solution with continuous sample paths.
\item
${\mathbb P} \left[ \tilde{v}_t \leq v_t, \forall t \geq 0 \right] = 1$.
\item
If the Feller condition $\sigma^2 < 2 \kappa \theta_{min}$ is satisfied for $\theta_{min}$, then
the process $v$ is strictly positive.
\end{enumerate}
\end{proposition}

We prove this result in Appendix \ref{a:ProofCIR}.

Note that if the Feller condition is violated, then $\tilde{v}$ can possibly reach $0$, but it still cannot become negative.
The piecewise continuity condition on $\theta$ means that discontinuous specifications of the mean-reversion level pose no problems.


\subsection{Seasonality Functions}
\label{ss:SeasonalityFunctions}

We present five types of seasonality functions $\theta$ that can be used as parametric forms to model seasonal variations of the volatility.
These functions are parametric and work with three parameters: $a, b$ and $t_0$.
The parameter $a$ determines the basic volatility level, $b$ the magnitude of the seasonality pattern,
and $t_0$ the time of the year when the volatility reaches its maximum.

The various seasonality patterns will allow for greater flexibility in fitting given futures markets,
since the reasons underpinning the seasonality phenomena in volatility may vary from one market to another.
The first two patterns considered below are smooth and are based on the sine-function.
The three others have points of non-differentiability or discontinuity, and may be used to represent a less regular evolution of the volatility.

\begin{enumerate}
\item
The \underline{sinusoidal pattern} is given, with $a \geq b > 0$ and $t_0 \in [0,1[$, by
\begin{equation}
\theta(t) = a + b \cos{ \left( 2 \pi \left( t - t_0 \right) \right)}.
\label{pattern:sinusoidal}
\end{equation}
\item
The \underline{exponential-sinusoidal pattern} is given, with $a, b > 0$ and $t_0 \in [0,1[$, by
\begin{equation}
\theta(t) = a e^{b \cos {\left( 2 \pi \left( t - t_0 \right) \right) }}.
\label{pattern:exponential-sinusoidal}
\end{equation}
This parametric form for $\theta$ is used in \citet{ArismendiBackProkopczukPaschkeRudolf2016}.
\item
The \underline{sawtooth pattern} is given, with $a, b > 0$ and $t_0 \in [0,1[$, by
\begin{equation}
\theta(t) = a + b \left( t - t_0 - \left \lfloor t - t_0 \right \rfloor \right),
\label{pattern:sawtooth}
\end{equation}
where $\left\lfloor . \right\rfloor$ denotes the floor function.
\item
The \underline{triangle pattern} is given, with $a,b > 0$ and $t_0 \in [0,1[$, by
\begin{equation}
\theta(t) = a + b \left| \frac{1}{2} - \left(t - t_0 - \left \lfloor t - t_0 \right \rfloor \right) \right|.
\label{pattern:triangle}
\end{equation}
\item
The \underline{spiked pattern} is given, with $a,b > 0$ and $t_0 \in [0,1[$, by
\begin{equation}
\theta(t) = a + b\left(\frac{2}{1+\left|\sin(\pi(t-t_0)) \right|}-1 \right)^2.
\label{pattern:spiked}
\end{equation}
This parametric form for $\theta$ can be found in \citet{GemanRoncoroni2006}, where it is used to model the time-varying intensity of a jump process.
\end{enumerate}

Figure \ref{Fig:theta_plots} presents the plots of these seasonal patterns with $t_0 = \frac{7}{12}$.

\begin{figure}[H]
\centering
	\includegraphics[height=5.0cm]{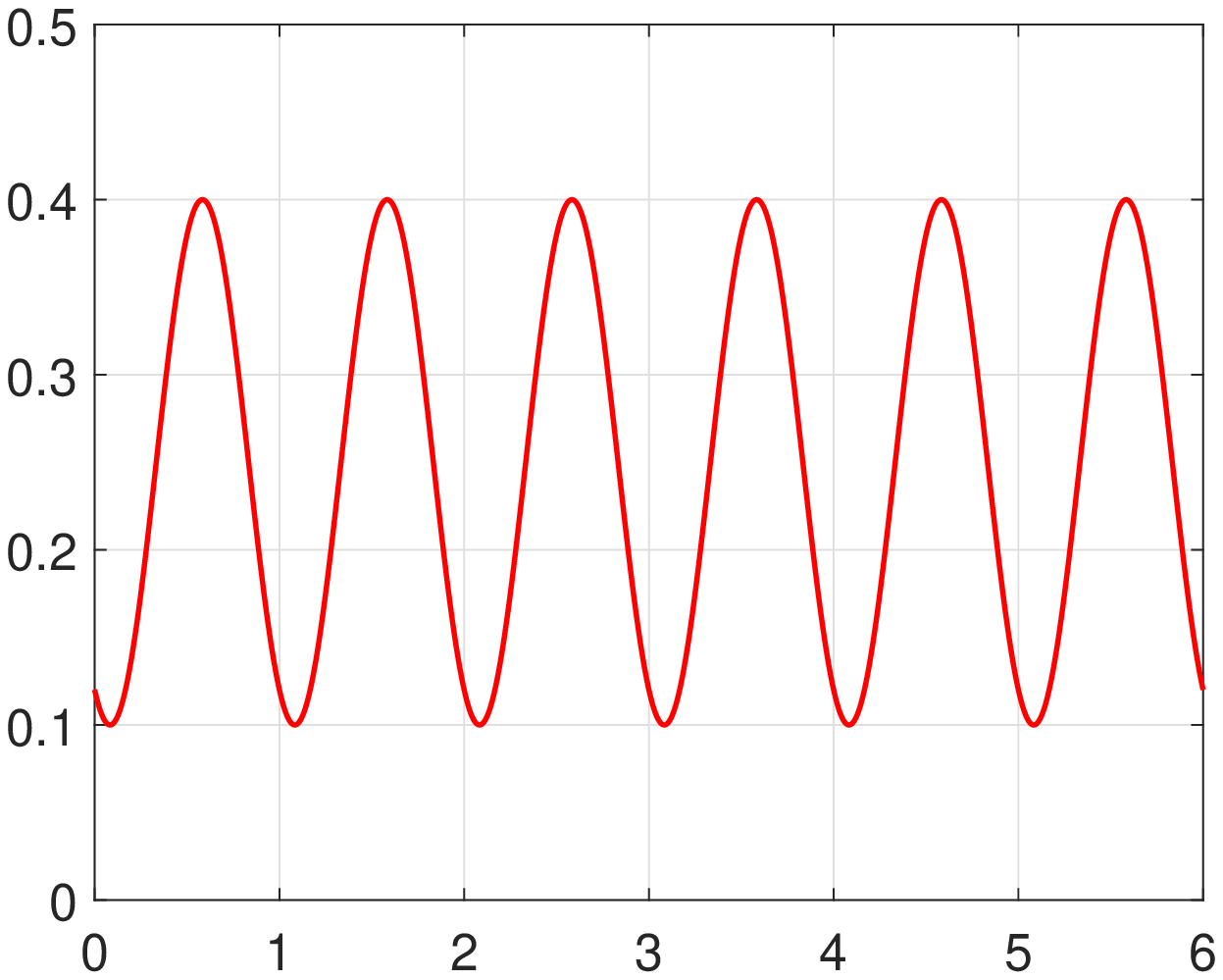} \qquad
	\includegraphics[height=5.0cm]{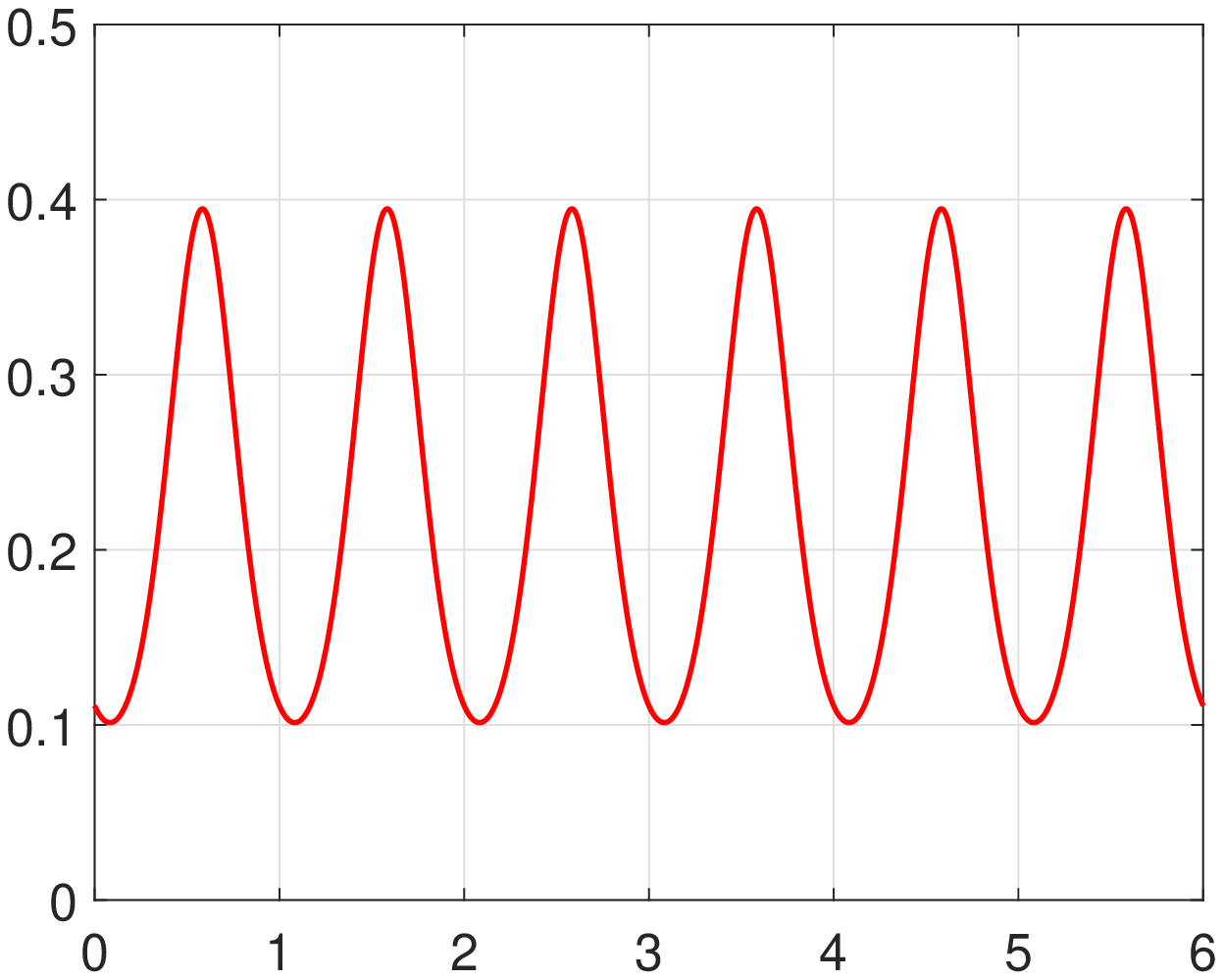}
	\vspace{0.8pc}
	
	\includegraphics[height=5.0cm]{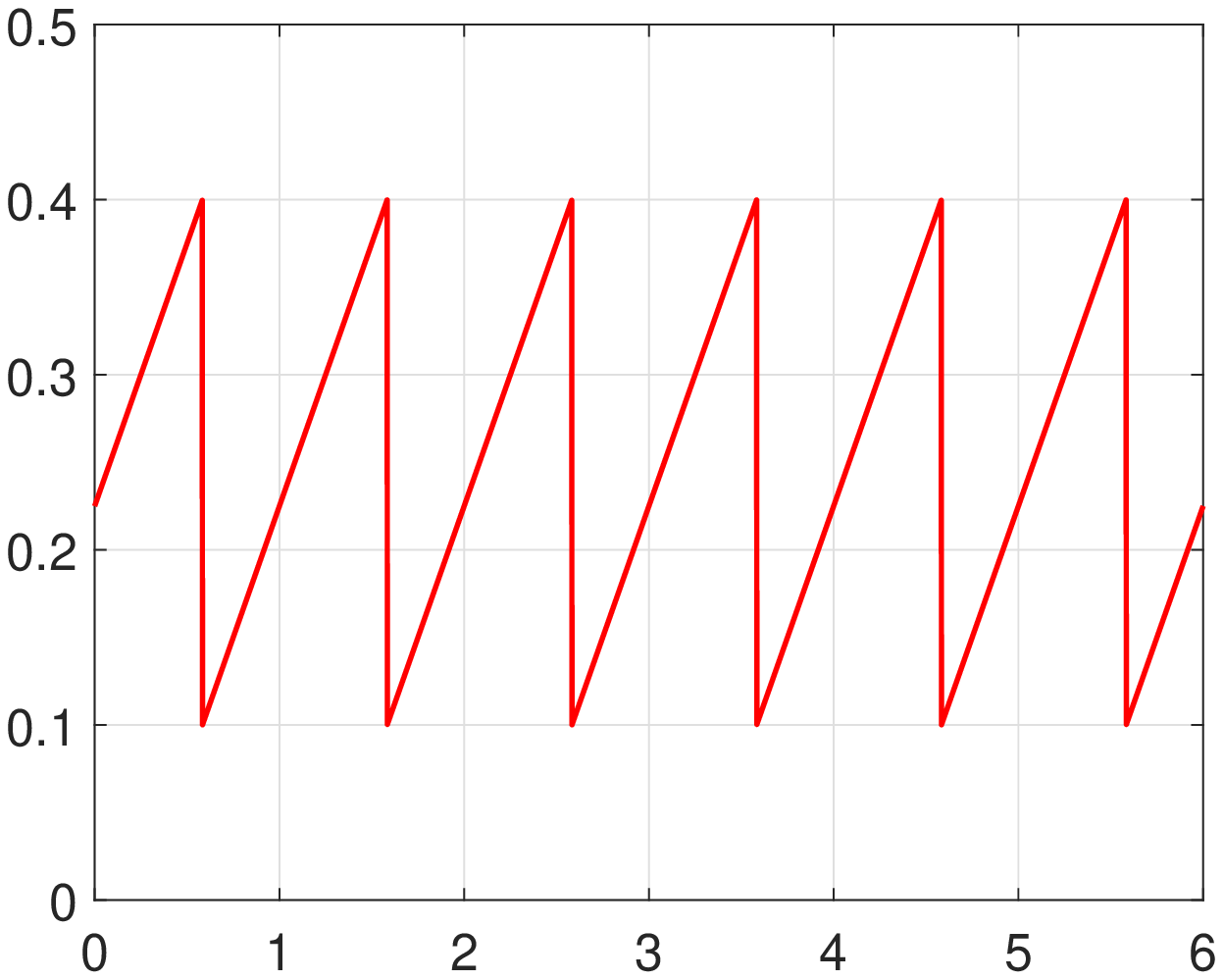} \qquad
	\includegraphics[height=5.0cm]{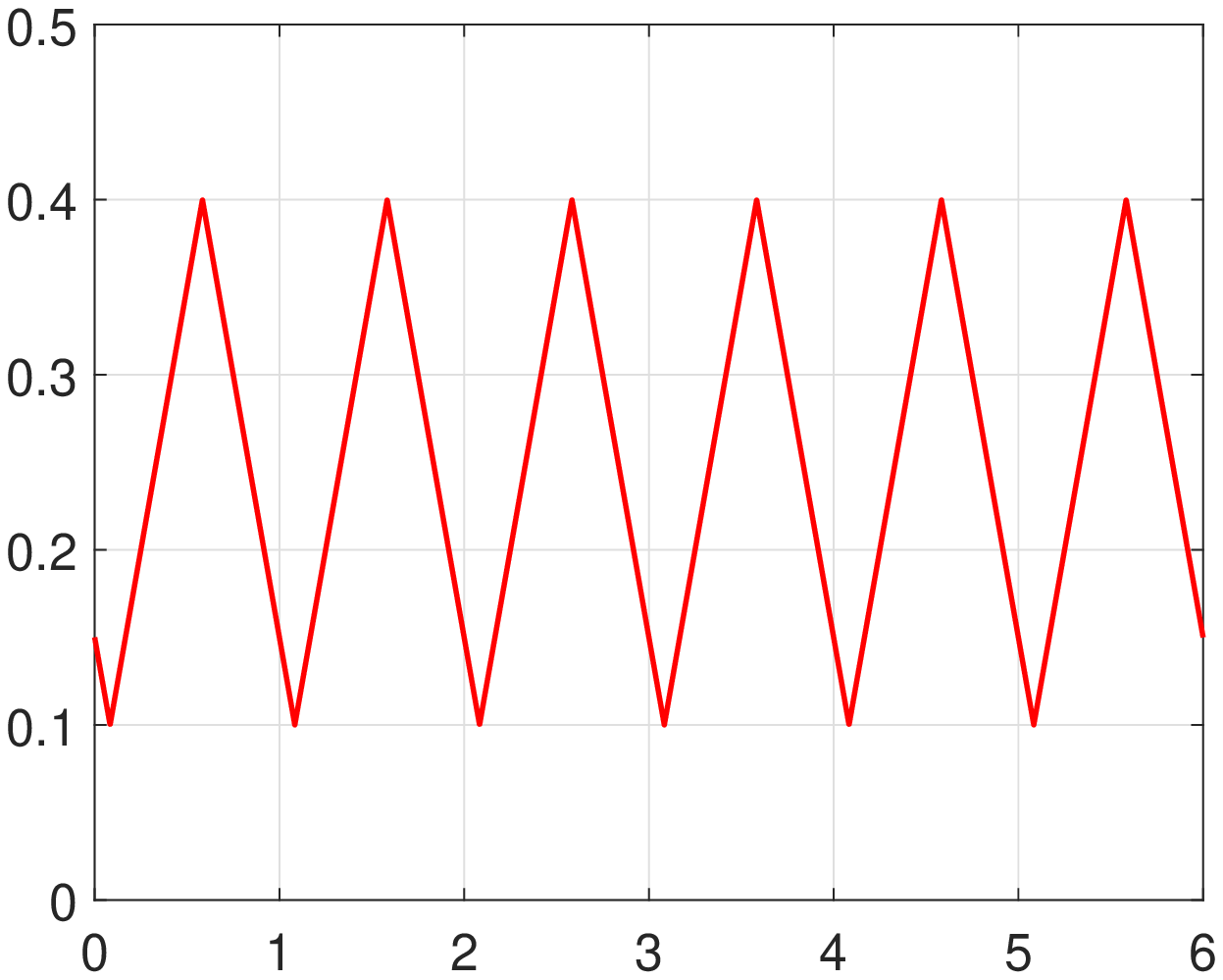}
	\vspace{0.8pc}
	
	\includegraphics[height=5.0cm]{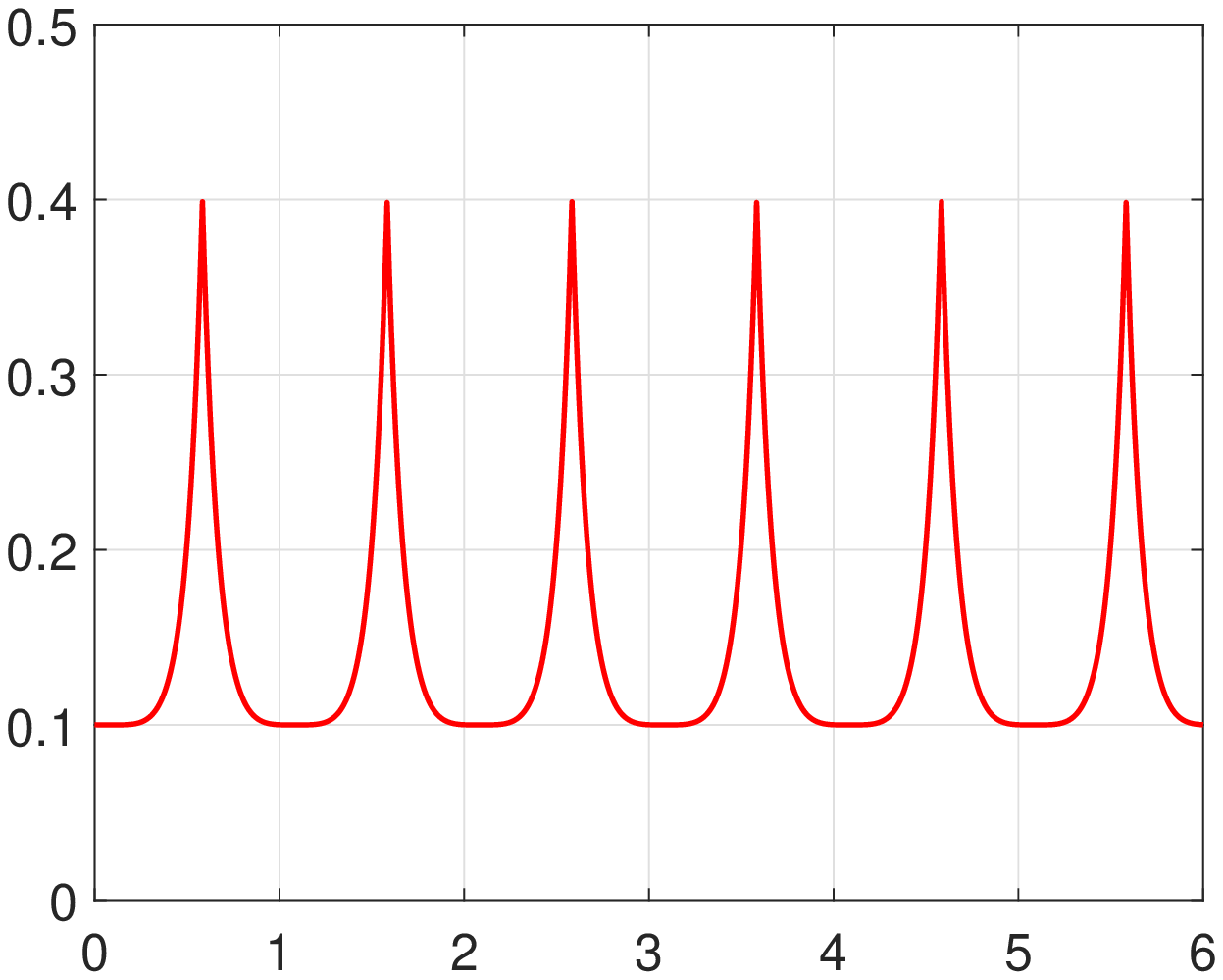}
	\caption{\label{Fig:theta_plots} Examples of seasonal patterns of $\theta(t)$ as defined in equations \eqref{pattern:sinusoidal} - \eqref{pattern:spiked}, with $t_0 = \frac{7}{12}$.
\textit{Upper left}: sinusoidal pattern ($a = 0.25, b = 0.15$). \textit{Upper right}: exponential-sinusoidal pattern ($a = 0.20, b = 0.68$).
\textit{Center left}: sawtooth pattern ($a = 0.10, b = 0.30$). \textit{Center right}: triangle pattern ($a = 0.10, b = 0.60$).
\textit{Lower}: spiked pattern ($a = 0.10, b = 0.30$).}
\end{figure}

The exponential-sinusoidal pattern \eqref{pattern:exponential-sinusoidal} is a popular choice in commodity modelling,
as it is smooth and highly tractable from a numerical point of view.
In contrast to the sinusoidal pattern \eqref{pattern:sinusoidal}, it is always strictly positive.
A more subtle advantage over \eqref{pattern:sinusoidal} could be that relatively speaking, due to the convexity of the exponential function,
it ``spends more time'' at low levels and ``less time'' at high levels, and that this property reflects the behaviour of the volatility more realistically.
The extreme case of such behaviour is the spiked pattern; however, this seasonality function is not everywhere differentiable.
The sawtooth pattern \eqref{pattern:sawtooth} is an example of the volatility gradually increasing before the harvest,
and then dropping on the harvest date before the new crop is sown or planted.
In Section \ref{s:ModelEstimation}, we carry out a statistical comparison of these seasonal volatility specifications
and see which ones are best suited for a given agricultural futures market.

In the following section, we will encounter integral transforms $\hat{\theta}$ of $\theta$ in several expressions such as the characteristic function.
For $T>0$ and $\lambda \in \mathbb{R}$, $\hat{\theta}$ is given by
\begin{equation}
\label{thetaTransform}
\hat{\theta}_T(\lambda) = \int^{T}_{0}{\theta(t)e^{\lambda t}dt}.
\end{equation}
In Appendix \ref{a:Transforms} we give closed-form expressions of $\hat{\theta}_T$ for three of the above seasonality functions.


\section{Modelling Commodity Futures with Seasonal Stochastic \\ Volatility and the Samuelson Effect}
\label{s:SeasonalStochasticVolatilityModelForCommodityFutures}

\subsection{The Model in the Risk-Neutral Measure ${\mathbb Q}$}
\label{ss:ModelInTheRiskNeutralMeasure}

We begin by giving a mathematical description of our model under the risk-neutral measure ${\mathbb Q}$.
Let $n \geq 1$ be an integer, and let $B_1^{\mathbb Q}, ..., B_{2n}^{\mathbb Q}$ be Brownian motions under ${\mathbb Q}$.
Let $T_m$ be the maturity of a given futures contract.
The futures price $F(t, T_m)$ at time $t, 0 \leq t \leq T_m$, is assumed to follow the stochastic differential equation (SDE)
\begin{equation}
\label{FuturesSDE_Q}
dF(t, T_m) = F(t, T_m) \sum_{j=1}^n e^{-\lambda_j(T_m - t)} \sqrt{v_j(t)} dB_j^{\mathbb Q}(t), \; F(0, T_m) = F_{m,0} > 0.
\end{equation}
The processes $v_j, j=1, ..., n,$ are stochastic variance processes with time-dependent seasonal mean-reversion level assumed to follow the SDE
\begin{equation}
\label{VarianceSDE_Q}
dv_j(t) = \kappa_j \left( \theta_j(t) - v_j(t) \right) dt + \sigma_j \sqrt{v_j(t)} dB_{n+j}^{\mathbb Q}(t), \; v_j(0) = v_{j,0} > 0.
\end{equation}
Various possibilities of the specification of the seasonal mean-reversion level functions $\theta_j: \mathbb{R}_0^+ \to \mathbb{R}^+$
are presented and discussed in Section \ref{ss:SeasonalityFunctions}.
Note that the initial futures curve $F(0, T_m), m = 1, 2, ...,$ is exogenous in our model and can therefore accommodate any seasonal pattern shown by the futures prices.

For the correlations, we assume
\begin{equation}
\label{FuturesVarianceCorrelations}
\langle dB_j^{\mathbb Q}(t), dB_{n+j}^{\mathbb Q}(t) \rangle = \rho_j dt, -1 < \rho_j < 1, j=1, ..., n,
\end{equation}
and that otherwise the Brownian motions $B_j^{\mathbb Q}, B_k^{\mathbb Q}, k \neq j, j + n,$ are independent of each other.
As we will see, this assumption has as a consequence that the characteristic function factors into $n$ separate expectations,
and thus keeps the model analytically tractable.

For fixed $T_m$, the futures log-price $\ln F(t, T_m)$ under ${\mathbb Q}$ follows the SDE
\begin{equation}
\label{LogFuturesSDE_Q}
d\ln F(t, T_m) = \sum_{j=1}^n \left( e^{-\lambda_j(T_m - t)} \sqrt{v_j(t)} dB_j^{\mathbb Q}(t) - \frac{1}{2} e^{-2\lambda_j(T_m - t)} v_j(t) dt \right),
\; \ln F(0, T_m) = \ln F_{m,0}.
\end{equation}
Integrating \eqref{LogFuturesSDE_Q} from time $0$ up to a time $T, T \leq T_m$, gives
\begin{equation}
\label{LogFuturesIntegratedSDE_Q}
\ln F(T, T_m) -  \ln F(0, T_m) = \sum_{j=1}^n \int_0^T e^{-\lambda_j(T_m - t)} \sqrt{v_j(t)} dB_j^{\mathbb Q}(t) - \frac{1}{2} \sum_{j=1}^n \int_0^T e^{-2\lambda_j(T_m - t)} v_j(t) dt.
\end{equation}

We define the log-return between times $0$ and $T$ of a futures contract with maturity $T_m$ as
$$
X_m(T) := \ln \left( \frac{F(T, T_m)}{F(0, T_m)} \right).
$$

\subsection{The Joint Characteristic Function}
\label{ss:JointCharacteristicFunction}

In financial applications such as option pricing, the joint characteristic function $\phi$ of two log-returns $X_1(T), X_2(T)$ plays an important role.
For $u = (u_1, u_2) \in {\mathbb C}^2$, $\phi$ is given by
\begin{equation}
\label{jointCharacteristicFunction_returns}
\phi(u)
= \phi(u; T, T_1, T_2)
= {\mathbb{E^Q}} \left[ \exp \left( i \sum_{k=1}^2 u_k X_k(T) \right) \right].
\end{equation}
The joint characteristic function $\Phi$ of the futures log-prices $\ln F(T, T_1), \ln F(T, T_2)$ is then given by
\begin{equation}
\label{jointCharacteristicFunction_prices}
\Phi(u) = \exp \left( i \sum_{k=1}^2 u_k \ln F(0, T_k) \right) \cdot \phi(u).
\end{equation}
Note that futures prices in our model are not mean-reverting,
and that the log-price $\ln F(t, T_m)$ at time $t$ and the log-return $\ln F(T, T_m) - \ln F(t, T_m)$ are independent random variables.

In the following proposition, we show how the joint characteristic function $\phi$, and therefore also the single characteristic function $\phi_1$,
is given by a system of two ordinary differential equations (ODE).
\begin{proposition}
\label{Prop:JointCharacteristicFunction}
The joint characteristic function $\phi$ at time $T \leq T_1, T_2$ for the log-returns $X_1(T), X_2(T)$ of two futures contracts with maturities $T_1, T_2$ is given by
\begin{align*}
\phi(u) &= \phi(u; T, T_1, T_2)
\\
&=
\prod_{j=1}^n
\exp \left( -i \frac{\rho_j}{\sigma_j} f_{j,1}(u,0) \left( v_j(0) + \kappa_j \hat{\theta}_{j, T} \right) \right)
\exp \left( A_j(0,T) v_j(0) + B_j(0,T) \right),
\end{align*}
where
\begin{align*}
f_{j,1}(u,t)        &= \sum_{k=1}^2 u_k e^{- \lambda_j(T_k - t)}, \quad f_{j,2}(u,t) = \sum_{k=1}^2 u_k e^{-2\lambda_j(T_k - t)},
\\
q_j(u,t)            &= i \rho_j \frac{\kappa_j - \lambda_j}{\sigma_j} f_{j,1}(u,t) - \frac{1}{2} (1 - \rho_j^2) f_{j,1}^2(u,t) - \frac{1}{2} i f_{j,2}(u,t),
\\
\hat{\theta}_{j, T} &= \int_0^T \theta_j(t) e^{\lambda_j t} dt,
\end{align*}
and the functions $A_j: (t,T) \mapsto A_j(t,T)$ and $B_j: (t,T) \mapsto B_j(t,T)$ satisfy the two differential equations
\begin{align*}
\frac{\partial A_j}{\partial t} - \kappa_j A_j + \frac{1}{2} \sigma_j^2 A_j^2 + q_j &= 0,
\\
\frac{\partial B_j}{\partial t} + \kappa_j \theta_j(t) A_j &= 0,
\end{align*}
with $A_j(T,T) = i \frac{\rho_j}{\sigma_j} f_{j,1}(u,T), \; B_j(T,T) = 0.$

The single characteristic function $\phi_1$ at time $T \leq T_1$ for the log-return $X_1(T)$ of a futures contract with maturity $T_1$ is given by setting $u_2 = 0$
in the joint characteristic function.
\end{proposition}

We prove this result in Appendix \ref{a:ProofCF}.

Note that the Riccati ODE for the functions $A$ does not depend on $\theta$.
Therefore, the same closed-form solution for $A$ as in \citet{SchneiderTavin2018} can be used.
Of course, if $\theta$ is a constant function, then the joint characteristic function given above is identical to the one given there.

Note also that the integrals $\hat{\theta}_{j, T}$ only depend on the specification of the seasonality functions $\theta_j$ and the maturity $T$, but not on $u$.
Therefore, their value can be calculated once and then stored, avoiding recalculations during repeated calls to the characteristic function.

\subsection{Option Pricing}
\label{ss:OptionPricing}


European options on futures contracts can be priced using the Fourier inversion technique as described in \citet{Heston1993} and \citet{BakshiMadan2000}.
Let $K$ denote the strike and $T$ the maturity of a European call option on a futures contract $F$ with maturity $T_m \geq T$.
The function needed for this technique is the single characteristic function $\Phi_1$ of the futures log-price $\ln F(T, T_m)$, given by
$\Phi_1(u) = e^{i u \ln F(0, T_m)} \phi_1(u)$,
with $\phi_1(u)$ obtained from Proposition \ref{Prop:JointCharacteristicFunction}.
European put options can be priced via put-call parity $C - P = e^{-rT} \left( F(0, T_1) - K \right)$.

Calendar spread options are also popular in agricultural markets.
Their payoff depends on the price difference of two futures contracts on the same commodity, but with different maturities $T_1$ and $T_2$.
\citet{CaldanaFusai2013} propose a risk-neutral valuation formula in case the joint characteristic function $\Phi$ of the two underlying futures contracts is known.
The formula is given in terms of a one-dimensional Fourier inversion, and can be used easily and efficiently with the multi-factor model presented here.

\subsection{The Model in the Physical Measure ${\mathbb P}$}
\label{ss:ModelInThePhysicalMeasure}

In order to present the model under the physical measure ${\mathbb P}$, we follow the ``completely affine'' specification of \citet{CasassusCollinDufresne2005}.
This setup is also used in \citet{DoranRonn2008}, \citet{TrolleSchwartz2009b}, and \citet{ChiarellaKangNikitopoulosTo2013}.
The market price of futures price risk and the market price of volatility risk are given by
\begin{align}
dB_j^{\mathbb P}(t)     &= dB_j^{\mathbb Q}(t) - \pi_j^F \sqrt{v_j(t)} dt, \\
dB_{n+j}^{\mathbb P}(t) &= dB_{n+j}^{\mathbb Q}(t) - \pi_j^v \sqrt{v_j(t)} dt,
\end{align}
for parameters $\pi_j^F, \pi_j^v, j = 1, \dots, n$.

For fixed $T_m$, the futures price $F(t, T_m)$ under ${\mathbb P}$ follows the SDE
\begin{align}
\label{FuturesSDE_P}
dF(t, T_m)
&= F(t, T_m) \sum_{j=1}^n e^{-\lambda_j(T_m - t)} \sqrt{v_j(t)} dB_j^{\mathbb Q}(t) \\
&= F(t, T_m) \left( \sum_{j=1}^n \pi_j^F e^{-\lambda_j(T_m - t)}v_j(t) dt + e^{-\lambda_j(T_m - t)} \sqrt{v_j(t)} dB_j^{\mathbb P}(t) \right),
\end{align}
and the variance process $v_j$ follows the SDE
\begin{align}
dv_j(t)
&= \kappa_j \left( \theta_j(t) - v_j(t) \right) dt + \sigma_j \sqrt{v_j(t)} dB_{n+j}^{\mathbb Q}(t)
\label{VarianceSDE_Q_again}
\\
&= \left( \kappa_j \left( \theta_j(t) - v_j(t) \right) + \sigma_j \pi_j^v v_j(t) \right) dt + \sigma_j \sqrt{v_j(t)} dB_{n+j}^{\mathbb P}(t).
\label{VarianceSDE_P}
\end{align}

The futures log-price $\ln F(t, T_m)$ under ${\mathbb P}$ follows the SDE
\begin{align}
\label{LogFuturesSDE_P}
d\ln F(t, T_m)
&= \sum_{j=1}^n \left(
e^{-\lambda_j(T_m - t)} \pi_j^F v_j(t) dt +
e^{-\lambda_j(T_m - t)} \sqrt{v_j(t)} dB_j^{\mathbb P}(t)
- \frac{1}{2} e^{-2\lambda_j(T_m - t)} v_j(t) dt \right),
\\
\nonumber
\ln F(0, T_m) &= \ln F_{m,0}.
\end{align}
Integrating \eqref{LogFuturesSDE_P} from time $0$ up to a time $T, T \leq T_m$, gives
\begin{align}
\label{LogFuturesIntegratedSDE_P}
\ln F(T, T_m) -  \ln F(0, T_m)
&= \sum_{j=1}^n \pi_j^F \int_0^T e^{-\lambda_j(T_m - t)} v_j(t) dt
\\
\nonumber
&\quad + \sum_{j=1}^n \int_0^T e^{-\lambda_j(T_m - t)} \sqrt{v_j(t)} dB_j^{\mathbb P}(t)
- \frac{1}{2} \sum_{j=1}^n \int_0^T e^{-2\lambda_j(T_m - t)} v_j(t) dt.
\end{align}


\section{State-Space Representation of the Model}
\label{s:StateSpaceRepresentation}

\subsection{The State Variables}
\label{ss:StateVariable}

In this section, we present our model in state-space form, so that we can run the Kalman filter and estimate the model's parameters by maximising the log-likelihood function.
We express the observed log-futures prices as functions of a set of state variables and a vector of model parameters.
Our approach is similar to that of \citet{ChiarellaKangNikitopoulosTo2013}, and our notation follows \citet{Tsay2010}.
We give the state-space representation here in general for the $n$-factor model;
however, in our empirical study in Section \ref{s:ModelEstimation} we will use the $1$-factor model.

The model with $n$ factors is parameterized by a vector of $8 n$ parameters denoted by
$\Psi = (\Psi_1, \dots, \Psi_n)$,
where
\begin{equation*}
\Psi_j = \left( \lambda_j, \kappa_j, \sigma_j, \rho_j, v_{j,0}, a_j, b_j, t_{j,0} \right), \; j = 1, \dots, n.
\end{equation*}

From equation \eqref{LogFuturesIntegratedSDE_P} we have
\begin{align}
\label{LogFuturesIntegratedSDE_P_t}
\ln F(t, T_m) -  \ln F(0, T_m)
&= \sum_{j=1}^n \pi_j^F \int_0^t e^{-\lambda_j(T_m - s)} v_j(s) ds
\\
\nonumber
&\quad + \sum_{j=1}^n \int_0^t e^{-\lambda_j(T_m - s)} \sqrt{v_j(s)} dB_j^{\mathbb P}(s)
- \frac{1}{2} \sum_{j=1}^n \int_0^t e^{-2\lambda_j(T_m - s)} v_j(s) ds.
\end{align}

We define our state variables, for $j=1,\dots, n$, as
\begin{align*}
s_{1, j}(t)
&:= \int_0^t e^{-\lambda_j(t - s)} \sqrt{v_j(s)} dB_j^{\mathbb Q}(s) \nonumber \\
&= \pi_j^F \int_0^t e^{-\lambda_j(t - s)} v_j(s) ds + \int_0^t e^{-\lambda_j(t - s)} \sqrt{v_j(s)} dB_j^{\mathbb P}(s), \\
s_{2, j}(t)
&:= \int_0^t e^{-2\lambda_j(t- s)} v_j(s) ds, \\
s_{3, j}(t)
&:= v_j(t) = s_{3, j}(0) + \kappa_j \int_0^t \theta_j(s) - s_{3, j}(s) ds + \sigma_j \int_0^t \sqrt{s_{3, j}(s)} dB_{n+j}^{\mathbb Q}(s) \nonumber \\
&= s_{3, j}(0) + \kappa_j \int_0^t \theta_j(s) - s_{3, j}(s) ds + \sigma_j \pi_j^v \int_0^t s_{3, j}(s) ds + \sigma_j \int_0^t \sqrt{s_{3, j}(s)} dB_{n+j}^{\mathbb P}(s)
\end{align*}

The dynamics of the state variables are then given by
\begin{align}
ds_{1, j}(t)
&= -\lambda_j s_{1, j}(t) dt + \sqrt{s_{3, j}(t)} dB_j^{\mathbb Q}(t) \nonumber \\
\label{ds1_P}
&= -\lambda_j s_{1, j}(t) dt + \pi_j^F s_{3, j}(t) dt + \sqrt{s_{3, j}(t)} dB_j^{\mathbb P}(t), \\
\label{ds2_P}
ds_{2, j}(t)
&= -2\lambda_j s_{2, j}(t) dt + s_{3, j}(t) dt, \\
ds_{3, j}(t)
&= \kappa_j \left( \theta_j(t) - s_{3, j}(t) \right) dt + \sigma_j \sqrt{s_{3, j}(t)} dB_{n+j}^{\mathbb Q}(t) \nonumber \\
\label{ds3_P}
&= \left( \kappa_j \left( \theta_j(t) - s_{3, j}(t) \right) + \sigma_j \pi_j^v s_{3, j}(t) \right) dt + \sigma_j \sqrt{s_{3, j}(t)} dB_{n+j}^{\mathbb P}(t).
\end{align}

We can express the futures log-price \eqref{LogFuturesIntegratedSDE_P_t} as a linear function of the first two state variables
\begin{equation*}
 \ln F(t,T_m) = \ln F(0,T_m) + \sum_{j=1}^n e^{-\lambda_j (T_m - t)} s_{1, j}(t) - \frac{1}{2} \sum_{j=1}^n e^{-2 \lambda_j (T_m - t)} s_{2, j}(t).
\end{equation*}

For notational simplicity, we focus one the $1$-factor model from here on, but it is straightforward to extend these results to the general $n$-factor case.
At time $t$, the vector of state variables is given by
\begin{equation*}
s_t =
\left(
\begin{array}{c}
s_1(t)
\\
s_2(t)
\\
s_3(t)
\end{array}
\right).
\end{equation*}

\subsection{Transition and Measurement Equations}
\label{ss:TransitionAndMeasurementEquations}

The transition (or state) equation in discrete time is given by the Euler discretisation scheme of equations \eqref{ds1_P}, \eqref{ds2_P} and \eqref{ds3_P} as
\begin{equation}
\label{TransitionEquation}
s_{t + \triangle t} = d_t + T_t s_t + R_t \eta_t,
\end{equation}
with $\eta_t \sim \mathcal{N}(0, Q_t)$.
where
\begin{equation*}
d_t =
\left(
\begin{array}{c}
0
\\
0
\\
\kappa \theta(t) \triangle t
\end{array}
\right)
,
\quad
T_t =
\left(
\begin{array}{ccc}
1 - \lambda \triangle t & 0 & \pi^F \triangle t
\\
0 & 1 - 2 \lambda \triangle t & \triangle t
\\
0 & 0 & 1 - (\kappa - \sigma \pi^v) \triangle t
\end{array}
\right)
,
\end{equation*}
\begin{equation*}
R_t =
\sqrt{s_3(t)}
\left(
\begin{array}{cc}
1 & 0
\\
0 & 0
\\
0 & \sigma
\end{array}
\right),
\quad
Q_t =
\left(
\begin{array}{cc}
\triangle t & \rho \triangle t
\\
\rho \triangle t & \triangle t
\end{array}
\right)
.
\end{equation*}
Note that the vector $d_t$ depends on the function $\theta$ and is therefore time-dependent,
and the matrix $R_t$ depends on the current volatility $\sqrt{s_3(t)} = \sqrt{v(t)}$ and is state-dependent.

The measurement (or observation) equation provides the link between the observable quantities and the state variables.
The measurement equation for futures market data of size $k$ is
\begin{equation}
\label{MeasurementEquation}
y_t = c_t + Z_t s_t + e_t,
\end{equation}
with $e_t \sim \mathcal{N}(0, H_t)$, where $H_t$ is the covariance matrix of the measurement error $e_t$,
\begin{equation}
\label{FuturesLogPrices}
y_t =
\left(
\begin{array}{c}
\ln F(t, T_{m_1})
\\
\vdots
\\
\ln F(t, T_{m_k})
\end{array}
\right),
\quad
c_t =
\left(
\begin{array}{c}
\ln F(0, T_{m_1})
\\
\vdots
\\
\ln F(0, T_{m_k})
\end{array}
\right),
\end{equation}
\begin{equation}
\label{MeasurementMatrix}
Z_t =
\left(
\begin{array}{ccc}
e^{-\lambda (T_{m_1} - t)} & -\frac{1}{2} e^{-2 \lambda (T_{m_1} - t)} & 0
\\
\vdots & \vdots & \vdots
\\
e^{-\lambda (T_{m_k} - t)} & -\frac{1}{2} e^{-2 \lambda (T_{m_k} - t)} & 0
\end{array}
\right).
\end{equation}



Alternatively, instead of working with the futures log-prices $y_t$ of equation \eqref{FuturesLogPrices},
we can work with futures log-returns
\begin{equation}
\label{FuturesLogReturns}
\hat{y}_t =
\left(
\begin{array}{c}
\ln F(t, T_{m_1}) - \ln F(t - \triangle t, T_{m_1})
\\
\vdots
\\
\ln F(t, T_{m_k}) - \ln F(t - \triangle t, T_{m_k})
\end{array}
\right).
\end{equation}
In this case, the first two $1$'s on the diagonal of the transition matrix $T_t$ in \eqref{TransitionEquation} need to be removed,
and the vector $c_t$ in \eqref{MeasurementEquation} is set to zero for all $t$.

\subsection{The Log-Likelihood Function}
\label{ss:LogLikelihoodFunction}

The conditional probability density function is used to write the joint density function as
\begin{equation*}
\mathcal{L}(y; \Psi) = \prod_{t=1}^T p(y_t | F_{t-1}),
\end{equation*}
where $p(y_t | F_{t-1})$ denotes the distribution of $y_t$ given $F_{t-1} := \{ y_1, ..., y_{t-1} \}$.
The log-likelihood function $\ln \mathcal{L}$ can be written in terms of the filter variables as
\begin{equation*}
\ln \mathcal{L}(y; \Psi) = -\frac{k T}{2} \ln 2 \pi - \frac{1}{2} \sum_{t=1}^T \ln |V_t| - \frac{1}{2} \sum_{t=1}^T v_t' V_t^{-1} v_t,
\end{equation*}
where $T$ is the length of the time series, $k$ the number of observed prices or returns at each time-step,
$v_t := y_t - y_{t|t-1}$ the $1$-step ahead forecast error of $y_t$ given $F_{t-1}$,
and $V_t := {\mathrm{Var}}(v_t | F_{t-1}) = {\mathrm{Var}}(v_t)$ is the covariance matrix of the error $v_t$.


\section{Model Estimation with the Kalman Filter}
\label{s:ModelEstimation}

\subsection{Description of the Datasets}
\label{ss:Data}

In this section we consider five agricultural commodities: corn, cotton, soybeans, sugar and wheat.
For each commodity we work with a data set spanning ten years of daily futures prices, from 1 November 2007 to 13 November 2017.
The contracts are
corn futures traded on CBOT,
cotton No. 2 futures traded on ICE,
soybean futures traded on CBOT,
sugar No. 11 futures traded on ICE,
and Chicago SRW wheat futures traded on the CBOT.
These contracts are all for physical delivery.
For corn, cotton and wheat we have ten contracts with expiries from 2 months to 2 years.
For soybeans we have thirteen contracts with expiries from 2 months to 2 years.
For sugar we have seven contracts with expiries from 2 months to 1.75 years.
These data sets have all been obtained from Thomson Reuters Eikon.
The codes to access the data (RIC) are
$Cc1, \dots, Cc10$ for corn,
$CTc1, \dots, CTc10$ for cotton,
$Sc1, \dots, Sc13$ for soybeans,
$YOc1, \dots, YOc7$ for sugar ,
and $Wc1, \dots, Wc10$ for wheat.
The corresponding traded contracts are for the following calendar months:
corn (C) and wheat (W): MAR, MAY, JUL, SEP, DEC;
cotton (CT): MAR, MAY, JUL, OCT, DEC;
soybean (S): JAN, MAR, MAY, JUL, AUG, SEP, NOV;
sugar (YO): MAR, MAY, JUL, OCT.

Table \ref{tab:data_description} summarizes the characteristics of our data sets, including,
for each commodity, the minimum, maximum and average prices and the average volatility of the contracts.
Figure \ref{Fig:price_return_ts_plots} plots, for each commodity, the time series of prices for the contracts with the shortest and longest expiries
as well as the corresponding log-returns.
Our data sets start with a period of high volatility, corresponding approximately to the first two years,
and each market has moved between contango and backwardation during the whole ten years.
It can be seen that the volatility of the futures with the longest expiry is lower than that of the futures with the shortest expiry,
which is in line with the Samuelson effect.
The average volatilities computed over the ten years of data are rather similar for each commodity and roughly equal to 25\%.

In Table \ref{tab:evidence_samuelson} we report the average volatility of each futures series in our sample.
As the time to expiry increases, the average volatility of the futures decreases.
This behaviour can be clearly observed for all five commodities, and nicely displays the Samuelson effect found in our datasets.
It is also a way to gauge its magnitude:
for corn and wheat, the volatility of the last contract is about eight to ten volatility points below the volatility of the front contract;
for soybeans, this difference is only around six points;
however, for cotton it almost reaches twelve points, and for sugar almost fifteen.
Therefore, from this first analysis, we can expect the values found for the volatility damping parameter $\lambda$
when estimating our models to be lower for soybeans and higher for cotton and sugar.

In Table \ref{tab:evidence_seasonality} we report the average volatility per calendar month for the first futures contract for each commodity.
In addition, we identify the two calendar months with the highest volatilities.
For corn, cotton and soybeans, we can see that there are two consecutive months during which the volatility is higher:
June and July for corn and cotton, and August and September for soybeans.
During these two months, the volatility is approximately five points higher than the average volatility across the other calendar months,
which provides evidence for the existence of a seasonal behaviour of the volatility for the first contract.
For sugar and wheat, the results are less clear, as the two calendar months with the highest volatilities are not consecutive:
March and October for sugar, and March and June for wheat.
We therefore expect the magnitude of the seasonal component in the model to be greater for corn, cotton and soybeans than for sugar and wheat.
In terms of when the volatility peaks, we expect $t_0$ to be slightly sooner for corn and cotton, and later for soybeans.
For sugar and wheat, it is less easy to draw an \textit{a priori} conclusion.

\begin{table}[H]
  \centering
\begin{tabular}{ccccccccc}
\toprule
Name & dates & start date & end date & futures & min price & max price & avg price & avg vol. \\
\midrule
$Corn$ & $2529$ & $01/11/07$ & $13/11/17$ & $10$ & $293.50$ & $838.75$ & $485.16$ & $25.07$\% \\
$Cotton$ & $2528$ & $01/11/07$ & $13/11/17$ & $10$ & $39.14$ & $215.15$ & $77.16$ & $23.24$\% \\
$Soybeans$ & $2529$ & $01/11/07$ & $13/11/17$ & $13$ & $783.50$ & $1771.00$ & $1120.96$ & $21.53$\% \\
$Sugar$ & $2528$ & $01/11/07$ & $13/11/17$ & $7$ & $0.10$ & $0.35$ & $0.18$ & $27.43$\% \\
$Wheat$ & $2529$ & $01/11/07$ & $13/11/17$ & $10$ & $361.00$ & $1282.50$ & $656.46$ & $27.55$\% \\
\bottomrule
\end{tabular}
\caption{Description of the datasets (Name, number of dates, start and end, number of futures, min/max and average prices, average volatility).
Min/max and average prices, as well as average volatility are taken across expiries for each commodity.}
\label{tab:data_description}
\end{table}

\begin{table}[H]
  \centering
\begin{tabular}{cccccc}
\toprule
Contract & Corn & Cotton & Soybeans & Sugar & Wheat \\
\midrule
$c1$ & $29.54$\% & $30.33$\% & $25.51$\% & $35.76$\% & $33.44$\% \\
$c2$ & $28.57$\% & $27.17$\% & $23.72$\% & $31.95$\% & $31.75$\% \\
$c3$ & $27.63$\% & $26.78$\% & $23.42$\% & $28.75$\% & $30.22$\% \\
$c4$ & $26.56$\% & $24.88$\% & $22.97$\% & $26.25$\% & $28.77$\% \\
$c5$ & $25.47$\% & $23.05$\% & $22.43$\% & $24.68$\% & $27.47$\% \\
$c6$ & $24.48$\% & $21.66$\% & $21.65$\% & $23.18$\% & $26.48$\% \\
$c7$ & $23.14$\% & $20.53$\% & $20.96$\% & $21.42$\% & $25.55$\% \\
$c8$ & $22.25$\% & $20.01$\% & $20.51$\% & $-$ & $24.58$\% \\
$c9$ & $21.64$\% & $19.21$\% & $20.19$\% & $-$ & $24.02$\% \\
$c10$ & $21.45$\% & $18.75$\% & $19.85$\% & $-$ & $23.25$\% \\
$c11$ & $-$ & $-$ & $19.73$\% & $-$ & $-$ \\
$c12$ & $-$ & $-$ & $19.51$\% & $-$ & $-$ \\
$c13$ & $-$ & $-$ & $19.32$\% & $-$ & $-$ \\
\bottomrule
\end{tabular}
\caption{Average volatilities for each futures contract of our dataset.}
\label{tab:evidence_samuelson}
\end{table}

\begin{table}[H]
  \centering
\begin{tabular}{cccccc}
\toprule
Calendar month & Corn & Cotton & Soybeans & Sugar & Wheat \\
\midrule
Januay	 & $26.52$ \%	 & $23.96$ \%	 & $22.91$ \%	 & $32.39$ \%	 & $28.97$ \% \\
February	 & $19.99$ \%	 & $24.54$ \%	 & $18.67$ \%	 & $34.14$ \%	 & $31.83$ \% \\
March	 & $27.29$ \%	 & $28.35$ \%	 & $21.99$ \%	 & \fbox{$37.92$ \%}	 & \fbox{$35.87$ \%} \\
April	 & $26.70$ \%	 & $26.02$ \%	 & $19.53$ \%	 & $32.72$ \%	 & $31.92$ \% \\
May	 & $25.83$ \%	 & $28.41$ \%	 & $20.26$ \%	 & $32.57$ \%	 & $29.66$ \% \\
June	 & \fbox{$33.71$ \%}	 & \fbox{$32.82$ \%}	 & $21.04$ \%	 & $34.29$ \%	 & \fbox{$35.33$ \%} \\
July	 & \fbox{$32.62$ \%}	& \fbox{$33.99$ \%} & $27.03$ \%	 & $35.73$ \%	 & $31.31$ \% \\
August	 & $29.07$ \%	 & $25.94$ \%	 & \fbox{$28.01$ \%}	 & $30.41$ \%	 & $34.46$ \% \\
Setpember	 & $28.93$ \%	 & $26.57$ \%	 & \fbox{$28.28$ \%}	 & $36.15$ \%	 & $29.50$ \% \\
October	 & $30.05$ \%	 & $27.82$ \%	 & $25.37$ \%	 & \fbox{$38.25$ \%}	 & $31.32$ \% \\
November	 & $25.78$ \%	 & $29.71$ \%	 & $22.74$ \%	 & $30.65$ \%	 & $27.59$ \% \\
December	 & $24.37$ \%	 & $24.59$ \%	 & $19.92$ \%	 & $30.43$ \%	 & $27.29$ \% \\
\bottomrule
\end{tabular}
\caption{Average volatilities, per calendar month, for the first futures series (\textit{c1}) in our dataset.
For each commodity, the boxes indicate the two calendar months with the highest volatilities.}
\label{tab:evidence_seasonality}
\end{table}

\subsection{Maximum Likelihood Estimation}
\label{ss:MaximumLikelihoodEstimation}

We use these data sets to estimate our model with the five seasonality specifications presented in Section \ref{s:SeasonalStochasticVolatility},
as well as the non-seasonal version of the model.
In the following, we refer to these six specifications of seasonality as \textit{the models}.
The methodology we use is a maximum likelihood estimation using the Kalman filter algorithm to obtain the hidden states
as formalized in Section \ref{s:StateSpaceRepresentation}.
Since we do not include option prices in our sample, we do not attempt to estimate the market price of volatility risk parameter $\pi^v$, and set it equal to zero.
The maximization algorithm we use is simulated annealing as described in \citet{GoffeFerrierRogers1994}.
Our implementation is done in C++ and Matlab.
The estimations are performed with the log-return time series \eqref{FuturesLogReturns}.
On a roll date the return must be calculated w.r.t. the same contract, i.e. ``diagonally'' between two adjacent series.
For the new ``last'' contract rolling into our sample, we cannot do this and instead set this return to zero.

For each commodity, we have estimated six models, one of which being non-seasonal.
Table \ref{tab:summary_results} gathers the obtained results.
For each commodity and model we provide the log-likelihood, the AIC (Akaike Information Criterion), BIC (Bayesian Information Criterion),
the value taken by the statistic $D1$ of the first likelihood ratio test (seasonal model versus non-seasonal, see below)
and its $p$-value (non-seasonal model being the null hypothesis).
For these models, we provide the AIC difference (denoted by $\Delta_{aic}$) and the Akaike weight (denoted by $\omega_i$).

In Table \ref{tab:summary_results}, we also provide the log-likelihood obtained
after estimating the nested versions of the considered models obtained when setting $\lambda=0$ (nested models without the Samuelson effect),
as well as the value taken by the statistic $D2$ of the second likelihood ratio test (model with the Samuelson effect versus model without, see below)
and its $p$-value (model without the Samuelson effect being the null hypothesis).

In order to rank models according to their performance we rely on the approaches presented in \citet{BurnhamAnderson2002}.
The AIC and $\Delta_{aic}$ are used to rank the estimated models for each commodity.
The idea is that the best model provides the smallest AIC and that $\Delta_{aic}$ allows one to gauge how far an alternative model is from the best model.
Another indicator we compute is the Akaike weight $\omega_i$ that can be interpreted as the weight of evidence in favor of model $i$ given the considered data and set of models.
In a Bayesian framework, $\omega_i$ can be interpreted as the probability that model $i$ is the best model (in the relative entropy sense) in the given set of models.
We refer to \citet{BurnhamAnderson2002} for further details and properties about AIC differences and Akaike weights.

Table \ref{tab:models_ranking} presents, for each commodity, the ranking obtained for the estimated models.
It can be noted that the exp-sinusoidal is always ranked first or second.
It suggests that we may choose this seasonality specification if we were to keep only one out of the five considered
(or six if we also count the non-seasonal specification).
It can also be noted that the non-seasonal model is always ranked sixth, i.e. last, in line with the conjectured need for a specific modelling of the seasonality.

\subsection{Testing for Seasonal Volatility}
\label{ss:TestingForSeasonalVolatility}

As the seasonal models have two additional parameters compared to the non-seasonal one,
it is important to check that the obtained increase of likelihood is significant and not due to over-fitting.
We perform a likelihood ratio test in order to investigate further this need for a specific modelling of the seasonality.
This test relies on the statistic $D1$ defined as twice the difference of log-likelihoods between the alternative model and the constrained model (non-seasonal in our case).
It is defined as
\begin{equation}
D1 = 2\left( \ln \mathcal{L}(\text{seasonal model}) - \ln \mathcal{L}(\text{non-seasonal model}) \right).
\end{equation}

$D1$ follows a $\chi^2$ distribution with $2$ degrees of freedom.
Intuitively, it takes values close to zero when the additional parameters are not useful,
and large values when the additional parameters are significantly useful to describe the data.

Figure \ref{Fig:lrtests_plots} reports, for each commodity, the values taken by the likelihood ratio test statistics $D1$ computed for each seasonal model.
We have also reported the quantiles of the $\chi^2$ distribution with $2$ degrees of freedom at $99\%$, $99.9\%$ and $99.99\%$ levels.
In summary, a majority of models pass the test at the $99.99\%$ significance level (21 out of 25), and all of them pass it at the $99\%$ significance level.
These results again confirm the need for a specific modelling of the seasonality in the considered markets.

As an aside, note that wheat is the commodity for which we obtained the smallest $D1$ values.
As conjectured in Section \ref{ss:Data}, this indicates that the seasonality of volatility has a lower magnitude in this market when compared to the others.
Nevertheless, the statistical test still confirms the usefulness of a seasonal model for wheat.

\subsection{Testing for the Samuelson Effect}
\label{ss:TestingForSamuelsonEffect}

We conduct a similar likelihood ratio test for the parameter $\lambda$, which determines the strength of the Samuelson effect.
The idea is to check that this parameter is significant in the seasonal specification.
This test relies on the statistic $D2$ defined as twice the difference of log-likelihoods between the alternative model and the constrained model
obtained by setting $\lambda=0$.
It is defined as
\begin{equation}
D2 = 2\left( \ln \mathcal{L}(\text{model with Samuelson effect}) - \ln \mathcal{L}(\text{model with }\lambda=0) \right).
\end{equation}

$D2$ follows a $\chi^2$ distribution with $1$ degree of freedom.
Again, it takes values close to zero when the damping factor controlling the Samuelson effect is not useful,
and large values when the parameter $\lambda$ is significantly useful to describe the data.

Values taken by $D2$, as well as the corresponding p-values, are reported in Table \ref{tab:summary_results}.
For the considered commodities, the parameter $\lambda$ is always found to be significant in the specification of the model (at the $99.99\%$ level).
These results confirm the importance played by the Samuelson effect in order to properly describe the behavior of volatility in these markets.

\subsection{The Filtered Variance}
\label{ss:FilteredVariance}

Figure \ref{Fig:sv3ts_plots} plots, for each commodity, the time series of the third state variable $v(t)$ as well as the seasonal component $\theta(t)$,
both obtained with the best performing model as identified in Table \ref{tab:models_ranking}.
For corn, cotton and soybeans, the seasonal behavior of $v(t)$ can easily be identified on the corresponding graphs.
For sugar and wheat, the seasonal behavior is slightly more difficult to identify.

\subsection{The Estimated Parameters}
\label{ss:EstimatedParameters}

Tables \ref{tab:estimated_parameters_corn},
\ref{tab:estimated_parameters_cotton},
\ref{tab:estimated_parameters_soybeans},
\ref{tab:estimated_parameters_sugar} and
\ref{tab:estimated_parameters_wheat}
gather, for the six considered models, the estimated parameters for corn, cotton, soybeans, sugar and wheat, respectively.
For each parameter and model we provide the estimated value as well as the corresponding value transformed in $\mathbb{R}$
and the standard error of the estimation (standard deviation with respect to the parameter in $\mathbb{R}$).

For each commodity, the estimation of the parameter $\lambda$ is rather stable across models.
This parameter drives the magnitude of the Samuelson effect in our specification.
We find it to be higher for sugar (approx. $0.28$) and lower for soybeans (approx. $0.13$).
For corn, cotton, and wheat, we find is to be around $0.21$.
These differences in terms of the value of $\lambda$ are in line with the differences in terms of the magnitude of the Samuelson effect
conjectured in Section \ref{ss:Data} based on the results shown in Table \ref{tab:evidence_samuelson}.

The estimations we performed also yield estimates for the market prices of futures price risk.
We find this market price of risk to be clearly positive for corn, soybeans, sugar and wheat, and close to zero for cotton.

For the seasonal component, the estimates of the parameter $t_0$ are found to be higher for soybeans than for corn and cotton.
For corn and cotton these estimates have similar values.
These two remarks are in line with the comments made in Section \ref{ss:Data} based on results in Table \ref{tab:evidence_seasonality}.

\begin{figure}[H]
\centering
	\includegraphics[height=3.7cm]{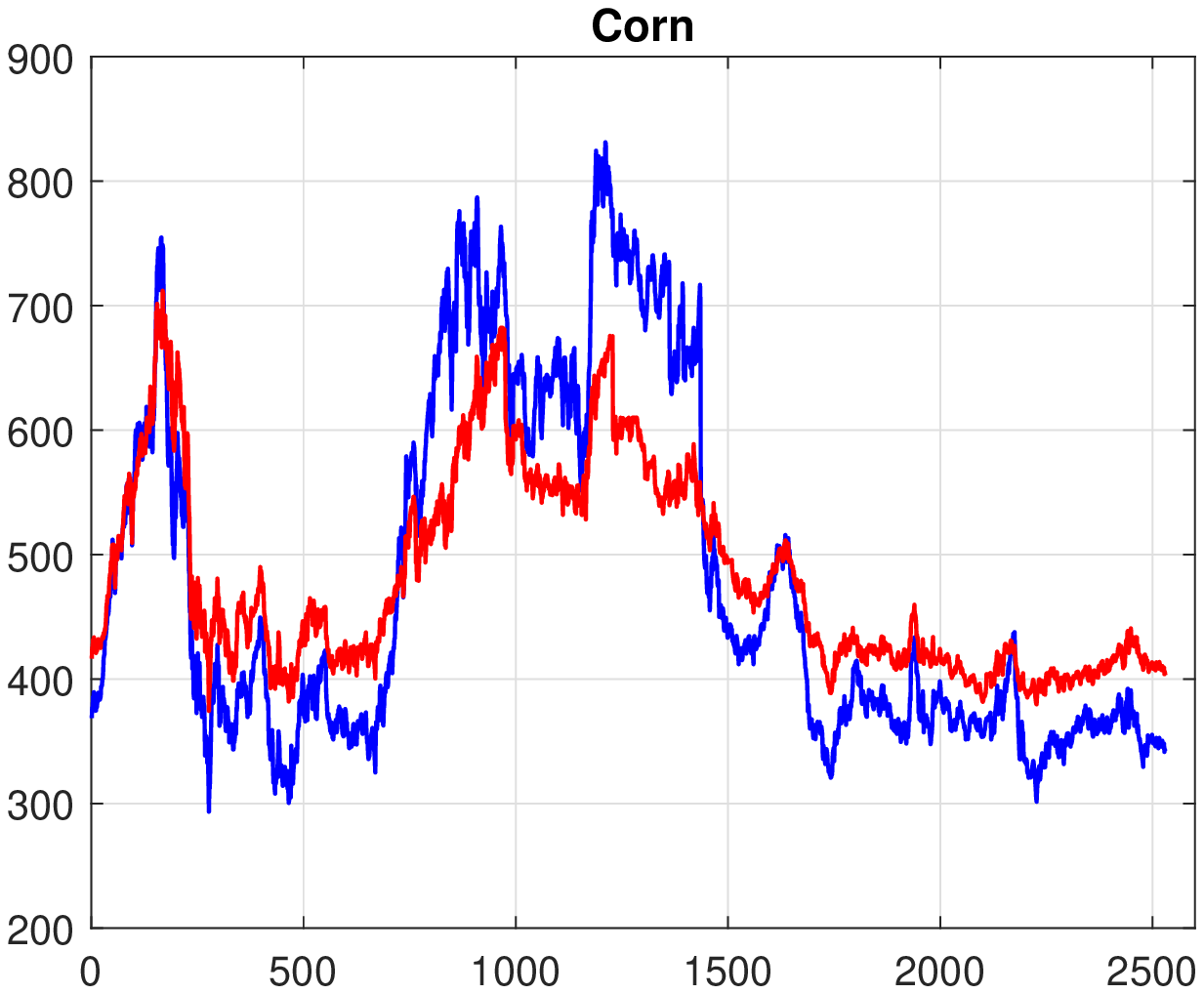} \qquad
	\includegraphics[height=3.5cm]{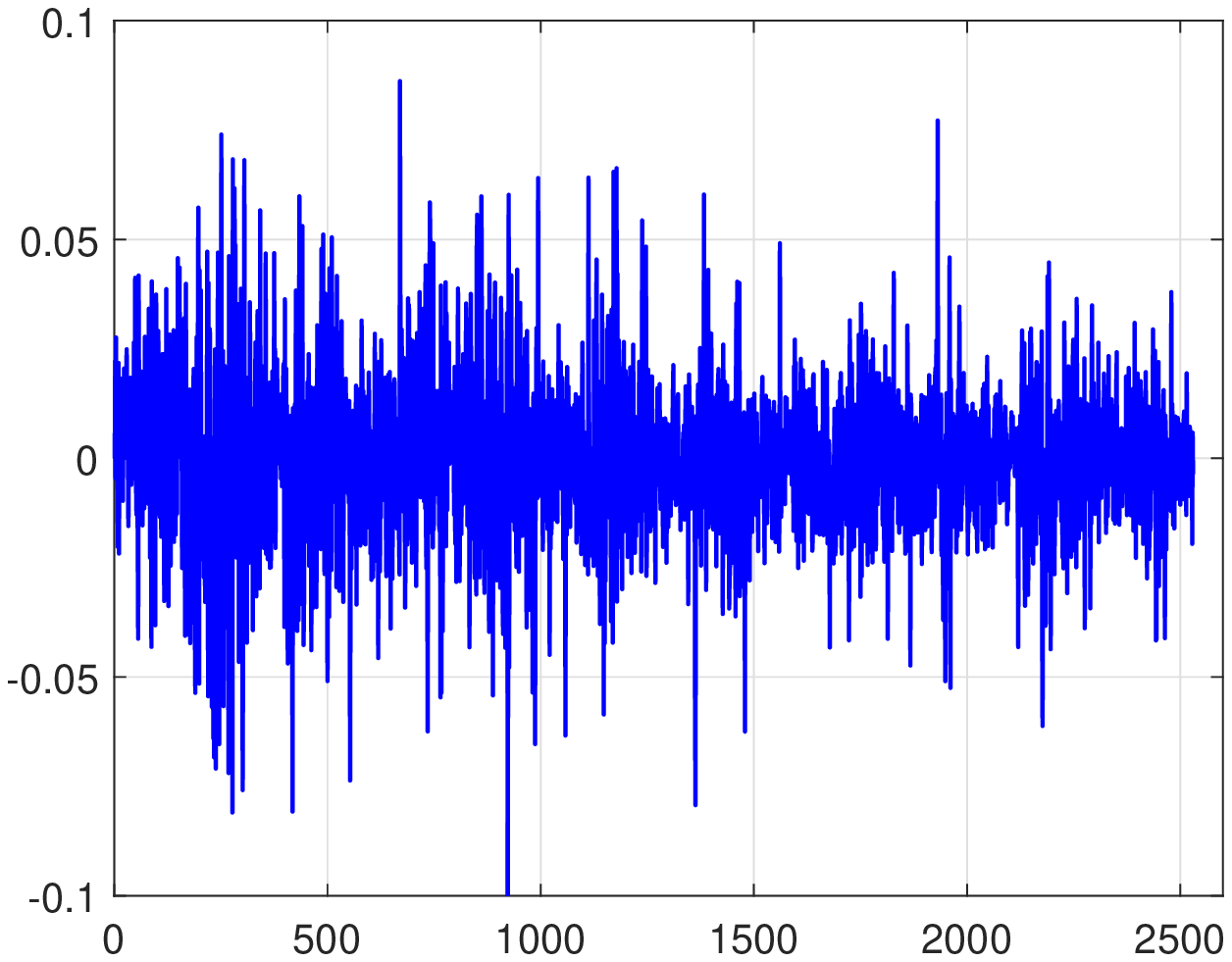} \qquad
	\includegraphics[height=3.5cm]{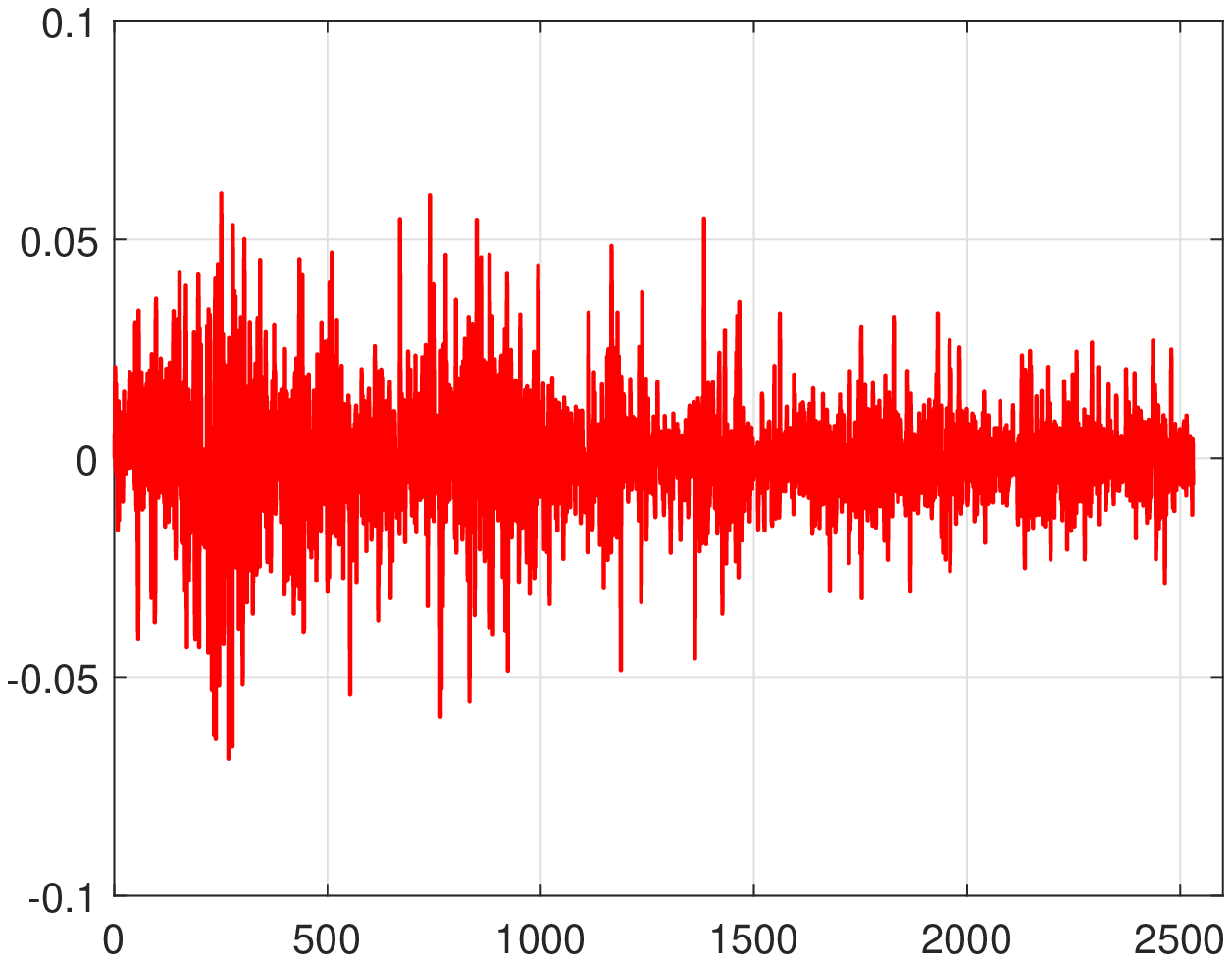}
	\vspace{0.2pc}
	
	\includegraphics[height=3.7cm]{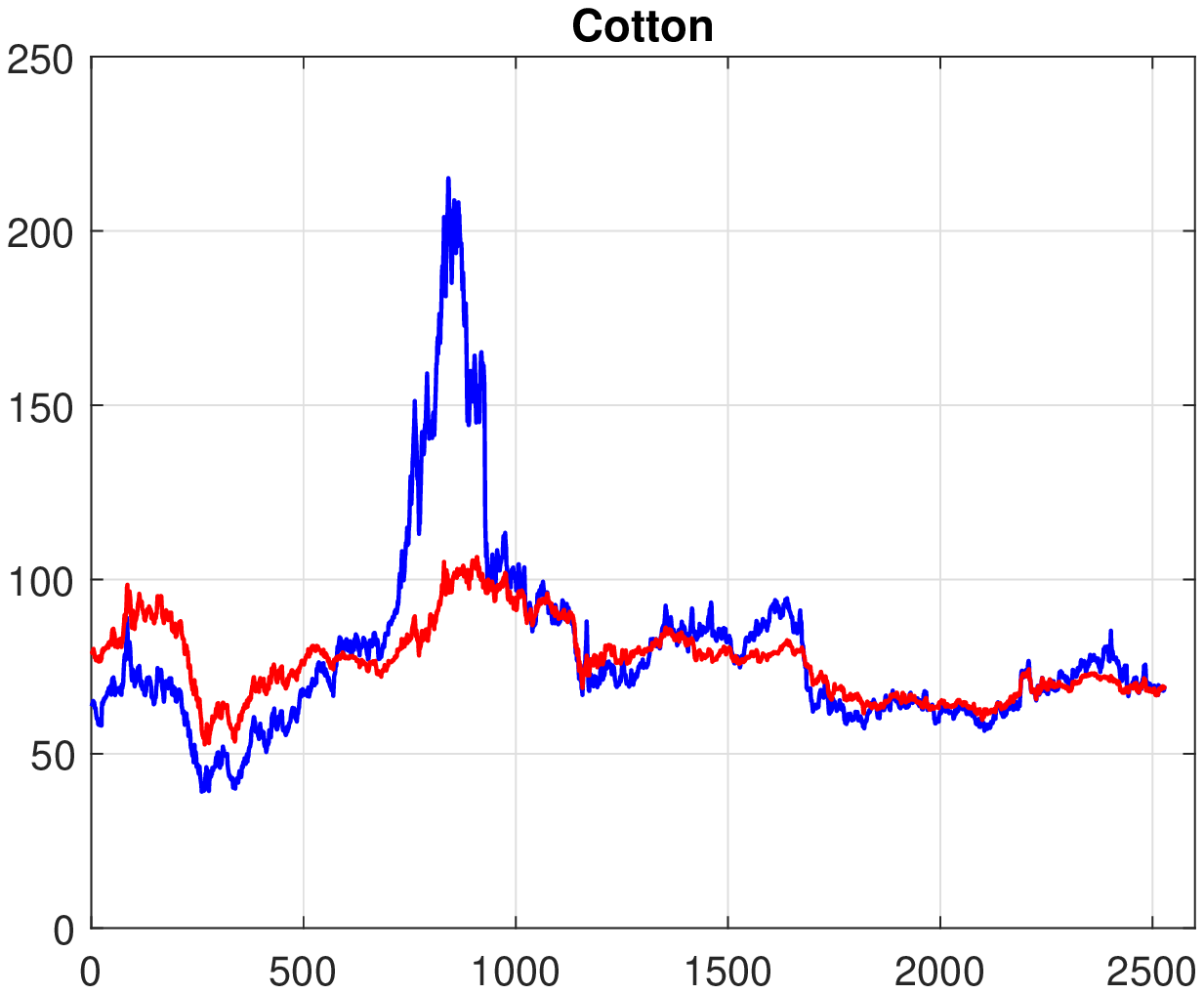} \qquad
	\includegraphics[height=3.5cm]{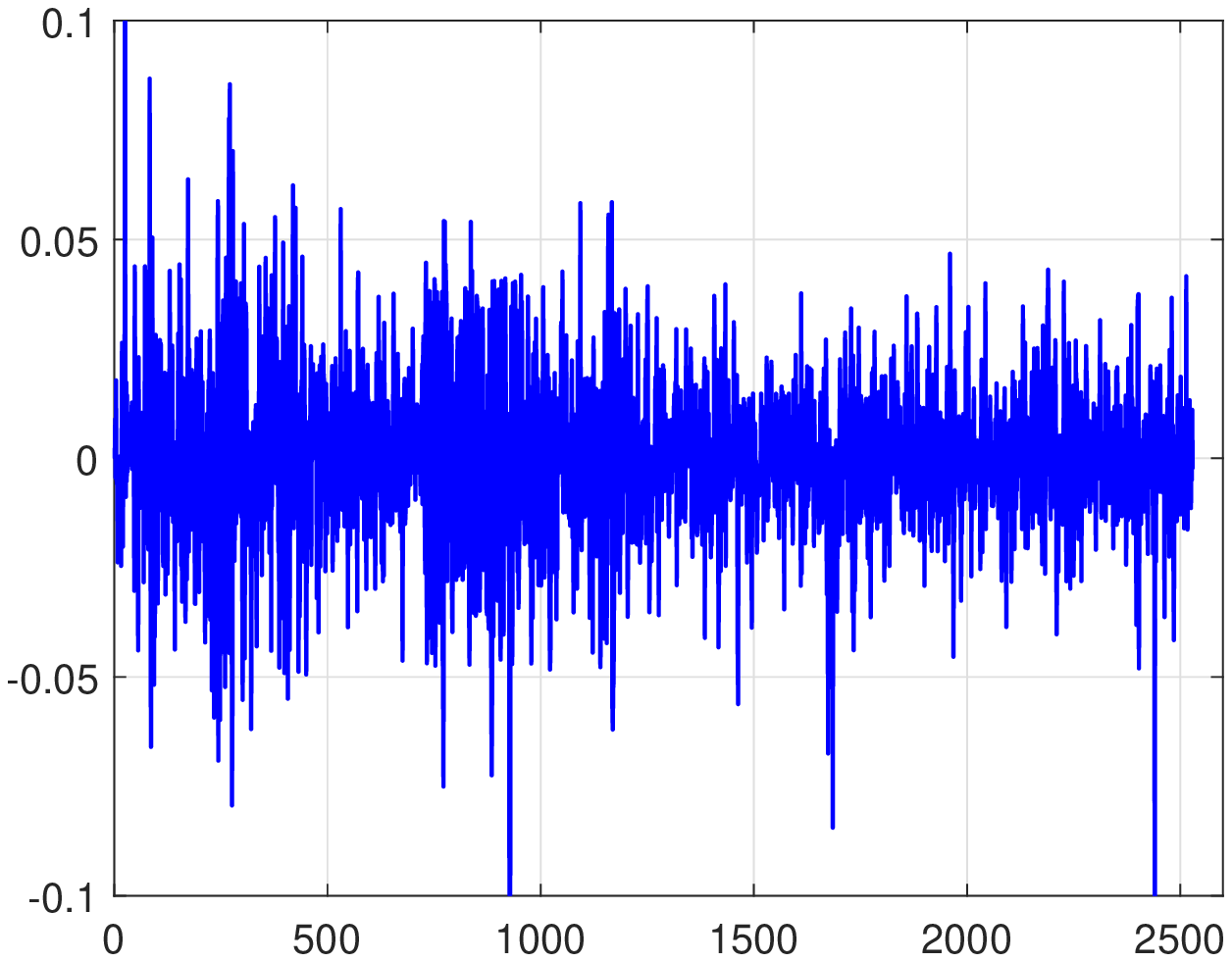} \qquad
	\includegraphics[height=3.5cm]{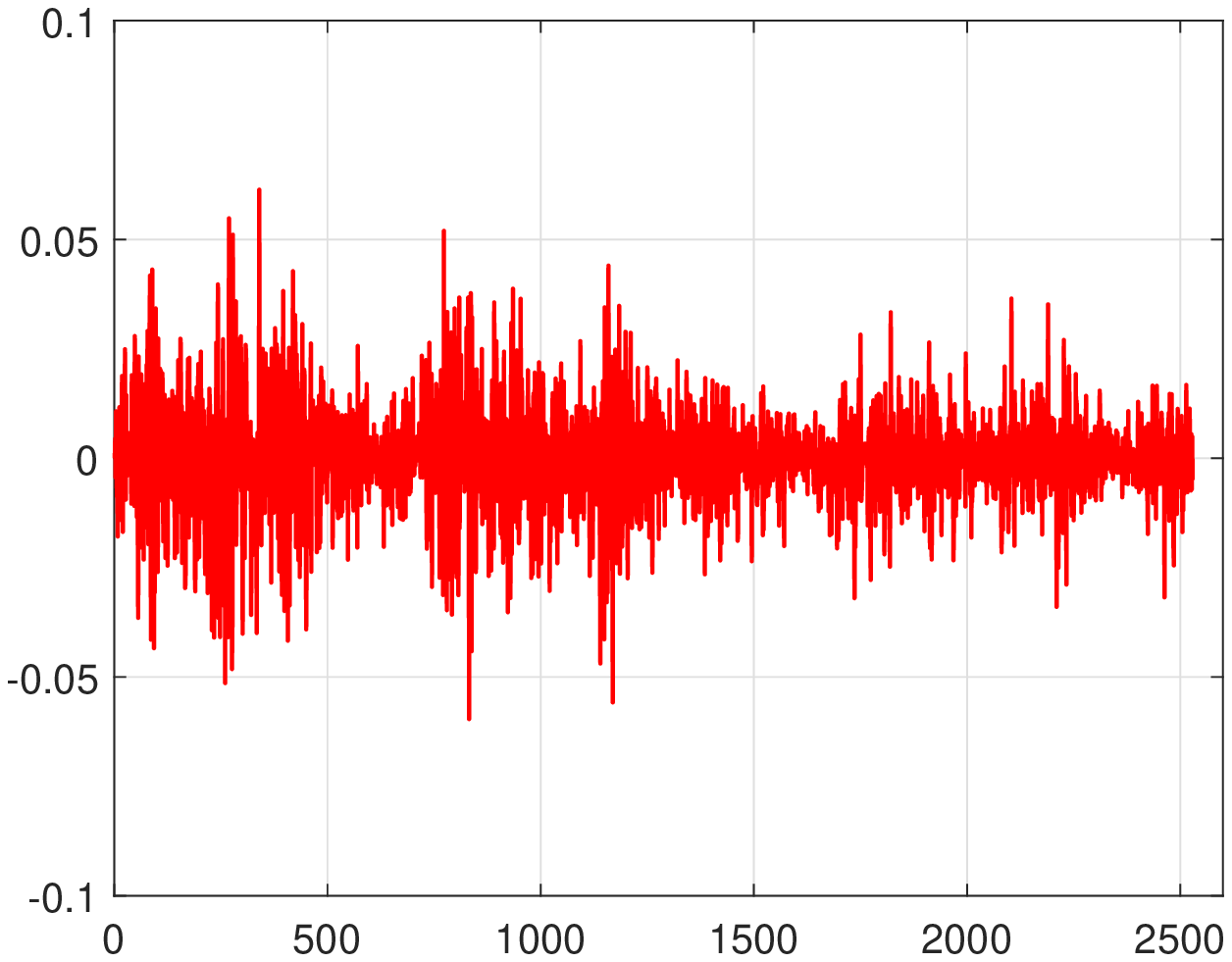}
	\vspace{0.2pc}
	
	\includegraphics[height=3.7cm]{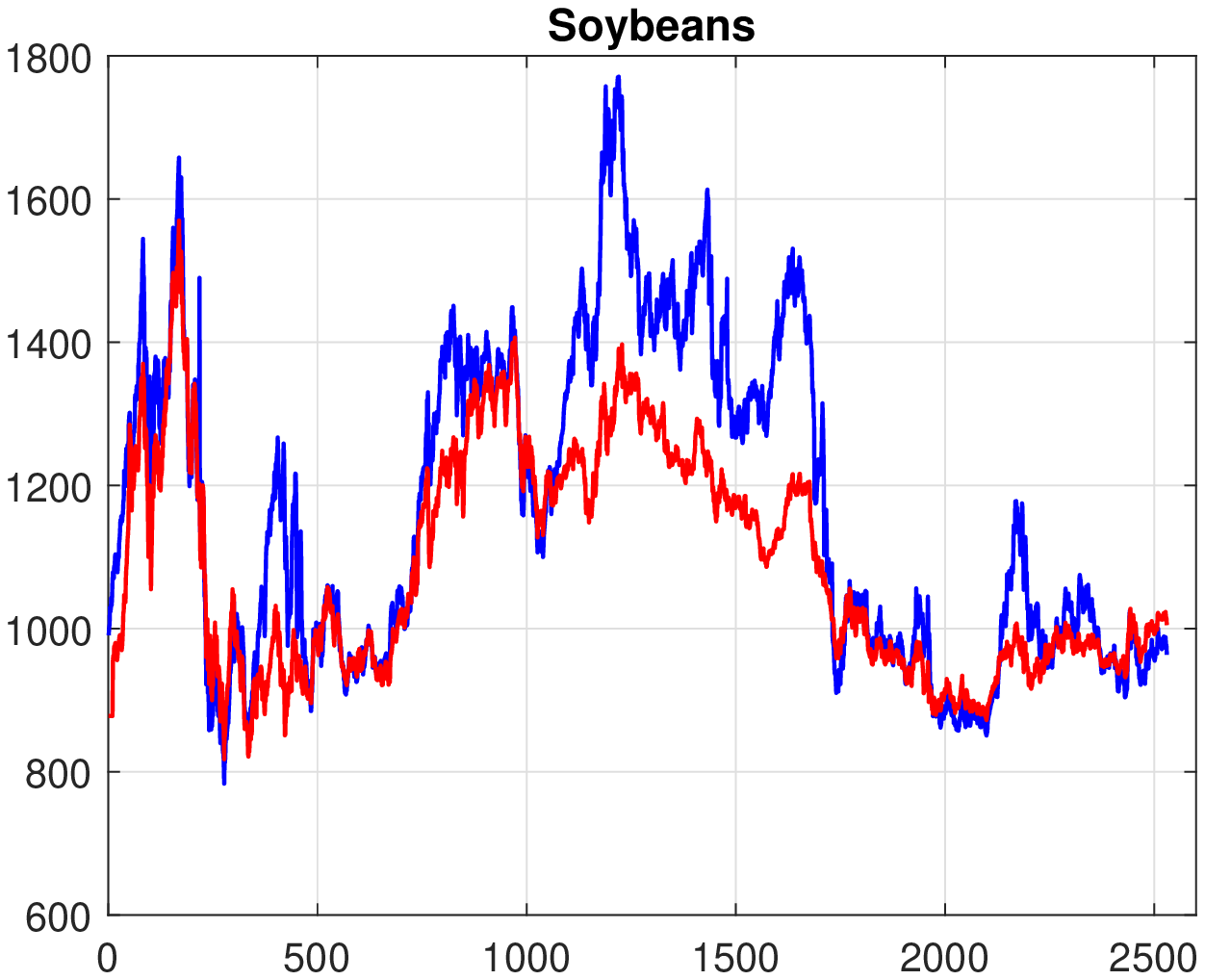} \qquad
	\includegraphics[height=3.5cm]{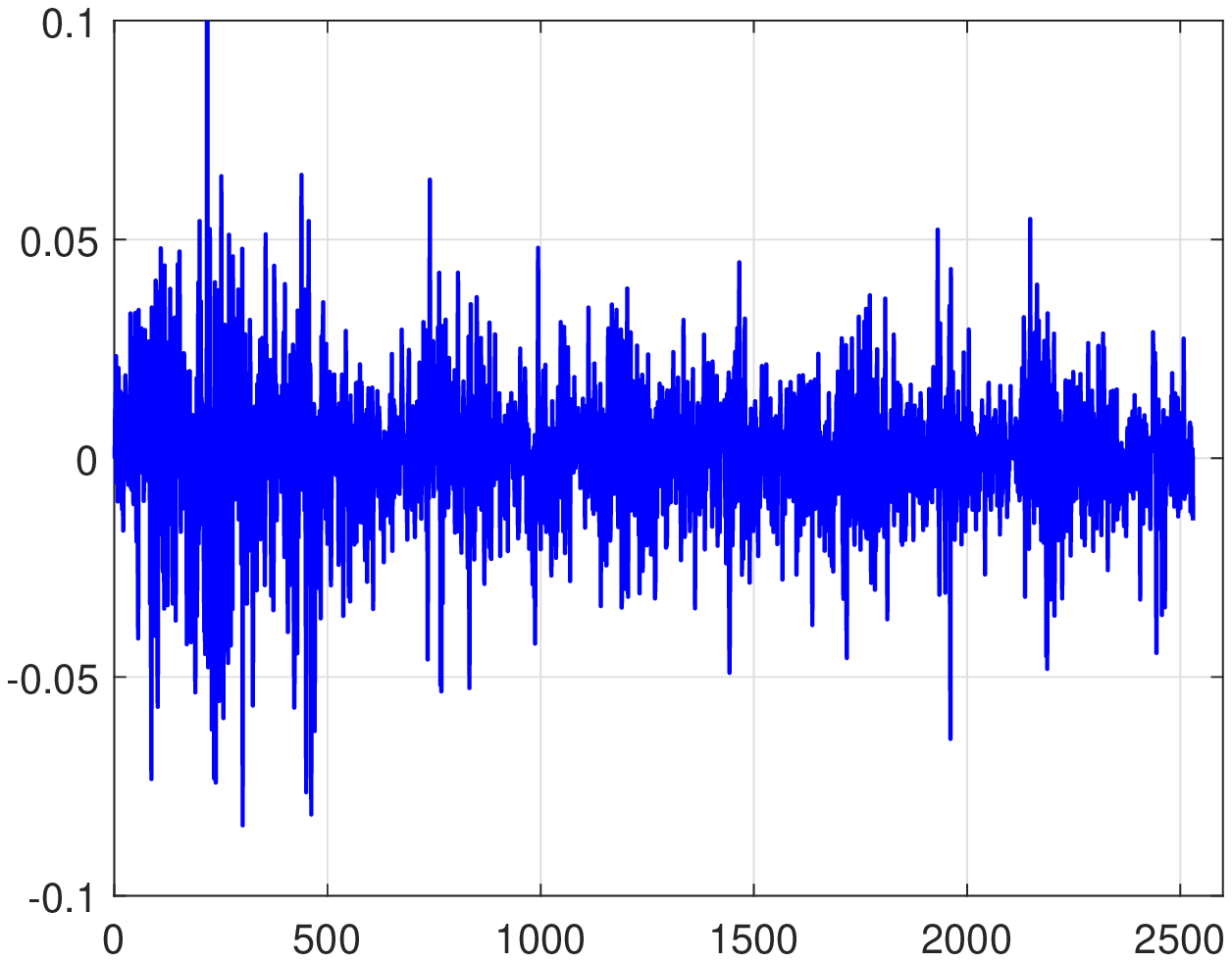} \qquad
	\includegraphics[height=3.5cm]{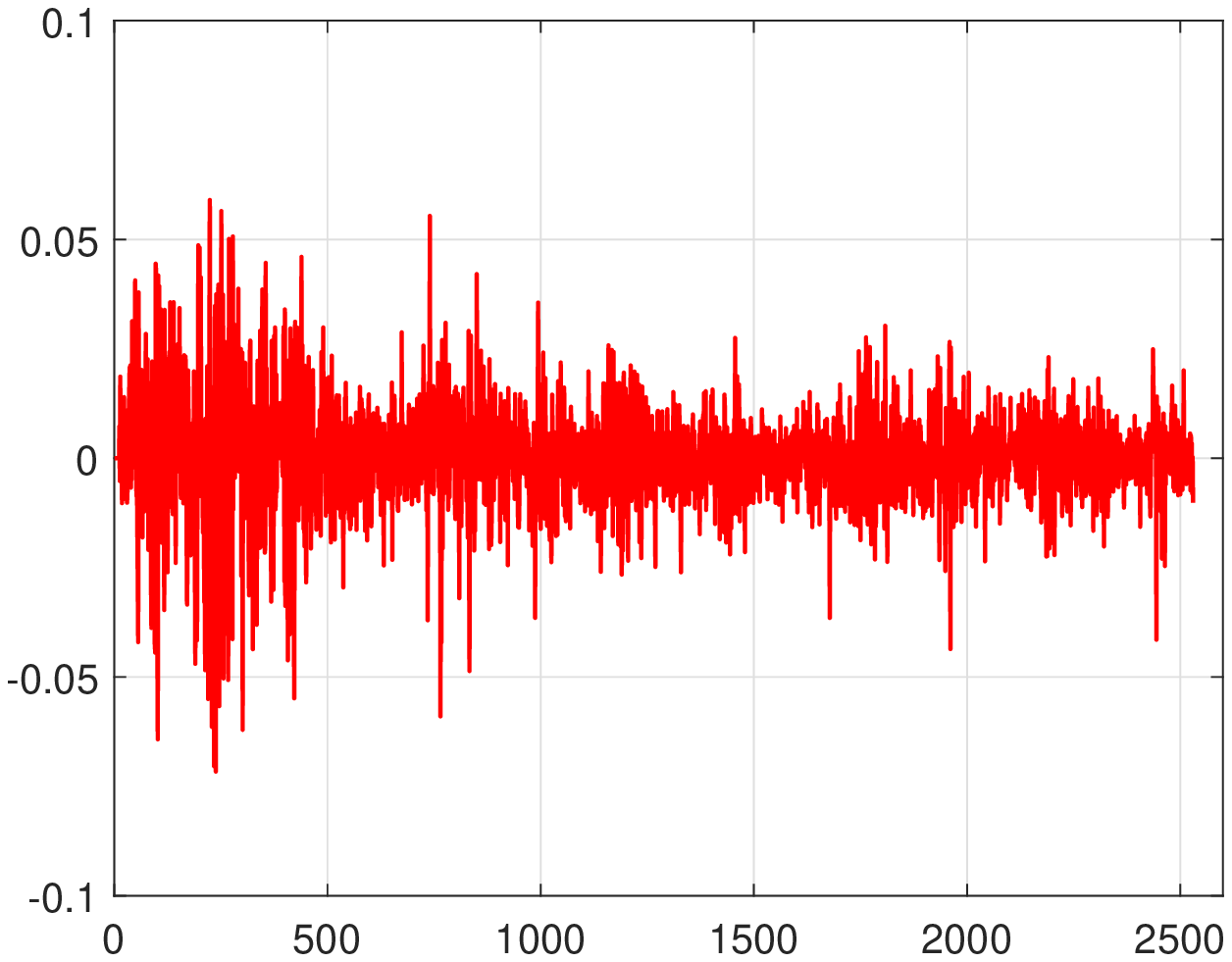}
	\vspace{0.2pc}
	
	\includegraphics[height=3.7cm]{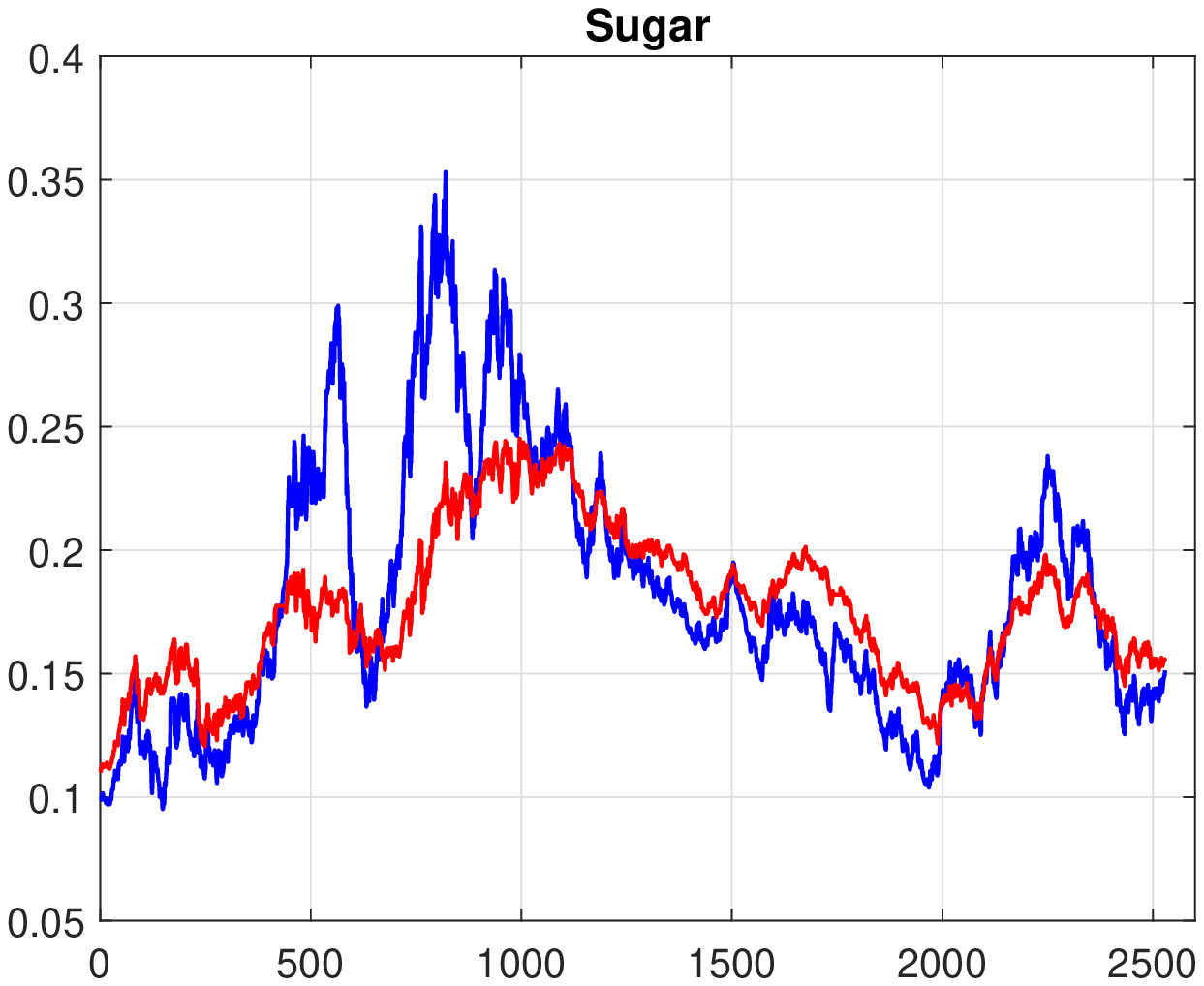} \qquad
	\includegraphics[height=3.5cm]{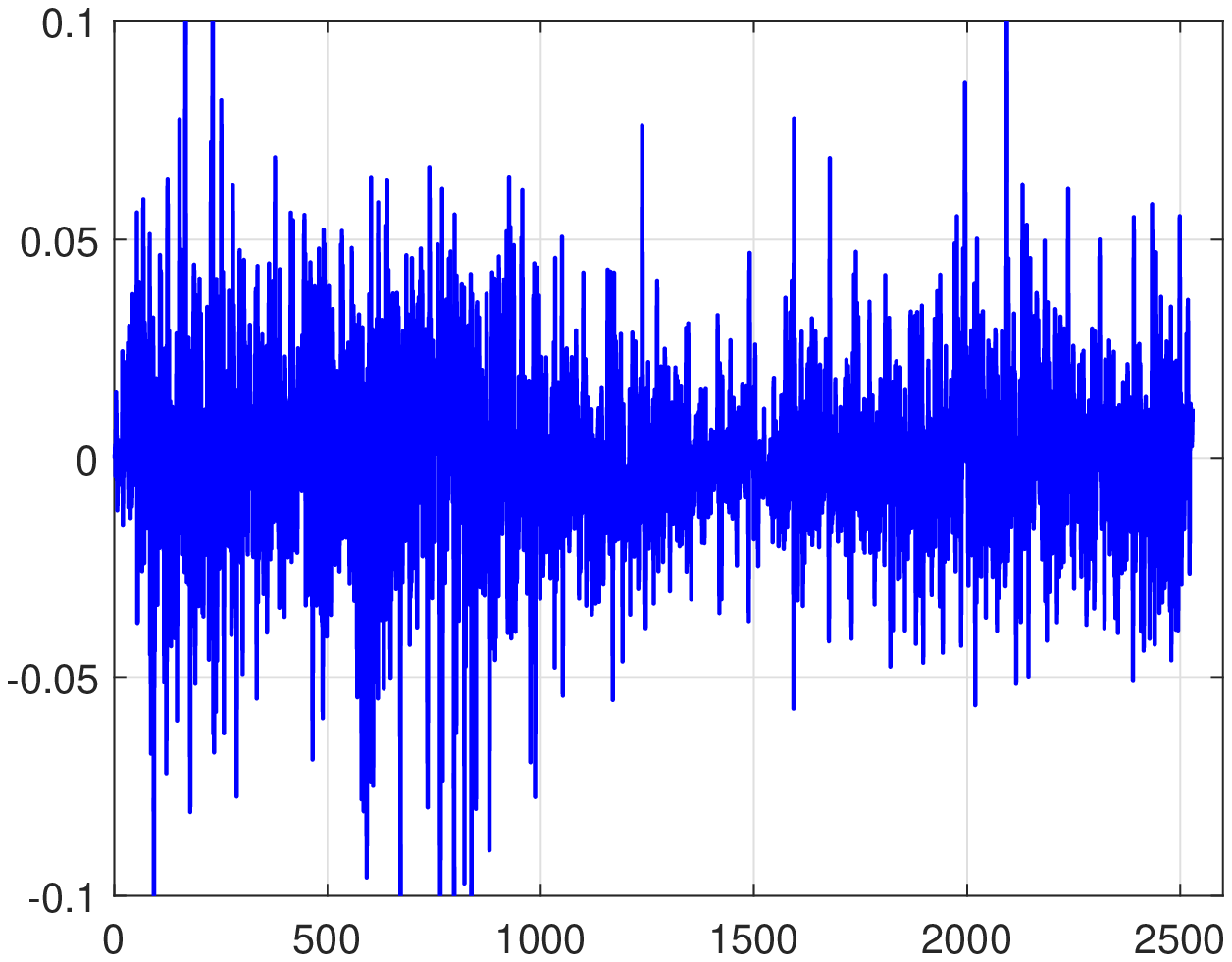} \qquad
	\includegraphics[height=3.5cm]{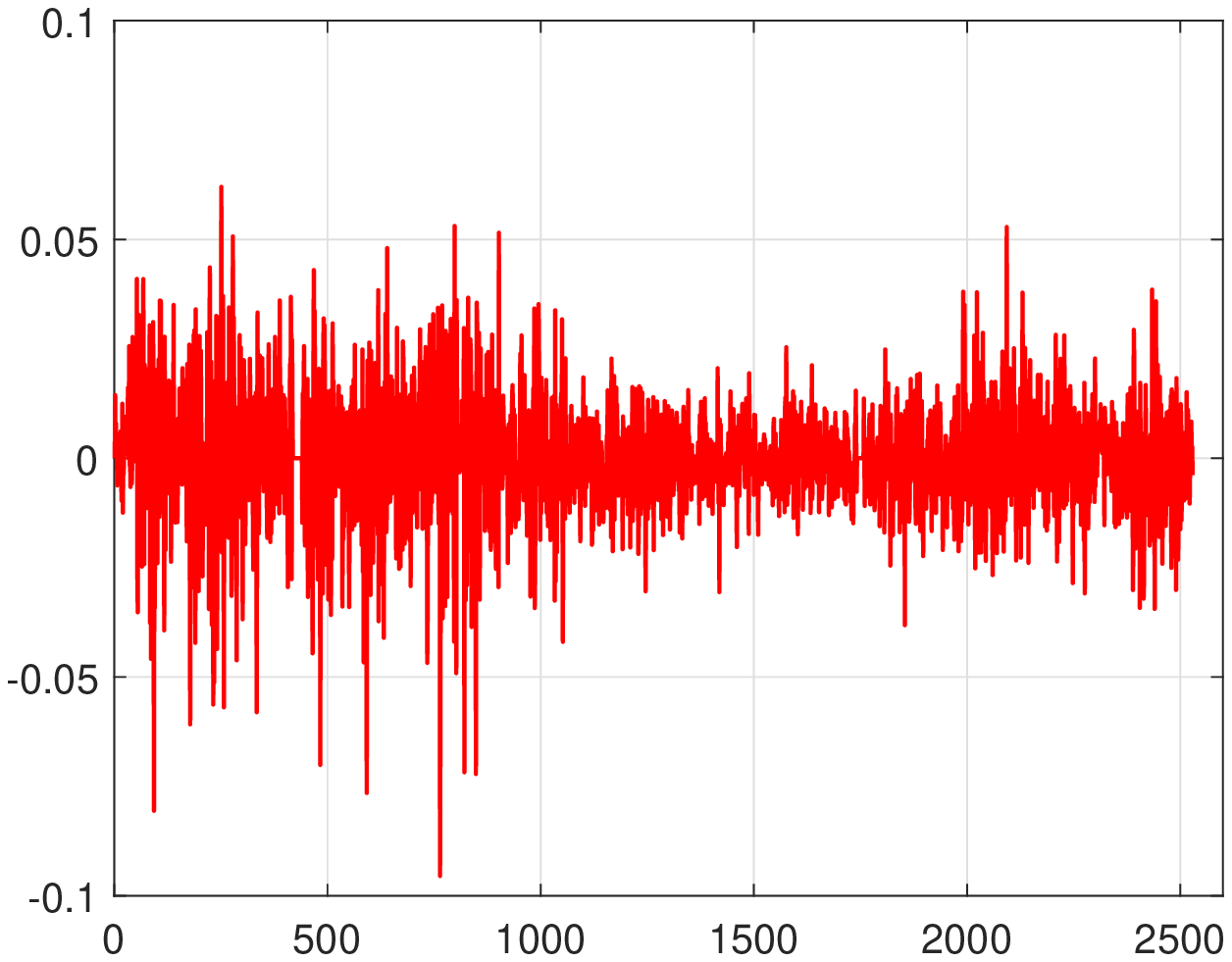}
	\vspace{0.2pc}
	
	\includegraphics[height=3.7cm]{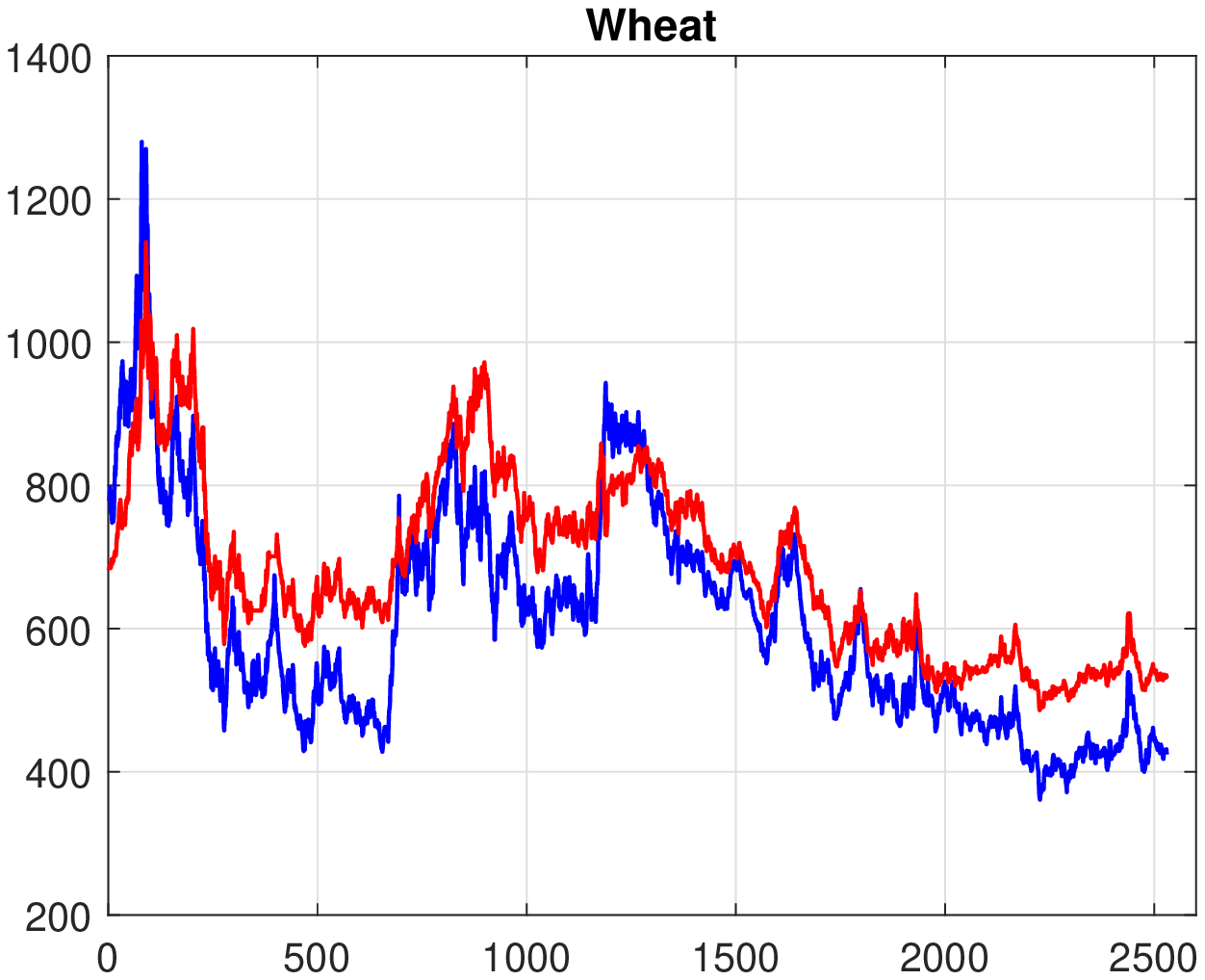} \qquad
	\includegraphics[height=3.5cm]{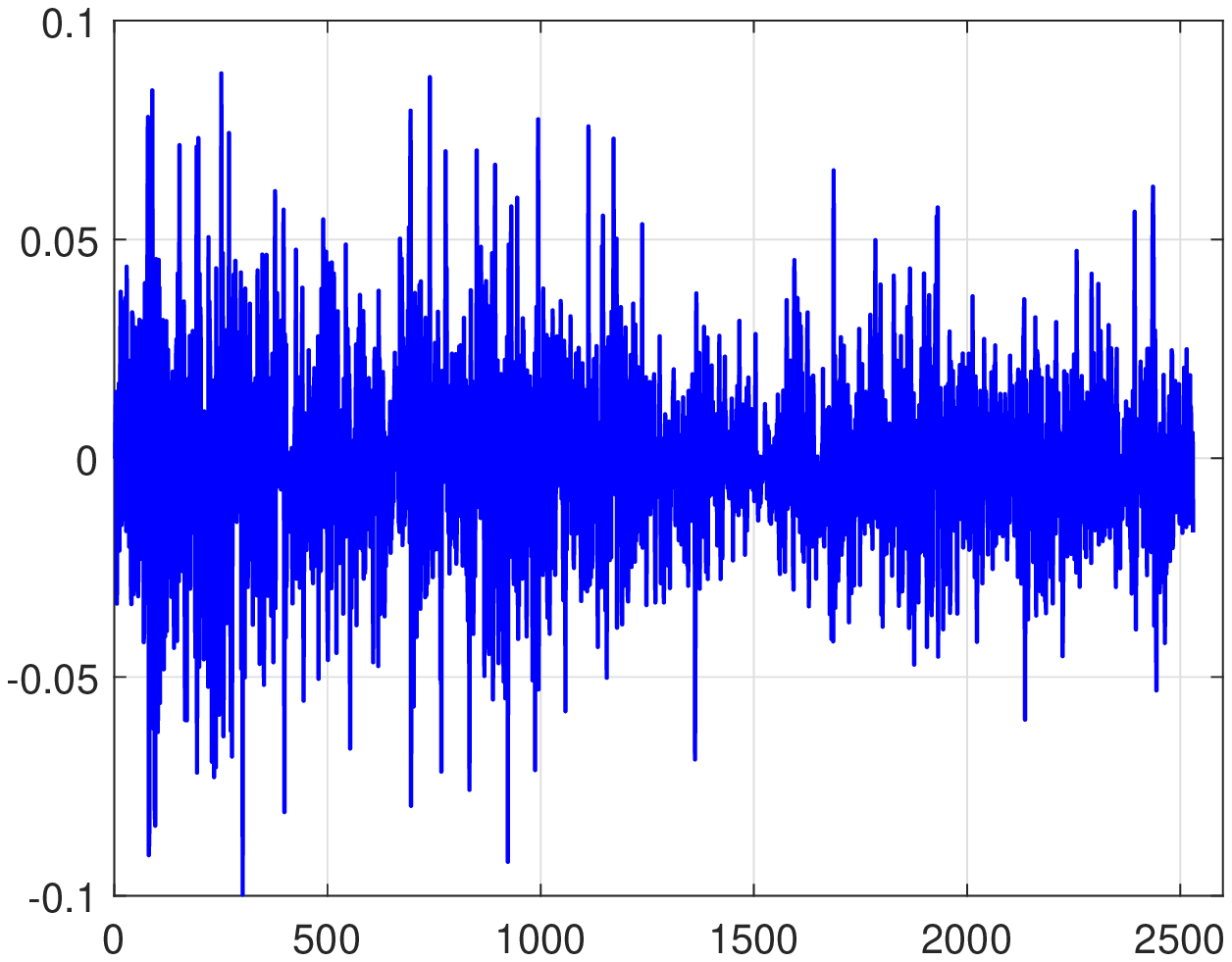} \qquad
	\includegraphics[height=3.5cm]{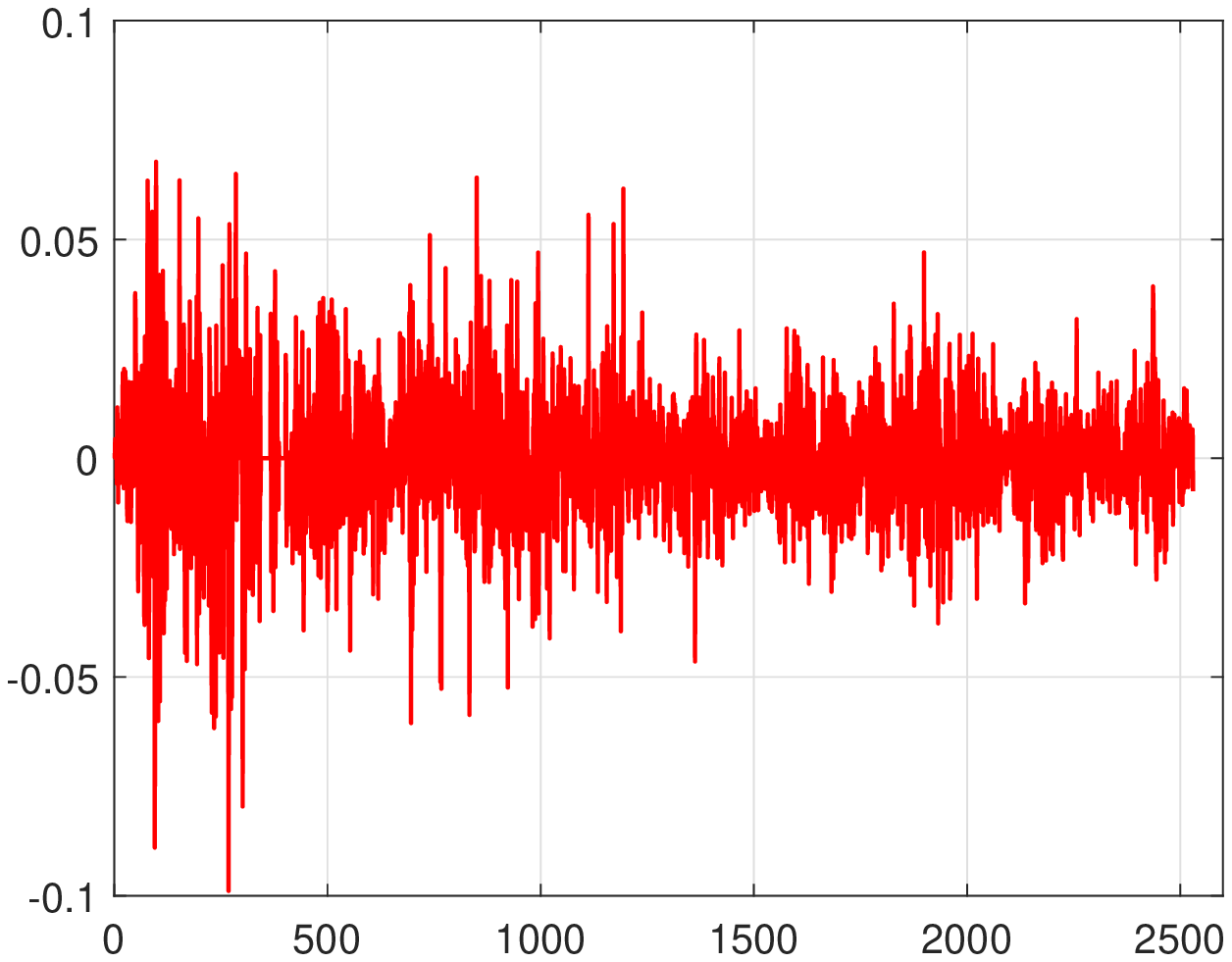}
	\caption{\label{Fig:price_return_ts_plots} Time series of futures prices and log-returns for each commodity.
From \textit{top} to \textit{bottom}: Corn, Cotton, Soybeans, Sugar and Wheat.
\textit{Left column}: prices of the futures with the shortest (\textit{blue}) and longest (\textit{red}) expiries.
\textit{Center column}: log-returns of the contract with the shortest expiry.
\textit{Right column}: log-returns of the futures with the longest expiry.}
\end{figure}

\begin{table}[H]
  \centering
\footnotesize 
     \begin{tabular}{lrrrrrr}
     \toprule
	& Sinusoidal & Exp-sinusoidal & Triangle & Sawtooth & Spiked & Non-seasonal \\
    \midrule
\multicolumn{4}{l}{\textit{Corn:}} \\
LL & $102465.71$ & $102484.74$ & $102472.79$ & $102480.13$ & $102484.19$ & $102453.7$ \\
AIC & $-204893.42$ & $-204931.48$ & $-204907.57$ & $-204922.27$ & $-204930.39$ & $-204873.41$ \\
BIC & $-204782.53$ & $-204820.6$ & $-204796.69$ & $-204811.39$ & $-204819.5$ & $-204774.2$ \\
D1 (LR test 1) & $24.01$ & $62.07$ & $38.16$ & $52.86$ & $60.98$ & $-$ \\
p-value (test 1) & $0.0000$ & $0.0000$ & $0.0000$ & $0.0000$ & $0.0000$ & $-$ \\
$\Delta_{aic}$ & $38.07$ & $0.0000$ & $23.91$ & $9.21$ & $1.1$ & $58.07$ \\
$\omega_i$ & $0.0000$ & $0.6463$ & $0.0000$ & $0.0000$ & $0.3537$ & $0.0000$ \\
LL w/o $\lambda$ & $100161.49$ & $100175.27$ & $100158.01$ & $100144.85$ & $100173.33$ & $100113.78$ \\
D2 (LR test 2) & $4608.43$ & $4618.94$ & $4629.55$ & $4670.56$ & $4621.73$ & $4679.84$ \\
p-value (test 2) & $0.0000$ & $0.0000$ & $0.0000$ & $0.0000$ & $0.0000$ & $0.0000$ \\ \\

\multicolumn{4}{l}{\textit{Cotton:}} \\
LL & $92282.69$ & $92283.76$ & $92280.73$ & $92272.13$ & $92296.98$ & $92261.8$ \\
AIC & $-184527.38$ & $-184529.52$ & $-184523.46$ & $-184506.26$ & $-184555.97$ & $-184489.61$ \\
BIC & $-184416.5$ & $-184418.65$ & $-184412.58$ & $-184395.39$ & $-184445.09$ & $-184390.4$ \\
D1 (LR test 1) & $41.77$ & $43.92$ & $37.85$ & $20.66$ & $70.36$ & $-$ \\
p-value (test 1) & $0.0000$ & $0.0000$ & $0.0000$ & $0.0000$ & $0.0000$ & $-$ \\
$\Delta_{aic}$ & $28.59$ & $26.45$ & $32.51$ & $49.71$ & $0.0000$ & $66.36$ \\
$\omega_i$ & $0.0000$ & $0.0000$ & $0.0000$ & $0.0000$ & $1$ & $0.0000$ \\
LL w/o $\lambda$ & $91488.25$ & $91497.24$ & $91486.01$ & $91484.56$ & $91512.44$ & $91426.22$ \\
D2 (LR test 2) & $1588.87$ & $1573.04$ & $1589.43$ & $1575.14$ & $1569.08$ & $1671.16$ \\
p-value (test 2) & $0.0000$ & $0.0000$ & $0.0000$ & $0.0000$ & $0.0000$ & $0.0000$ \\ \\

\multicolumn{4}{l}{\textit{Soybeans:}} \\
LL & $141142.8$ & $141153.78$ & $141142.27$ & $141140.68$ & $141168.86$ & $141128.88$ \\
AIC & $-282241.6$ & $-282263.57$ & $-282240.54$ & $-282237.36$ & $-282293.73$ & $-282217.75$ \\
BIC & $-282113.21$ & $-282135.18$ & $-282112.14$ & $-282108.97$ & $-282165.33$ & $-282101.03$ \\
D1 (LR test 1) & $27.84$ & $49.82$ & $26.78$ & $23.61$ & $79.97$ & $-$ \\
p-value (test 1) & $0.0000$ & $0.0000$ & $0.0000$ & $0.0000$ & $0.0000$ & $-$ \\
$\Delta_{aic}$ & $52.13$ & $30.16$ & $53.19$ & $56.37$ & $0.0000$ & $75.97$ \\
$\omega_i$ & $0.0000$ & $0.0000$ & $0.0000$ & $0.0000$ & $1$ & $0.0000$ \\
LL w/o $\lambda$ & $139993.99$ & $140003$ & $140002.22$ & $139995.84$ & $140024.2$ & $139974.08$ \\
D2 (LR test 2) & $2297.62$ & $2301.57$ & $2280.1$ & $2289.68$ & $2289.32$ & $2309.6$ \\
p-value (test 2) & $0.0000$ & $0.0000$ & $0.0000$ & $0.0000$ & $0.0000$ & $0.0000$ \\ \\

\multicolumn{4}{l}{\textit{Sugar:}} \\
LL & $64438.15$ & $64438.58$ & $64434.82$ & $64433.59$ & $64437.2$ & $64418.96$ \\
AIC & $-128844.31$ & $-128845.15$ & $-128837.64$ & $-128835.17$ & $-128842.39$ & $-128809.91$ \\
BIC & $-128750.94$ & $-128751.78$ & $-128744.27$ & $-128741.8$ & $-128749.02$ & $-128728.22$ \\
D1 (LR test 1) & $38.39$ & $39.24$ & $31.73$ & $29.26$ & $36.48$ & $-$ \\
p-value (test 1) & $0.0000$ & $0.0000$ & $0.0000$ & $0.0000$ & $0.0000$ & $-$ \\
$\Delta_{aic}$ & $0.85$ & $0.0000$ & $7.51$ & $9.98$ & $2.76$ & $35.24$ \\
$\omega_i$ & $0.4059$ & $0.5812$ & $0.0000$ & $0.0000$ & $0.0129$ & $0.0000$ \\
LL w/o $\lambda$ & $63326.34$ & $63328.12$ & $63326.64$ & $63324.38$ & $63329.57$ & $63313.15$ \\
D2 (LR test 2) & $2223.62$ & $2220.92$ & $2216.36$ & $2218.41$ & $2215.26$ & $2211.6$ \\
p-value (test 2) & $0.0000$ & $0.0000$ & $0.0000$ & $0.0000$ & $0.0000$ & $0.0000$ \\ \\

\multicolumn{4}{l}{\textit{Wheat:}} \\
LL & $101640.24$ & $101638.48$ & $101636.36$ & $101638.01$ & $101635.26$ & $101630.6$ \\
AIC & $-203242.47$ & $-203238.96$ & $-203234.71$ & $-203238.03$ & $-203232.51$ & $-203227.2$ \\
BIC & $-203131.59$ & $-203128.07$ & $-203123.83$ & $-203127.14$ & $-203121.63$ & $-203127.99$ \\
D1 (LR test 1) & $19.27$ & $15.76$ & $11.51$ & $14.83$ & $9.32$ & $-$ \\
p-value (test 1) & $0.0001$ & $0.0004$ & $0.0032$ & $0.0006$ & $0.0095$ & $-$ \\
$\Delta_{aic}$ & $0.0000$ & $3.52$ & $7.76$ & $4.45$ & $9.96$ & $15.27$ \\
$\omega_i$ & $0.9979$ & $0.0021$ & $0.0000$ & $0.0001$ & $0.0000$ & $0.0000$ \\
LL w/o $\lambda$ & $99162.99$ & $99162.51$ & $99159.02$ & $99152.05$ & $99155.02$ & $99146.93$ \\
D2 (LR test 2) & $4954.5$ & $4951.93$ & $4954.68$ & $4971.92$ & $4960.47$ & $4967.34$ \\
p-value (test 2) & $0.0000$ & $0.0000$ & $0.0000$ & $0.0000$ & $0.0000$ & $0.0000$ \\
    \bottomrule
    \end{tabular}
\caption{Summary of the results obtained for each commodity with the considered models: log-likelihood, Akaike Information Criterion,
Bayesian Information Criterion, Statistics D1 and D2 of the likelihood ratio tests and their $p$-values,
AIC differences $\Delta_{aic}$ and Akaike weights $\omega_i$.} 
\label{tab:summary_results}
\end{table}


\begin{table}[H]
  \centering
     \begin{tabular}{ccccccccccc}
     \toprule
\multicolumn{3}{c}{\textit{Corn:}} & & \multicolumn{3}{c}{\textit{Cotton:}} & & \multicolumn{3}{c}{\textit{Soybeans:}} \\
\cmidrule(lr){1-3} \cmidrule(lr){5-7} \cmidrule(lr){9-11}
rank & $\Delta_{aic}$ & model & & rank & $\Delta_{aic}$ & model & & rank & $\Delta_{aic}$ & model \\
\cmidrule(lr){1-3} \cmidrule(lr){5-7} \cmidrule(lr){9-11}
$1$ & $0$ & Exp-sinusoidal & & $1$ & $0$ & Spiked & & $1$ & $0$ & Spiked  \\
$2$ & $1.1$ & Spiked & & $2$ & $26.45$ & Exp-sinusoidal & & $2$ & $30.16$ & Exp-sinusoidal  \\
$3$ & $9.21$ & Sawtooth & & $3$ & $28.59$ & Sinusoidal & & $3$ & $52.13$ & Sinusoidal  \\
$4$ & $23.91$ & Triangle & & $4$ & $32.51$ & Triangle & & $4$ & $53.19$ & Triangle  \\
$5$ & $38.07$ & Sinusoidal & & $5$ & $49.71$ & Sawtooth & & $5$ & $56.37$ & Sawtooth  \\
$6$ & $58.07$ & Non-seasonal & & $6$ & $66.36$ & Non-seasonal & & $6$ & $75.97$ & Non-seasonal  \\ \\
\multicolumn{3}{c}{\textit{Sugar:}} & & \multicolumn{3}{c}{\textit{Wheat:}} & & & & \\
\cmidrule(lr){1-3} \cmidrule(lr){5-7}
rank & $\Delta_{aic}$ & model & & rank & $\Delta_{aic}$ & model & &  &  & \\
\cmidrule(lr){1-3} \cmidrule(lr){5-7}
$1$ & $0$ & Exp-sinusoidal & & $1$ & $0$ & Sinusoidal & & & & \\
$2$ & $0.85$ & Sinusoidal & & $2$ & $3.52$ & Exp-sinusoidal & & & & \\
$3$ & $2.76$ & Spiked & & $3$ & $4.45$ & Sawtooth & & & & \\
$4$ & $7.51$ & Triangle & & $4$ & $7.76$ & Triangle & & & & \\
$5$ & $9.98$ & Sawtooth & & $5$ & $9.96$ & Spiked & & & & \\
$6$ & $35.24$ & Non-seasonal & & $6$ & $15.27$ & Non-seasonal & & & & \\
    \bottomrule
    \end{tabular}
\caption{Ranking of models for each commodity, based on AIC. For each model, the differences with respect to the smallest AIC, named $\Delta_{aic}$,
is provided with the model rank and name.}
\label{tab:models_ranking}
\end{table}

\begin{figure}[H]
\centering
	\includegraphics[height=6.0cm]{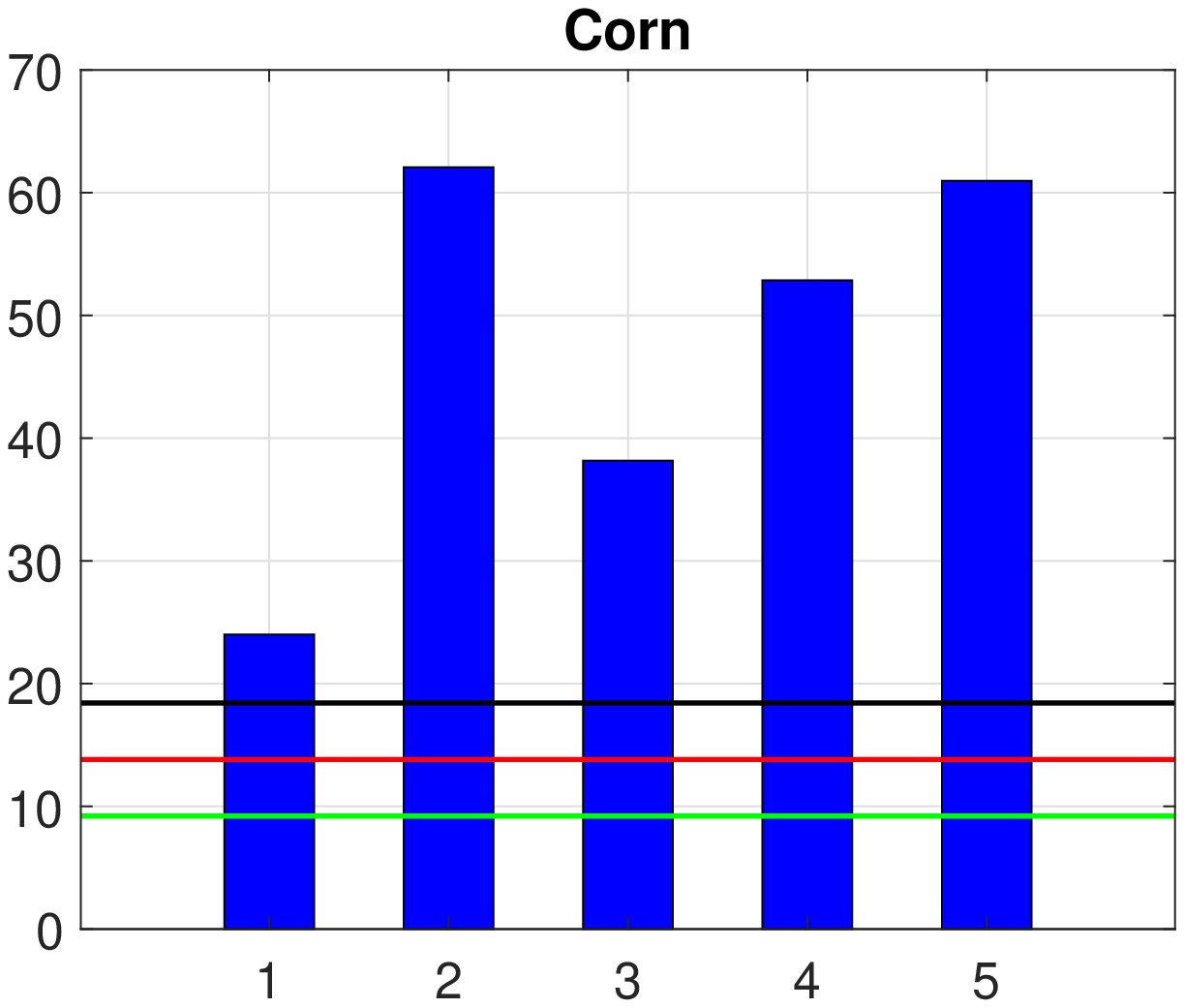} \qquad
	\includegraphics[height=6.0cm]{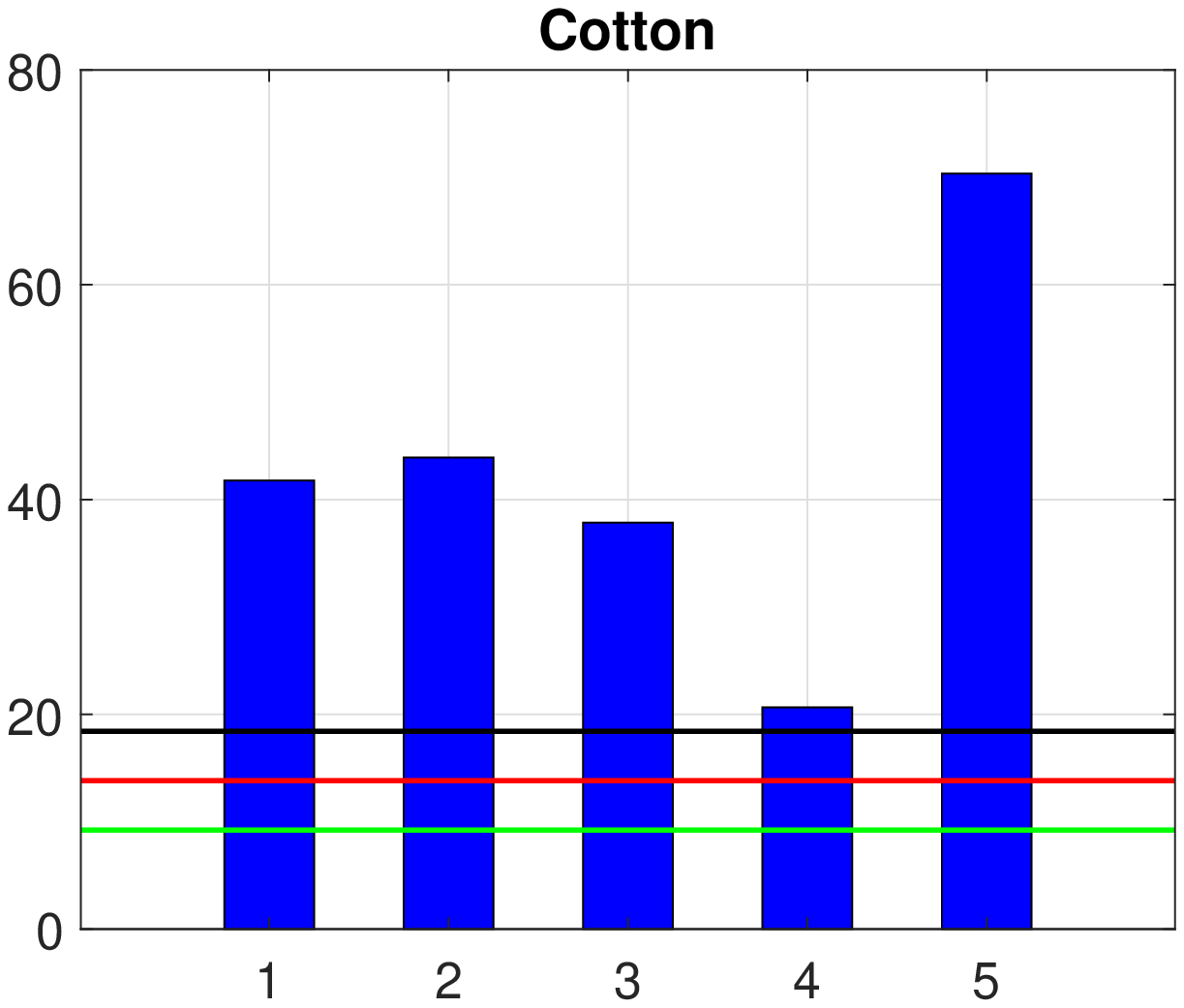}
	\vspace{1.0pc}
	
	\includegraphics[height=6.0cm]{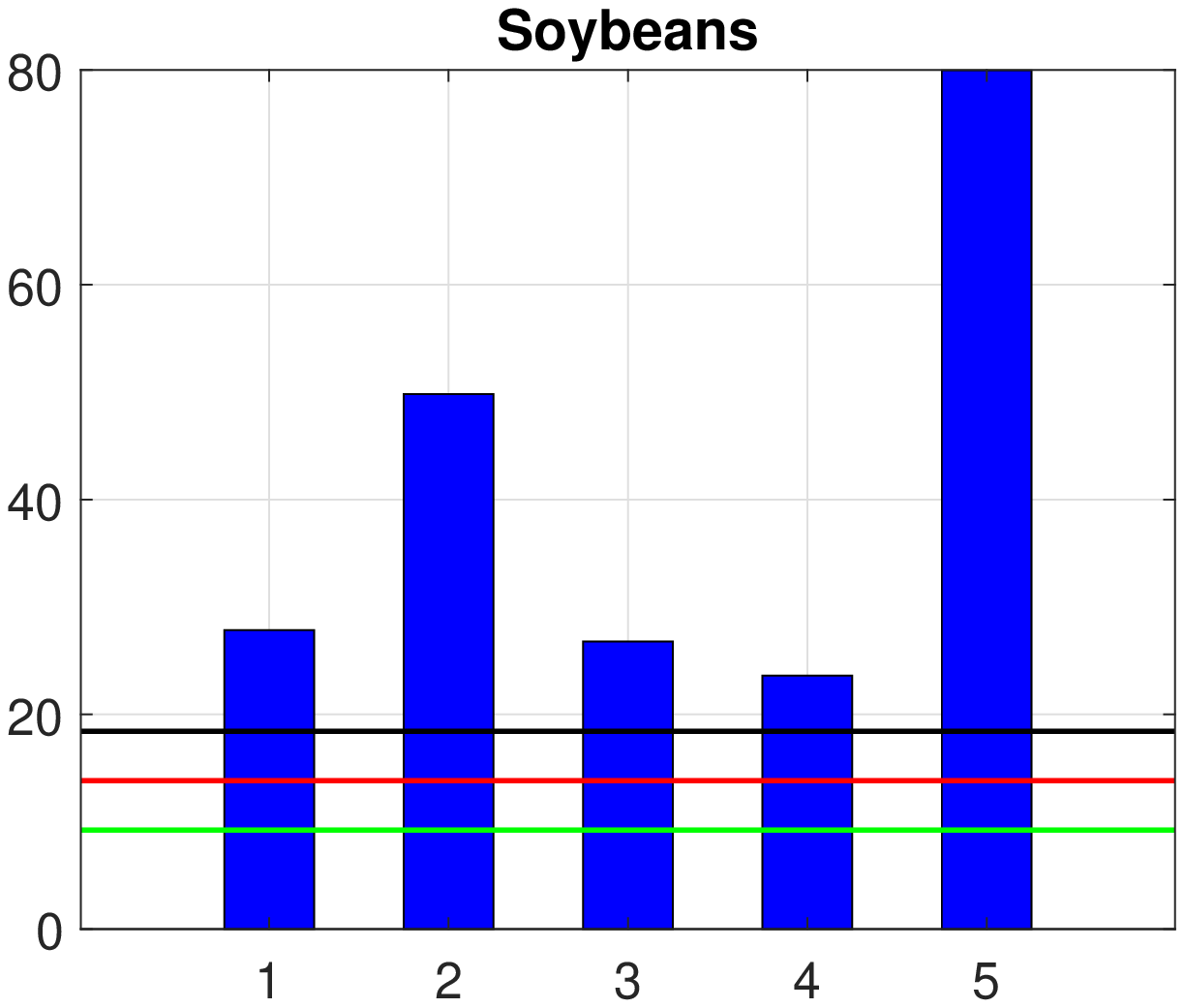} \qquad
	\includegraphics[height=6.0cm]{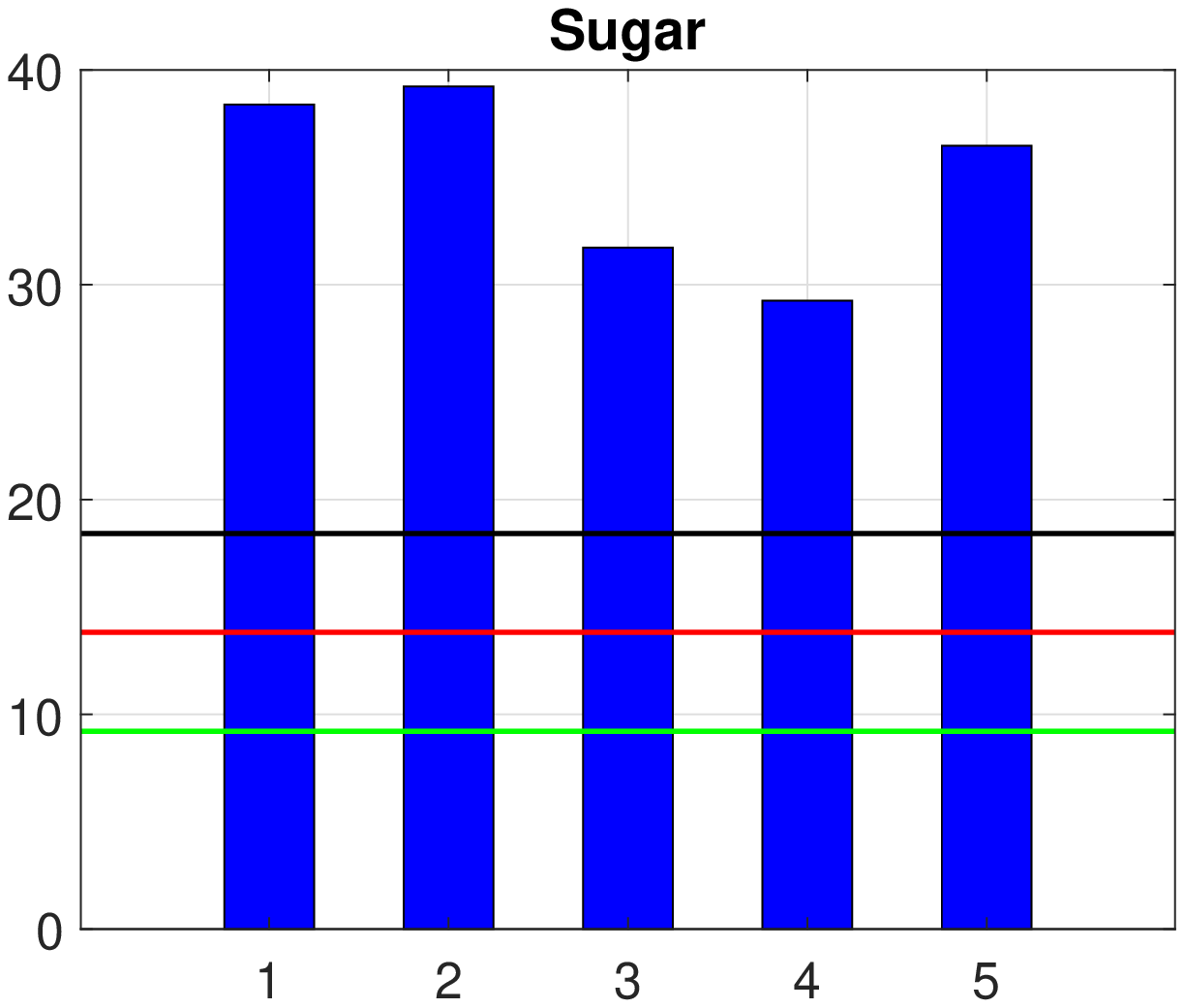}
	\vspace{1.0pc}
	
	\includegraphics[height=6.0cm]{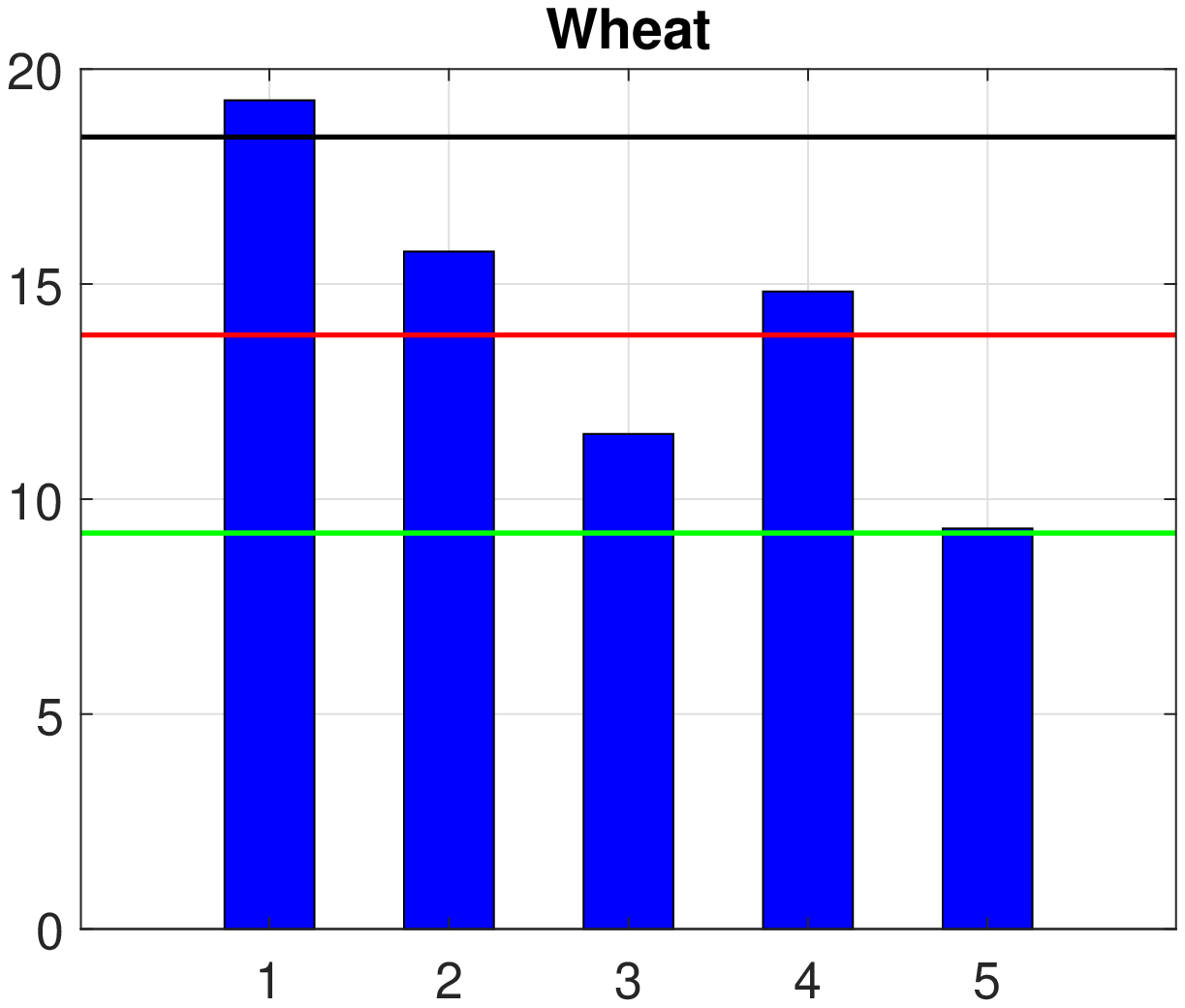}
	\caption{\label{Fig:lrtests_plots} Results of the likelihood ratio tests for each commodity: seasonal models versus the non-seasonal model.
Bars represent the values taken by the $D$ statistic of the test and the horizontal lines correspond to the significance thresholds to reject
the (null) hypothesis of a non-seasonal model, at 99\% (\textit{green}), 99.9\% (\textit{red}) and 99.99\% (\textit{black}) confidence levels.
Tested seasonal models are numbered as: 1. sinusoidal, 2. exp-sinusoidal, 3. triangle, 4. sawtooth and 5. spiked.
\textit{Upper left}: Corn. \textit{Upper right}: Cotton. \textit{Center left}: Soybeans. \textit{Center right}: Sugar. \textit{Lower}: Wheat.}
\end{figure}

\begin{figure}[H]
\centering
	\includegraphics[height=6.0cm]{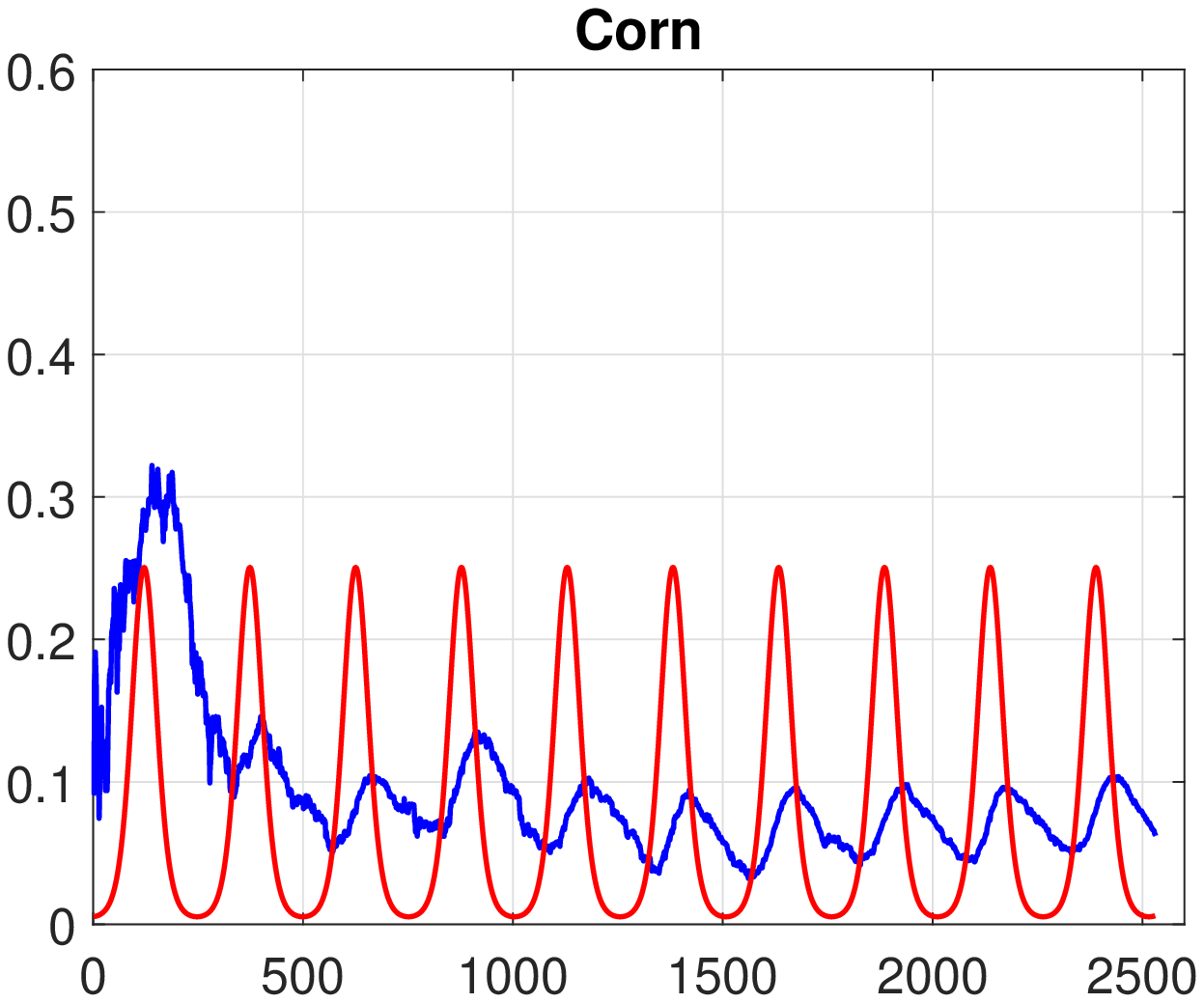} \qquad
	\includegraphics[height=6.0cm]{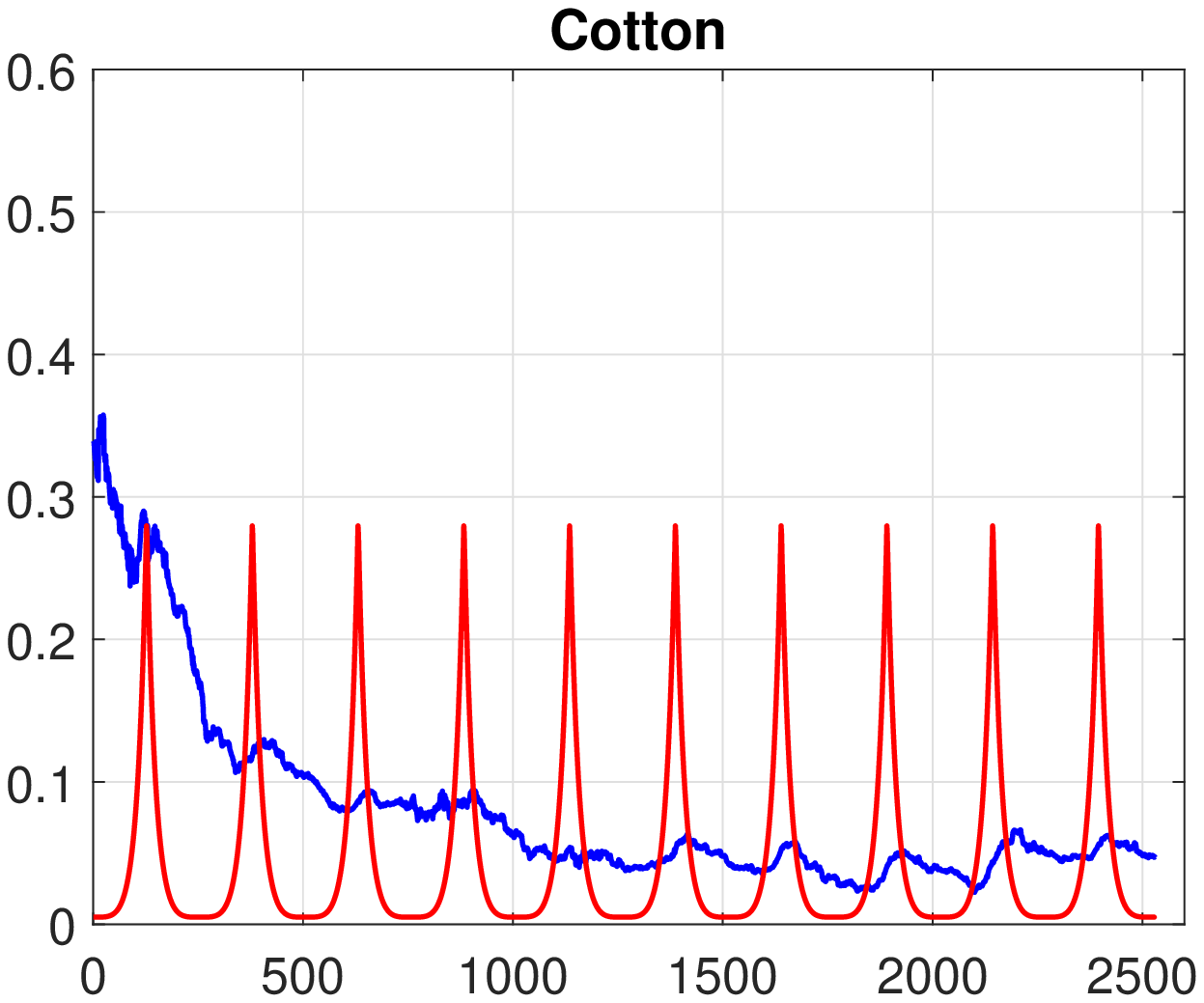}
	\vspace{1.0pc}
	
	\includegraphics[height=6.0cm]{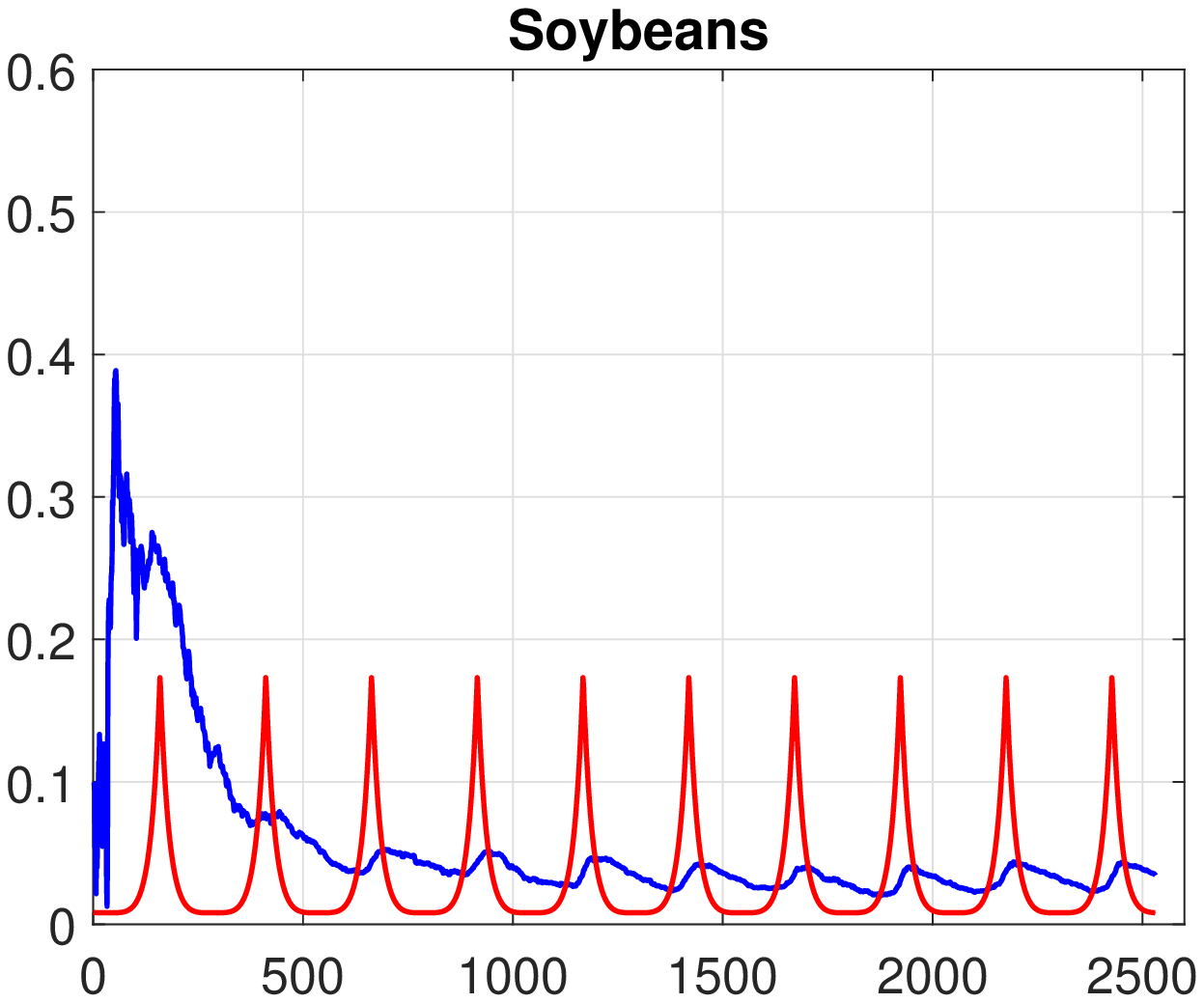} \qquad
	\includegraphics[height=6.0cm]{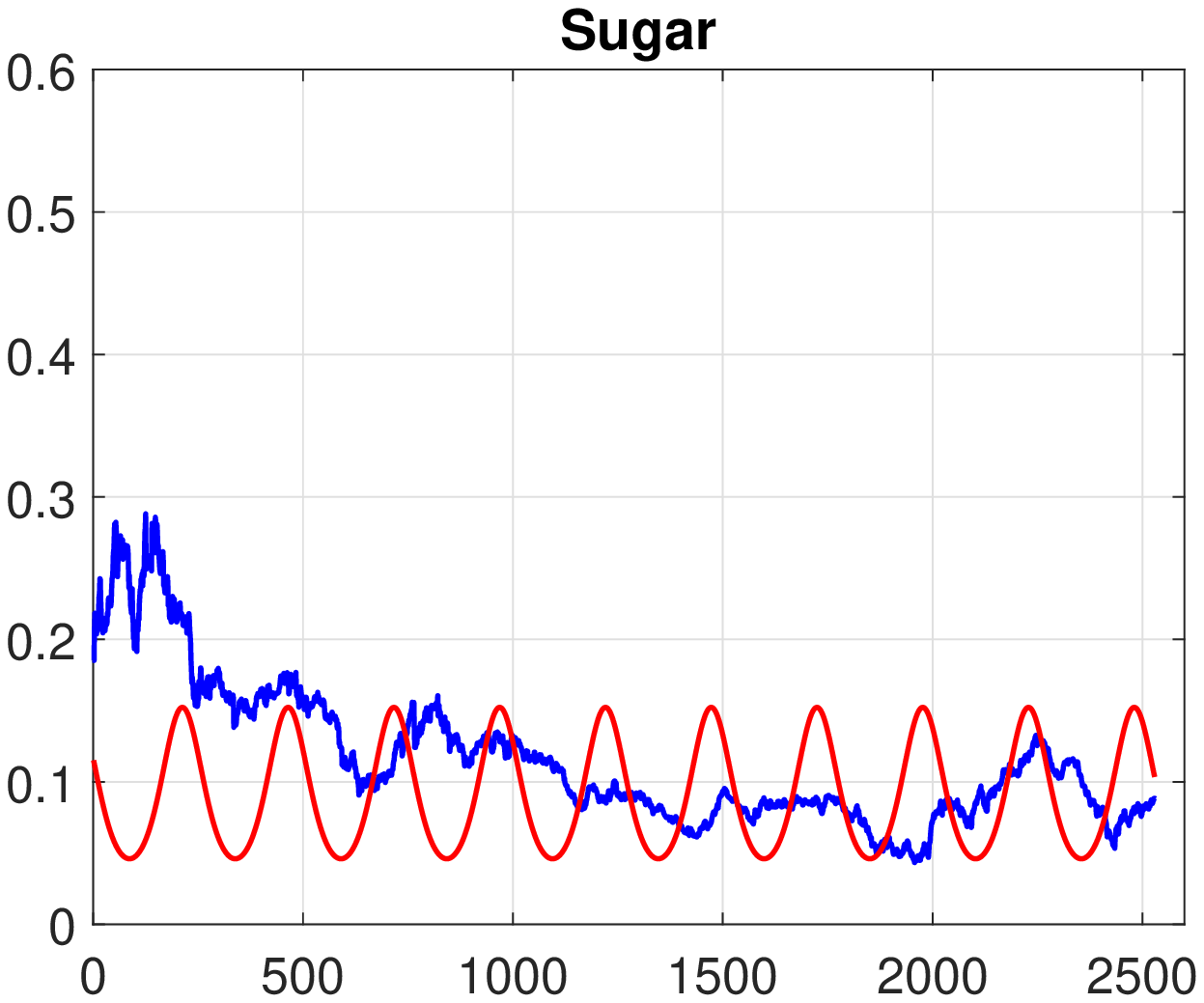}
	\vspace{1.0pc}
	
	\includegraphics[height=6.0cm]{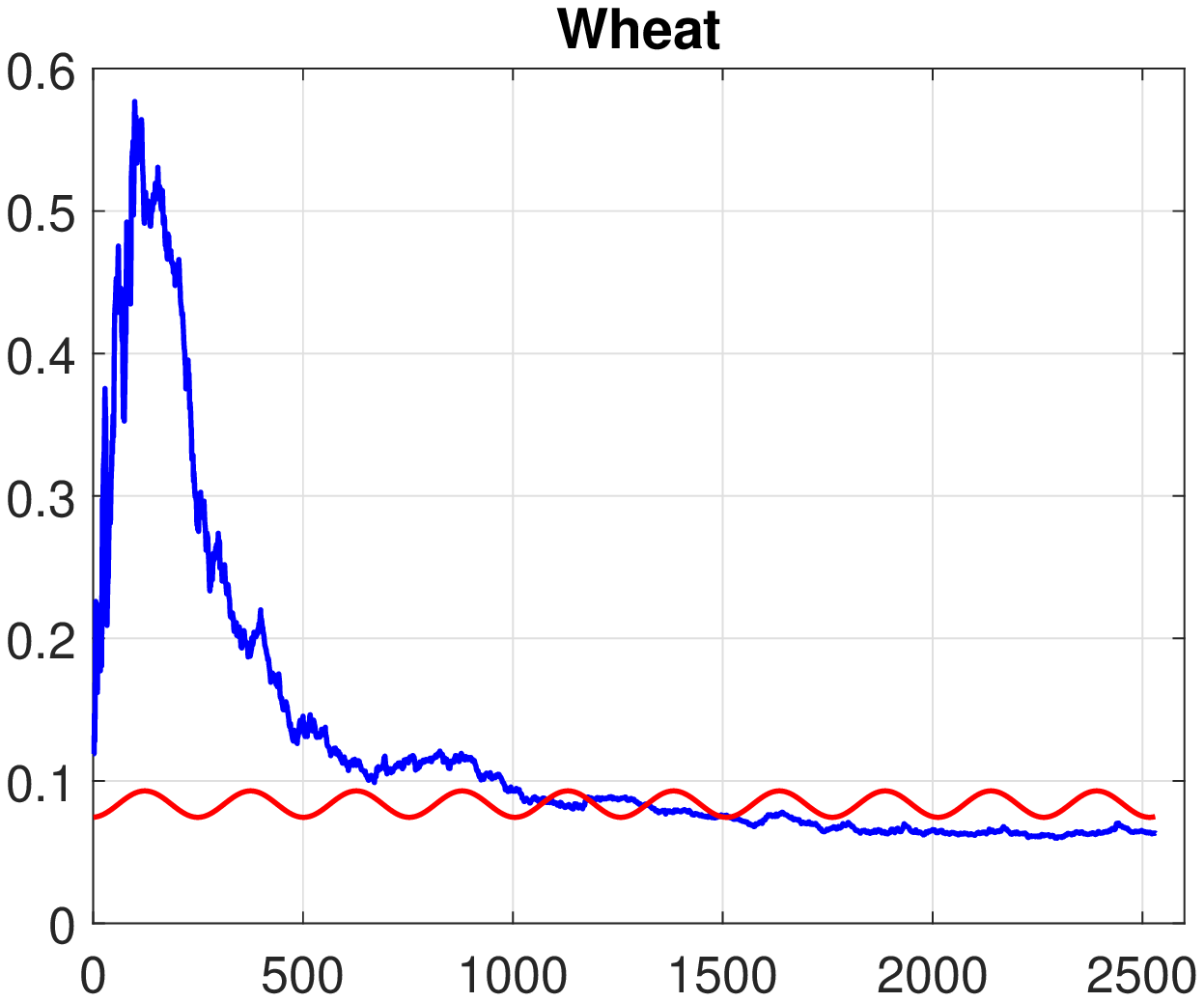}
	\caption{\label{Fig:sv3ts_plots} Time series of $v(t)$ (\textit{blue}) and $\theta(t)$ (\textit{red}) obtained the best model according to Table \ref{tab:models_ranking},
for each commodity.
\textit{Upper left}: Corn. \textit{Upper right}: Cotton.
\textit{Center left}: Soybeans. \textit{Center right}: Sugar.
\textit{Lower}: Wheat.}
\end{figure}

\begin{table}[H]
  \centering
\footnotesize 
     \begin{tabular}{crrrrrr}
     \toprule
	& Sinusoidal & Exp-sinusoidal & Triangle & Sawtooth & Spiked & Non-seasonal \\
    \midrule
$\lambda$ & $0.2122$ & $0.2122$ & $0.2113$ & $0.2094$ & $0.2075$ & $0.2122$ \\
 & $-1.5503$ & $-1.5503$ & $-1.5543$ & $-1.5634$ & $-1.5727$ & $-1.5504$ \\
 & $(0.0134)$ & $(0.0135)$ & $(0.0136)$ & $(0.0138)$ & $(0.0138)$ & $(0.0135)$ \\
$\pi^F$ & $2.1940$ & $2.4622$ & $1.9434$ & $3.1085$ & $2.1273$ & $3.0675$ \\
 & $2.1940$ & $2.4622$ & $1.9434$ & $3.1085$ & $2.1273$ & $3.0675$ \\
 & $(0.1205)$ & $(0.1354)$ & $(0.1625)$ & $(0.0980)$ & $(0.1454)$ & $(0.1377)$ \\
$v_0$ & $0.0851$ & $0.0925$ & $0.0775$ & $0.1012$ & $0.0796$ & $0.1080$ \\
 & $-2.4644$ & $-2.3804$ & $-2.5570$ & $-2.2906$ & $-2.5303$ & $-2.2261$ \\
 & $(0.6804)$ & $(0.5332)$ & $(0.4761)$ & $(0.2837)$ & $(0.5373)$ & $(0.4920)$ \\
$\kappa$ & $1.5624$ & $1.4066$ & $1.5300$ & $0.5307$ & $1.4500$ & $0.9043$ \\
 & $0.4462$ & $0.3411$ & $0.4252$ & $-0.6336$ & $0.3716$ & $-0.1005$ \\
 & $(0.0530)$ & $(0.0576)$ & $(0.0612)$ & $(0.0737)$ & $(0.0517)$ & $(0.0904)$ \\
$\sigma$ & $0.3984$ & $0.3364$ & $0.3528$ & $0.1033$ & $0.3311$ & $0.1579$ \\
 & $-0.9202$ & $-1.0894$ & $-1.0418$ & $-2.2700$ & $-1.1054$ & $-1.8456$ \\
 & $(0.0851)$ & $(0.0943)$ & $(0.2152)$ & $(0.0875)$ & $(0.1380)$ & $(0.0881)$ \\
$\rho$ & $-0.0269$ & $-0.0295$ & $-0.0172$ & $-0.0390$ & $-0.0342$ & $0.0639$ \\
 & $-0.0422$ & $-0.0463$ & $-0.0270$ & $-0.0614$ & $-0.0538$ & $0.1006$ \\
 & $(0.0380)$ & $(0.0451)$ & $(0.0487)$ & $(0.0752)$ & $(0.0459)$ & $(0.0727)$ \\
$a$ & $0.0719$ & $0.0364$ & $0.0015$ & $0.0040$ & $0.0124$ & $0.0742$ \\
 & $-2.6330$ & $-3.3132$ & $-6.5253$ & $-5.5215$ & $-4.3901$ & $-2.6003$ \\
 & $(0.0308)$ & $(0.0325)$ & $(1.0243)$ & $(0.7224)$ & $(0.1831)$ & $(0.0327)$ \\
$b$ & $0.0597$ & $1.9290$ & $0.2966$ & $0.1138$ & $0.3948$ & - \\
 & $1.7073$ & $0.6570$ & $-1.2154$ & $-2.1734$ & $-0.9294$ & - \\
 & $(0.6640)$ & $(0.0244)$ & $(0.0332)$ & $(0.0403)$ & $(0.0392)$ & - \\
$t_0$ & $0.3120$ & $0.3112$ & $0.3027$ & $0.3546$ & $0.2850$ & - \\
 & $-0.6706$ & $-0.6743$ & $-0.7136$ & $-0.4916$ & $-0.8012$ & - \\
 & $(0.1057)$ & $(0.0784)$ & $(0.2101)$ & $(0.0762)$ & $(0.2818)$ & - \\
$h_1$ & $0.0066$ & $0.0066$ & $0.0066$ & $0.0066$ & $0.0066$ & $0.0066$ \\
 & $-5.0215$ & $-5.0185$ & $-5.0237$ & $-5.0262$ & $-5.0242$ & $-5.0219$ \\
 & $(0.0144)$ & $(0.0144)$ & $(0.0144)$ & $(0.0143)$ & $(0.0144)$ & $(0.0144)$ \\
$h_2$ & $0.0040$ & $0.0040$ & $0.0040$ & $0.0040$ & $0.0040$ & $0.0040$ \\
 & $-5.5154$ & $-5.5148$ & $-5.5176$ & $-5.5170$ & $-5.5193$ & $-5.5157$ \\
 & $(0.0149)$ & $(0.0149)$ & $(0.0149)$ & $(0.0149)$ & $(0.0149)$ & $(0.0149)$ \\
$h_3$ & $0.0027$ & $0.0027$ & $0.0027$ & $0.0027$ & $0.0027$ & $0.0027$ \\
 & $-5.9012$ & $-5.9070$ & $-5.9156$ & $-5.9110$ & $-5.9169$ & $-5.9043$ \\
 & $(0.0158)$ & $(0.0158)$ & $(0.0157)$ & $(0.0157)$ & $(0.0156)$ & $(0.0158)$ \\
$h_4$ & $0.0019$ & $0.0019$ & $0.0019$ & $0.0019$ & $0.0019$ & $0.0019$ \\
 & $-6.2774$ & $-6.2778$ & $-6.2805$ & $-6.2846$ & $-6.2624$ & $-6.2780$ \\
 & $(0.0179)$ & $(0.0179)$ & $(0.0179)$ & $(0.0178)$ & $(0.0178)$ & $(0.0178)$ \\
$h_5$ & $0.0015$ & $0.0015$ & $0.0015$ & $0.0015$ & $0.0015$ & $0.0015$ \\
 & $-6.5067$ & $-6.5097$ & $-6.5100$ & $-6.5146$ & $-6.5042$ & $-6.5069$ \\
 & $(0.0203)$ & $(0.0203)$ & $(0.0203)$ & $(0.0203)$ & $(0.0203)$ & $(0.0203)$ \\
$\dots$ & $0.0021$ & $0.0021$ & $0.0021$ & $0.0021$ & $0.0021$ & $0.0021$ \\
 & $-6.1692$ & $-6.1701$ & $-6.1713$ & $-6.1747$ & $-6.1767$ & $-6.1708$ \\
 & $(0.0165)$ & $(0.0164)$ & $(0.0164)$ & $(0.0164)$ & $(0.0165)$ & $(0.0164)$ \\
 & $0.0031$ & $0.0031$ & $0.0030$ & $0.0030$ & $0.0030$ & $0.0030$ \\
 & $-5.7924$ & $-5.7918$ & $-5.7961$ & $-5.7965$ & $-5.8039$ & $-5.7937$ \\
 & $(0.0150)$ & $(0.0150)$ & $(0.0149)$ & $(0.0149)$ & $(0.0149)$ & $(0.0150)$ \\
 & $0.0038$ & $0.0038$ & $0.0038$ & $0.0038$ & $0.0038$ & $0.0038$ \\
 & $-5.5798$ & $-5.5794$ & $-5.5657$ & $-5.5783$ & $-5.5845$ & $-5.5801$ \\
 & $(0.0146)$ & $(0.0146)$ & $(0.0148)$ & $(0.0146)$ & $(0.0146)$ & $(0.0146)$ \\
 & $0.0042$ & $0.0043$ & $0.0043$ & $0.0043$ & $0.0042$ & $0.0043$ \\
 & $-5.4625$ & $-5.4597$ & $-5.4521$ & $-5.4573$ & $-5.4683$ & $-5.4606$ \\
 & $(0.0145)$ & $(0.0145)$ & $(0.0146)$ & $(0.0145)$ & $(0.0144)$ & $(0.0145)$ \\
 & $0.0047$ & $0.0047$ & $0.0047$ & $0.0047$ & $0.0047$ & $0.0047$ \\
 & $-5.3582$ & $-5.3585$ & $-5.3637$ & $-5.3595$ & $-5.3707$ & $-5.3599$ \\
 & $(0.0144)$ & $(0.0143)$ & $(0.0143)$ & $(0.0143)$ & $(0.0142)$ & $(0.0143)$ \\

    \bottomrule
    \end{tabular}
\caption{Estimated parameters for Corn. For each parameter: estimated value in the standard space, transformed in $\mathbb{R}$
and standard error between parenthesis (standard deviation of the estimate in $\mathbb{R}$).}
\label{tab:estimated_parameters_corn}
\end{table}

\begin{table}[H]
  \centering
\footnotesize 
     \begin{tabular}{crrrrrr}
     \toprule
	& Sinusoidal & Exp-sinusoidal & Triangle & Sawtooth & Spiked & Non-seasonal \\
    \midrule
$\lambda$ & $0.2174$ & $0.2209$ & $0.2200$ & $0.2175$ & $0.2176$ & $0.2191$ \\
 & $-1.5258$ & $-1.5098$ & $-1.5142$ & $-1.5256$ & $-1.5251$ & $-1.5184$ \\
 & $(0.0207)$ & $(0.0203)$ & $(0.0205)$ & $(0.0209)$ & $(0.0209)$ & $(0.0208)$ \\
$\pi^F$ & $0.0001$ & $0.0002$ & $0.0002$ & $0.0003$ & $0.0001$ & $-0.0069$ \\
 & $0.0001$ & $0.0002$ & $0.0002$ & $0.0003$ & $0.0001$ & $-0.0069$ \\
 & $(0.1644)$ & $(0.1644)$ & $(0.1778)$ & $(0.2043)$ & $(0.1927)$ & $(0.1877)$ \\
$v_0$ & $0.3535$ & $0.3324$ & $0.3404$ & $0.3492$ & $0.3391$ & $0.3926$ \\
 & $-1.0398$ & $-1.1015$ & $-1.0777$ & $-1.0520$ & $-1.0814$ & $-0.9349$ \\
 & $(0.0730)$ & $(0.0756)$ & $(0.0754)$ & $(0.0828)$ & $(0.0754)$ & $(0.0788)$ \\
$\kappa$ & $0.7649$ & $0.7953$ & $0.7867$ & $0.8736$ & $0.8119$ & $0.8133$ \\
 & $-0.2680$ & $-0.2291$ & $-0.2399$ & $-0.1351$ & $-0.2084$ & $-0.2066$ \\
 & $(0.0532)$ & $(0.0554)$ & $(0.0551)$ & $(0.0593)$ & $(0.0581)$ & $(0.0523)$ \\
$\sigma$ & $0.1069$ & $0.1048$ & $0.0925$ & $0.0817$ & $0.0894$ & $0.1145$ \\
 & $-2.2356$ & $-2.2553$ & $-2.3804$ & $-2.5046$ & $-2.4148$ & $-2.1669$ \\
 & $(0.0928)$ & $(0.0940)$ & $(0.1028)$ & $(0.1176)$ & $(0.1027)$ & $(0.1402)$ \\
$\rho$ & $0.7025$ & $0.7222$ & $0.7615$ & $0.7514$ & $0.8206$ & $0.5079$ \\
 & $1.9821$ & $2.1439$ & $2.5432$ & $2.4291$ & $3.4533$ & $1.0252$ \\
 & $(0.4140)$ & $(0.4828)$ & $(0.6766)$ & $(0.6739)$ & $(1.0302)$ & $(0.1526)$ \\
$a$ & $0.0445$ & $0.0360$ & $0.0050$ & $0.0050$ & $0.0050$ & $0.0445$ \\
 & $-3.1120$ & $-3.3250$ & $-5.2983$ & $-5.2983$ & $-5.2983$ & $-3.1118$ \\
 & $(0.0312)$ & $(0.0312)$ & $(0.2761)$ & $(0.2974)$ & $(0.2455)$ & $(0.0336)$ \\
$b$ & $0.0385$ & $1.0461$ & $0.1643$ & $0.0810$ & $0.2785$ & - \\
 & $2.2162$ & $0.0451$ & $-1.8063$ & $-2.5135$ & $-1.2782$ & - \\
 & $(1.3002)$ & $(0.0623)$ & $(0.0358)$ & $(0.0361)$ & $(0.0359)$ & - \\
$t_0$ & $0.2927$ & $0.2793$ & $0.3051$ & $0.4251$ & $0.3344$ & - \\
 & $-0.7622$ & $-0.8308$ & $-0.7021$ & $-0.2397$ & $-0.5727$ & - \\
 & $(0.1933)$ & $(0.1943)$ & $(0.4005)$ & $(0.0612)$ & $(0.0486)$ & - \\
$h_1$ & $0.0118$ & $0.0117$ & $0.0117$ & $0.0117$ & $0.0117$ & $0.0117$ \\
 & $-4.4418$ & $-4.4516$ & $-4.4515$ & $-4.4473$ & $-4.4489$ & $-4.4486$ \\
 & $(0.0144)$ & $(0.0143)$ & $(0.0143)$ & $(0.0144)$ & $(0.0144)$ & $(0.0144)$ \\
$h_2$ & $0.0080$ & $0.0080$ & $0.0081$ & $0.0081$ & $0.0080$ & $0.0080$ \\
 & $-4.8259$ & $-4.8251$ & $-4.8175$ & $-4.8119$ & $-4.8230$ & $-4.8249$ \\
 & $(0.0147)$ & $(0.0147)$ & $(0.0148)$ & $(0.0148)$ & $(0.0147)$ & $(0.0147)$ \\
$h_3$ & $0.0074$ & $0.0074$ & $0.0075$ & $0.0075$ & $0.0074$ & $0.0074$ \\
 & $-4.9029$ & $-4.8997$ & $-4.8958$ & $-4.8909$ & $-4.9009$ & $-4.9072$ \\
 & $(0.0147)$ & $(0.0148)$ & $(0.0148)$ & $(0.0148)$ & $(0.0147)$ & $(0.0147)$ \\
$h_4$ & $0.0056$ & $0.0055$ & $0.0055$ & $0.0055$ & $0.0055$ & $0.0055$ \\
 & $-5.1916$ & $-5.2064$ & $-5.1959$ & $-5.1967$ & $-5.2005$ & $-5.2000$ \\
 & $(0.0153)$ & $(0.0152)$ & $(0.0153)$ & $(0.0152)$ & $(0.0152)$ & $(0.0153)$ \\
$h_5$ & $0.0043$ & $0.0043$ & $0.0043$ & $0.0043$ & $0.0043$ & $0.0043$ \\
 & $-5.4514$ & $-5.4475$ & $-5.4573$ & $-5.4497$ & $-5.4516$ & $-5.4524$ \\
 & $(0.0159)$ & $(0.0160)$ & $(0.0159)$ & $(0.0159)$ & $(0.0159)$ & $(0.0159)$ \\
$\dots$ & $0.0032$ & $0.0032$ & $0.0032$ & $0.0032$ & $0.0032$ & $0.0032$ \\
 & $-5.7421$ & $-5.7434$ & $-5.7466$ & $-5.7541$ & $-5.7476$ & $-5.7473$ \\
 & $(0.0175)$ & $(0.0174)$ & $(0.0175)$ & $(0.0174)$ & $(0.0174)$ & $(0.0174)$ \\
 & $0.0029$ & $0.0029$ & $0.0029$ & $0.0029$ & $0.0029$ & $0.0029$ \\
 & $-5.8404$ & $-5.8300$ & $-5.8397$ & $-5.8417$ & $-5.8400$ & $-5.8408$ \\
 & $(0.0179)$ & $(0.0179)$ & $(0.0179)$ & $(0.0179)$ & $(0.0179)$ & $(0.0178)$ \\
 & $0.0036$ & $0.0035$ & $0.0035$ & $0.0035$ & $0.0035$ & $0.0035$ \\
 & $-5.6388$ & $-5.6469$ & $-5.6515$ & $-5.6595$ & $-5.6510$ & $-5.6422$ \\
 & $(0.0163)$ & $(0.0162)$ & $(0.0162)$ & $(0.0162)$ & $(0.0162)$ & $(0.0162)$ \\
 & $0.0043$ & $0.0043$ & $0.0044$ & $0.0043$ & $0.0043$ & $0.0043$ \\
 & $-5.4548$ & $-5.4496$ & $-5.4374$ & $-5.4483$ & $-5.4442$ & $-5.4400$ \\
 & $(0.0152)$ & $(0.0153)$ & $(0.0154)$ & $(0.0153)$ & $(0.0153)$ & $(0.0153)$ \\
 & $0.0048$ & $0.0048$ & $0.0048$ & $0.0048$ & $0.0048$ & $0.0048$ \\
 & $-5.3385$ & $-5.3399$ & $-5.3390$ & $-5.3419$ & $-5.3383$ & $-5.3379$ \\
 & $(0.0150)$ & $(0.0150)$ & $(0.0150)$ & $(0.0150)$ & $(0.0150)$ & $(0.0150)$ \\
    \bottomrule
    \end{tabular}
\caption{Estimated parameters for Cotton. For each parameter: estimated value in the standard space, transformed in $\mathbb{R}$
and standard error between parenthesis (standard deviation of the estimate in $\mathbb{R}$).}
\label{tab:estimated_parameters_cotton}
\end{table}

\begin{table}[H]
  \centering
\footnotesize 
     \begin{tabular}{crrrrrr}
     \toprule
	&	Sinusoidal &	Exp-sinusoidal & Triangle &	Sawtooth & Spiked & Non-seasonal \\
    \midrule
$\lambda$ & $0.1318$ & $0.1325$ & $0.1322$ & $0.1325$ & $0.1342$ & $0.1319$ \\
 & $-2.0265$ & $-2.0213$ & $-2.0232$ & $-2.0215$ & $-2.0084$ & $-2.0260$ \\
 & $(0.0204)$ & $(0.0201)$ & $(0.0204)$ & $(0.0207)$ & $(0.0200)$ & $(0.0203)$ \\
$\pi^F$ & $2.3298$ & $2.3394$ & $2.3767$ & $2.3722$ & $2.4152$ & $2.4146$ \\
 & $2.3298$ & $2.3394$ & $2.3767$ & $2.3722$ & $2.4152$ & $2.4146$ \\
 & $(0.0693)$ & $(0.0629)$ & $(0.0571)$ & $(0.0718)$ & $(0.0942)$ & $(0.0905)$ \\
$v_0$ & $0.0996$ & $0.1011$ & $0.1000$ & $0.0949$ & $0.0998$ & $0.0983$ \\
 & $-2.3064$ & $-2.2912$ & $-2.3021$ & $-2.3550$ & $-2.3045$ & $-2.3197$ \\
 & $(0.0916)$ & $(0.0810)$ & $(0.1059)$ & $(0.1018)$ & $(0.0753)$ & $(0.0741)$ \\
$\kappa$ & $1.2492$ & $1.2604$ & $1.2148$ & $1.1479$ & $1.2370$ & $1.2092$ \\
 & $0.2225$ & $0.2315$ & $0.1946$ & $0.1379$ & $0.2127$ & $0.1899$ \\
 & $(0.0439)$ & $(0.0421)$ & $(0.0462)$ & $(0.0475)$ & $(0.0456)$ & $(0.0469)$ \\
$\sigma$ & $0.0722$ & $0.0755$ & $0.0774$ & $0.0482$ & $0.0754$ & $0.0472$ \\
 & $-2.6278$ & $-2.5830$ & $-2.5592$ & $-3.0319$ & $-2.5844$ & $-3.0537$ \\
 & $(0.2237)$ & $(0.2059)$ & $(0.5031)$ & $(0.5391)$ & $(0.2321)$ & $(0.4108)$ \\
$\rho$ & $0.3560$ & $0.3630$ & $0.2616$ & $0.3576$ & $0.3490$ & $0.3044$ \\
 & $0.6259$ & $0.6413$ & $0.4358$ & $0.6293$ & $0.6107$ & $0.5183$ \\
 & $(0.2808)$ & $(0.2623)$ & $(0.2111)$ & $(0.3636)$ & $(0.2495)$ & $(0.3691)$ \\
$a$ & $0.0328$ & $0.0262$ & $0.0113$ & $0.0111$ & $0.0080$ & $0.0324$ \\
 & $-3.4161$ & $-3.6403$ & $-4.4830$ & $-4.5008$ & $-4.8229$ & $-3.4311$ \\
 & $(0.0337)$ & $(0.0335)$ & $(0.1041)$ & $(0.1033)$ & $(0.1357)$ & $(0.0337)$ \\
$b$ & $0.0158$ & $0.9294$ & $0.1031$ & $0.0468$ & $0.1670$ & - \\
 & $-0.0595$ & $-0.0732$ & $-2.2722$ & $-3.0628$ & $-1.7896$ & - \\
 & $(1.2639)$ & $(0.0900)$ & $(0.0496)$ & $(0.0508)$ & $(0.0447)$ & - \\
$t_0$ & $0.4432$ & $0.4412$ & $0.4611$ & $0.6110$ & $0.4612$ & - \\
 & $-0.1803$ & $-0.1868$ & $-0.1228$ & $0.3636$ & $-0.1224$ & - \\
 & $(0.1177)$ & $(0.0795)$ & $(0.2176)$ & $(0.0016)$ & $(0.0894)$ & - \\
$h_1$ & $0.0082$ & $0.0082$ & $0.0082$ & $0.0082$ & $0.0081$ & $0.0078$ \\
 & $-4.8065$ & $-4.8034$ & $-4.8057$ & $-4.8046$ & $-4.8139$ & $-4.8542$ \\
 & $(0.0144)$ & $(0.0145)$ & $(0.0144)$ & $(0.0144)$ & $(0.0143)$ & $(0.0138)$ \\
$h_2$ & $0.0052$ & $0.0052$ & $0.0052$ & $0.0053$ & $0.0052$ & $0.0053$ \\
 & $-5.2612$ & $-5.2615$ & $-5.2613$ & $-5.2480$ & $-5.2519$ & $-5.2466$ \\
 & $(0.0141)$ & $(0.0141)$ & $(0.0141)$ & $(0.0142)$ & $(0.0143)$ & $(0.0143)$ \\
$h_3$ & $0.0042$ & $0.0042$ & $0.0042$ & $0.0043$ & $0.0043$ & $0.0042$ \\
 & $-5.4795$ & $-5.4731$ & $-5.4650$ & $-5.4580$ & $-5.4598$ & $-5.4655$ \\
 & $(0.0140)$ & $(0.0141)$ & $(0.0142)$ & $(0.0143)$ & $(0.0143)$ & $(0.0142)$ \\
$h_4$ & $0.0034$ & $0.0033$ & $0.0033$ & $0.0034$ & $0.0034$ & $0.0033$ \\
 & $-5.6864$ & $-5.7092$ & $-5.7025$ & $-5.6866$ & $-5.6926$ & $-5.7025$ \\
 & $(0.0145)$ & $(0.0142)$ & $(0.0142)$ & $(0.0145)$ & $(0.0146)$ & $(0.0144)$ \\
$h_5$ & $0.0026$ & $0.0025$ & $0.0025$ & $0.0026$ & $0.0025$ & $0.0025$ \\
 & $-5.9562$ & $-5.9794$ & $-5.9765$ & $-5.9526$ & $-6.0049$ & $-5.9877$ \\
 & $(0.0148)$ & $(0.0146)$ & $(0.0146)$ & $(0.0148)$ & $(0.0143)$ & $(0.0145)$ \\
$\dots$ & $0.0017$ & $0.0017$ & $0.0017$ & $0.0018$ & $0.0017$ & $0.0017$ \\
 & $-6.3640$ & $-6.3614$ & $-6.3615$ & $-6.3404$ & $-6.3535$ & $-6.3606$ \\
 & $(0.0155)$ & $(0.0156)$ & $(0.0155)$ & $(0.0157)$ & $(0.0158)$ & $(0.0156)$ \\
 & $0.0011$ & $0.0011$ & $0.0011$ & $0.0010$ & $0.0011$ & $0.0011$ \\
 & $-6.8551$ & $-6.8576$ & $-6.8548$ & $-6.8621$ & $-6.8447$ & $-6.8490$ \\
 & $(0.0193)$ & $(0.0193)$ & $(0.0193)$ & $(0.0193)$ & $(0.0194)$ & $(0.0193)$ \\
 & $0.0012$ & $0.0012$ & $0.0012$ & $0.0012$ & $0.0012$ & $0.0012$ \\
 & $-6.7450$ & $-6.7503$ & $-6.7512$ & $-6.7544$ & $-6.7433$ & $-6.7295$ \\
 & $(0.0178)$ & $(0.0178)$ & $(0.0178)$ & $(0.0179)$ & $(0.0178)$ & $(0.0178)$ \\
 & $0.0015$ & $0.0015$ & $0.0015$ & $0.0015$ & $0.0016$ & $0.0015$ \\
 & $-6.4903$ & $-6.4851$ & $-6.4874$ & $-6.4919$ & $-6.4665$ & $-6.4802$ \\
 & $(0.0158)$ & $(0.0158)$ & $(0.0158)$ & $(0.0158)$ & $(0.0160)$ & $(0.0159)$ \\
 & $0.0021$ & $0.0021$ & $0.0021$ & $0.0021$ & $0.0021$ & $0.0021$ \\
 & $-6.1731$ & $-6.1736$ & $-6.1795$ & $-6.1711$ & $-6.1699$ & $-6.1837$ \\
 & $(0.0152)$ & $(0.0152)$ & $(0.0151)$ & $(0.0153)$ & $(0.0152)$ & $(0.0151)$ \\
%
    \bottomrule
    \end{tabular}
\caption{Estimated parameters for Soybeans. For each parameter: estimated value in the standard space, transformed in $\mathbb{R}$
and standard error between parenthesis (standard deviation of the estimate in $\mathbb{R}$). Results for $h_{11},h_{12}$ and $h_{13}$ are not reported in the table. They are similar to the results obtained for $h_1,\dots,h_{10}$ and available upon request.}
\label{tab:estimated_parameters_soybeans}
\end{table}

\begin{table}[H]
  \centering
\footnotesize 
     \begin{tabular}{crrrrrr}
     \toprule
	&	Sinusoidal &	Exp-sinusoidal & Triangle &	Sawtooth & Spiked & Non-seasonal \\
    \midrule
$\lambda$ & $0.2794$ & $0.2818$ & $0.2782$ & $0.2835$ & $0.2794$ & $0.2803$ \\
 & $-1.2751$ & $-1.2667$ & $-1.2794$ & $-1.2607$ & $-1.2752$ & $-1.2718$ \\
 & $(0.0173)$ & $(0.0172)$ & $(0.0176)$ & $(0.0171)$ & $(0.0173)$ & $(0.0170)$ \\
$\pi^F$ & $0.8002$ & $0.4171$ & $0.3216$ & $0.4495$ & $0.4140$ & $0.2977$ \\
 & $0.8002$ & $0.4171$ & $0.3216$ & $0.4495$ & $0.4140$ & $0.2977$ \\
 & $(0.1410)$ & $(0.1292)$ & $(0.1229)$ & $(0.1321)$ & $(0.1411)$ & $(0.1524)$ \\
$v_0$ & $0.1009$ & $0.1854$ & $0.1884$ & $0.1964$ & $0.2147$ & $0.2889$ \\
 & $-2.2940$ & $-1.6851$ & $-1.6691$ & $-1.6276$ & $-1.5386$ & $-1.2417$ \\
 & $(0.5087)$ & $(0.1744)$ & $(0.1545)$ & $(0.1747)$ & $(0.1653)$ & $(0.1547)$ \\
$\kappa$ & $1.0072$ & $0.8496$ & $0.8528$ & $0.8530$ & $0.9226$ & $1.1732$ \\
 & $0.0072$ & $-0.1630$ & $-0.1592$ & $-0.1590$ & $-0.0806$ & $0.1597$ \\
 & $(0.0816)$ & $(0.0736)$ & $(0.0745)$ & $(0.0703)$ & $(0.0729)$ & $(0.0680)$ \\
$\sigma$ & $0.1600$ & $0.1344$ & $0.1423$ & $0.1797$ & $0.1763$ & $0.2723$ \\
 & $-1.8326$ & $-2.0070$ & $-1.9495$ & $-1.7167$ & $-1.7358$ & $-1.3010$ \\
 & $(0.0751)$ & $(0.0867)$ & $(0.0880)$ & $(0.0768)$ & $(0.0882)$ & $(0.0588)$ \\
$\rho$ & $0.8012$ & $0.7758$ & $0.7578$ & $0.5955$ & $0.5807$ & $0.5540$ \\
 & $3.0970$ & $2.7207$ & $2.4999$ & $1.3561$ & $1.2922$ & $1.1859$ \\
 & $(0.9103)$ & $(0.7390)$ & $(0.5857)$ & $(0.2113)$ & $(0.2040)$ & $(0.1436)$ \\
$a$ & $0.0996$ & $0.0837$ & $0.0615$ & $0.0430$ & $0.0596$ & $0.0958$ \\
 & $-2.3071$ & $-2.4808$ & $-2.7884$ & $-3.1470$ & $-2.8199$ & $-2.3457$ \\
 & $(0.0292)$ & $(0.0304)$ & $(0.0453)$ & $(0.0635)$ & $(0.0451)$ & $(0.0274)$ \\
$b$ & $0.0373$ & $0.5987$ & $0.1213$ & $0.0959$ & $0.1977$ & - \\
 & $-0.4164$ & $-0.5129$ & $-2.1099$ & $-2.3450$ & $-1.6211$ & - \\
 & $(0.7464)$ & $(0.1824)$ & $(0.0910)$ & $(0.0571)$ & $(0.0933)$ & - \\
$t_0$ & $0.4811$ & $0.6727$ & $0.6457$ & $0.8072$ & $0.6669$ & - \\
 & $-0.0595$ & $0.6027$ & $0.4928$ & $1.4438$ & $0.5783$ & - \\
 & $(0.4209)$ & $(0.2814)$ & $(0.6122)$ & $(0.2812)$ & $(0.1866)$ & - \\
$h_1$ & $0.0108$ & $0.0109$ & $0.0110$ & $0.0107$ & $0.0108$ & $0.0108$ \\
 & $-4.5252$ & $-4.5213$ & $-4.5128$ & $-4.5416$ & $-4.5253$ & $-4.5316$ \\
 & $(0.0145)$ & $(0.0145)$ & $(0.0146)$ & $(0.0144)$ & $(0.0145)$ & $(0.0144)$ \\
$h_2$ & $0.0073$ & $0.0074$ & $0.0074$ & $0.0072$ & $0.0073$ & $0.0073$ \\
 & $-4.9157$ & $-4.9100$ & $-4.9107$ & $-4.9318$ & $-4.9175$ & $-4.9144$ \\
 & $(0.0148)$ & $(0.0148)$ & $(0.0148)$ & $(0.0148)$ & $(0.0148)$ & $(0.0148)$ \\
$h_3$ & $0.0050$ & $0.0050$ & $0.0050$ & $0.0048$ & $0.0049$ & $0.0050$ \\
 & $-5.3079$ & $-5.2937$ & $-5.2914$ & $-5.3339$ & $-5.3094$ & $-5.3006$ \\
 & $(0.0156)$ & $(0.0158)$ & $(0.0158)$ & $(0.0157)$ & $(0.0157)$ & $(0.0158)$ \\
$h_4$ & $0.0032$ & $0.0032$ & $0.0032$ & $0.0032$ & $0.0032$ & $0.0032$ \\
 & $-5.7444$ & $-5.7371$ & $-5.7515$ & $-5.7553$ & $-5.7416$ & $-5.7530$ \\
 & $(0.0180)$ & $(0.0180)$ & $(0.0179)$ & $(0.0183)$ & $(0.0181)$ & $(0.0180)$ \\
$h_5$ & $0.0025$ & $0.0025$ & $0.0025$ & $0.0026$ & $0.0025$ & $0.0025$ \\
 & $-5.9922$ & $-5.9836$ & $-5.9814$ & $-5.9482$ & $-5.9726$ & $-5.9768$ \\
 & $(0.0203)$ & $(0.0203)$ & $(0.0203)$ & $(0.0199)$ & $(0.0202)$ & $(0.0202)$ \\
$\dots$ & $0.0030$ & $0.0031$ & $0.0030$ & $0.0031$ & $0.0031$ & $0.0030$ \\
 & $-5.7933$ & $-5.7857$ & $-5.8021$ & $-5.7607$ & $-5.7849$ & $-5.7942$ \\
 & $(0.0172)$ & $(0.0173)$ & $(0.0172)$ & $(0.0171)$ & $(0.0172)$ & $(0.0172)$ \\
 & $0.0041$ & $0.0041$ & $0.0041$ & $0.0041$ & $0.0041$ & $0.0041$ \\
 & $-5.4988$ & $-5.5054$ & $-5.5090$ & $-5.4907$ & $-5.5018$ & $-5.5077$ \\
 & $(0.0154)$ & $(0.0154)$ & $(0.0154)$ & $(0.0154)$ & $(0.0154)$ & $(0.0153)$ \\
    \bottomrule
    \end{tabular}
\caption{Estimated parameters for Sugar. For each parameter: estimated value in the standard space, transformed in $\mathbb{R}$
and standard error between parenthesis (standard deviation of the estimate in $\mathbb{R}$).}
\label{tab:estimated_parameters_sugar}
\end{table}

\begin{table}[H]
  \centering
\footnotesize 
     \begin{tabular}{crrrrrr}
     \toprule
	&	Sinusoidal &	Exp-sinusoidal & Triangle &	Sawtooth & Spiked & Non-seasonal \\
    \midrule
$\lambda$ & $0.2126$ & $0.2126$ & $0.2134$ & $0.2131$ & $0.2118$ & $0.2124$ \\
 & $-1.5481$ & $-1.5481$ & $-1.5446$ & $-1.5461$ & $-1.5523$ & $-1.5492$ \\
 & $(0.0132)$ & $(0.0132)$ & $(0.0131)$ & $(0.0132)$ & $(0.0133)$ & $(0.0132)$ \\
$\pi^F$ & $3.3976$ & $3.4084$ & $3.4329$ & $3.4834$ & $3.4930$ & $3.5206$ \\
 & $3.3976$ & $3.4084$ & $3.4329$ & $3.4834$ & $3.4930$ & $3.5206$ \\
 & $(0.0774)$ & $(0.0745)$ & $(0.0777)$ & $(0.0691)$ & $(0.0704)$ & $(0.0629)$ \\
$v_0$ & $0.1191$ & $0.1226$ & $0.1546$ & $0.1189$ & $0.1554$ & $0.2000$ \\
 & $-2.1277$ & $-2.0985$ & $-1.8670$ & $-2.1296$ & $-1.8617$ & $-1.6094$ \\
 & $(0.3249)$ & $(0.3054)$ & $(0.2577)$ & $(0.3299)$ & $(0.2492)$ & $(0.1880)$ \\
$\kappa$ & $0.4051$ & $0.3886$ & $0.4113$ & $0.3587$ & $0.3660$ & $0.3199$ \\
 & $-0.9037$ & $-0.9452$ & $-0.8884$ & $-1.0253$ & $-1.0052$ & $-1.1397$ \\
 & $(0.1478)$ & $(0.1538)$ & $(0.1499)$ & $(0.1365)$ & $(0.1528)$ & $(0.1239)$ \\
$\sigma$ & $0.0381$ & $0.0338$ & $0.0363$ & $0.0409$ & $0.0394$ & $0.0338$ \\
 & $-3.2667$ & $-3.3878$ & $-3.3163$ & $-3.1971$ & $-3.2342$ & $-3.3861$ \\
 & $(0.1192)$ & $(0.1359)$ & $(0.1320)$ & $(0.1068)$ & $(0.1214)$ & $(0.1096)$ \\
$\rho$ & $0.5622$ & $0.5583$ & $0.5296$ & $0.5096$ & $0.4950$ & $0.5466$ \\
 & $1.2175$ & $1.2024$ & $1.0975$ & $1.0307$ & $0.9843$ & $1.1583$ \\
 & $(0.3728)$ & $(0.3845)$ & $(0.3409)$ & $(0.2829)$ & $(0.2616)$ & $(0.3970)$ \\
$a$ & $0.0837$ & $0.0779$ & $0.0457$ & $0.0471$ & $0.0751$ & $0.0825$ \\
 & $-2.4810$ & $-2.5529$ & $-3.0865$ & $-3.0551$ & $-2.5893$ & $-2.4944$ \\
 & $(0.0362)$ & $(0.0373)$ & $(0.0662)$ & $(0.0683)$ & $(0.0431)$ & $(0.0379)$ \\
$b$ & $0.0634$ & $0.4556$ & $0.1459$ & $0.0812$ & $0.0798$ & - \\
 & $1.0543$ & $-0.7861$ & $-1.9246$ & $-2.5105$ & $-2.5281$ & - \\
 & $(1.2316)$ & $(0.3725)$ & $(0.0822)$ & $(0.0789)$ & $(0.2559)$ & - \\
$t_0$ & $0.3173$ & $0.3203$ & $0.3155$ & $0.3567$ & $0.3080$ & - \\
 & $-0.6465$ & $-0.6334$ & $-0.6547$ & $-0.4833$ & $-0.6890$ & - \\
 & $(0.2457)$ & $(0.3239)$ & $(0.4992)$ & $(0.1535)$ & $(0.3349)$ & - \\
$h_1$ & $0.0064$ & $0.0064$ & $0.0064$ & $0.0064$ & $0.0064$ & $0.0064$ \\
 & $-5.0450$ & $-5.0475$ & $-5.0455$ & $-5.0443$ & $-5.0507$ & $-5.0517$ \\
 & $(0.0143)$ & $(0.0142)$ & $(0.0143)$ & $(0.0143)$ & $(0.0142)$ & $(0.0141)$ \\
$h_2$ & $0.0047$ & $0.0048$ & $0.0047$ & $0.0047$ & $0.0048$ & $0.0047$ \\
 & $-5.3509$ & $-5.3464$ & $-5.3509$ & $-5.3528$ & $-5.3387$ & $-5.3496$ \\
 & $(0.0144)$ & $(0.0145)$ & $(0.0144)$ & $(0.0144)$ & $(0.0145)$ & $(0.0144)$ \\
$h_3$ & $0.0033$ & $0.0033$ & $0.0033$ & $0.0033$ & $0.0033$ & $0.0033$ \\
 & $-5.7060$ & $-5.7076$ & $-5.7138$ & $-5.7099$ & $-5.7034$ & $-5.7109$ \\
 & $(0.0151)$ & $(0.0151)$ & $(0.0150)$ & $(0.0151)$ & $(0.0151)$ & $(0.0150)$ \\
$h_4$ & $0.0019$ & $0.0019$ & $0.0019$ & $0.0019$ & $0.0019$ & $0.0019$ \\
 & $-6.2445$ & $-6.2442$ & $-6.2434$ & $-6.2428$ & $-6.2444$ & $-6.2427$ \\
 & $(0.0172)$ & $(0.0172)$ & $(0.0173)$ & $(0.0173)$ & $(0.0172)$ & $(0.0173)$ \\
$h_5$ & $0.0014$ & $0.0014$ & $0.0014$ & $0.0013$ & $0.0014$ & $0.0013$ \\
 & $-6.6055$ & $-6.6044$ & $-6.5981$ & $-6.6096$ & $-6.6034$ & $-6.6095$ \\
 & $(0.0216)$ & $(0.0216)$ & $(0.0215)$ & $(0.0216)$ & $(0.0216)$ & $(0.0216)$ \\
$\dots$ & $0.0018$ & $0.0018$ & $0.0018$ & $0.0018$ & $0.0018$ & $0.0018$ \\
 & $-6.3412$ & $-6.3449$ & $-6.3365$ & $-6.3361$ & $-6.3417$ & $-6.3374$ \\
 & $(0.0173)$ & $(0.0173)$ & $(0.0173)$ & $(0.0173)$ & $(0.0173)$ & $(0.0173)$ \\
 & $0.0028$ & $0.0028$ & $0.0028$ & $0.0028$ & $0.0028$ & $0.0028$ \\
 & $-5.8779$ & $-5.8787$ & $-5.8835$ & $-5.8858$ & $-5.8863$ & $-5.8789$ \\
 & $(0.0150)$ & $(0.0150)$ & $(0.0150)$ & $(0.0149)$ & $(0.0149)$ & $(0.0150)$ \\
 & $0.0035$ & $0.0034$ & $0.0035$ & $0.0035$ & $0.0035$ & $0.0035$ \\
 & $-5.6657$ & $-5.6699$ & $-5.6690$ & $-5.6635$ & $-5.6665$ & $-5.6678$ \\
 & $(0.0147)$ & $(0.0146)$ & $(0.0146)$ & $(0.0147)$ & $(0.0147)$ & $(0.0146)$ \\
 & $0.0040$ & $0.0041$ & $0.0040$ & $0.0041$ & $0.0040$ & $0.0040$ \\
 & $-5.5104$ & $-5.5077$ & $-5.5163$ & $-5.5075$ & $-5.5122$ & $-5.5113$ \\
 & $(0.0144)$ & $(0.0145)$ & $(0.0144)$ & $(0.0145)$ & $(0.0144)$ & $(0.0144)$ \\
 & $0.0066$ & $0.0066$ & $0.0066$ & $0.0066$ & $0.0066$ & $0.0066$ \\
 & $-5.0195$ & $-5.0191$ & $-5.0214$ & $-5.0197$ & $-5.0198$ & $-5.0241$ \\
 & $(0.0142)$ & $(0.0142)$ & $(0.0142)$ & $(0.0142)$ & $(0.0142)$ & $(0.0141)$ \\
    \bottomrule
    \end{tabular}
\caption{Estimated parameters for Wheat. For each parameter: estimated value in the standard space, transformed in $\mathbb{R}$
and standard error between parenthesis (standard deviation of the estimate in $\mathbb{R}$).}
\label{tab:estimated_parameters_wheat}
\end{table}

\subsection{An Alternative Estimation Procedure}
\label{ss:AlternativeEstimationProcedure}



In economic terms, using futures returns as we have done corresponds to the feasible strategy of holding each contract until its expiry and booking the daily returns.
However, this is strictly speaking not a ``constant-maturity'' series, as the actual time to expiry of a maturity series in our sample varies between two roll dates.
For example, on the roll date 17 July 2017, the new nearby corn contract is $60$ calendar days away from its maturity on 15 September 2017;
from then on, this time to maturity decreases with every day, until the contract expires.

An alternative is to construct an investable constant-maturity series via linear interpolation
between the returns of two contracts with times to expiry on either side of the constant maturity.
This approach has been proposed by \citet{Galai1979} for building an index for call options,
and by \citet{AlexanderKorovilas2013} in the context of VIX futures.
In order to test the robustness of our results, we have carried out an identical estimation procedure to the one described above in \ref{ss:MaximumLikelihoodEstimation},
based on constant-maturity returns calculated from our futures data using the equations (1) and (2) of \citet{AlexanderKorovilas2013}.

The new estimation results obtained with this alternative datasets are presented in Table \ref{tab:altdata_summary_results} in Appendix \ref{a:altdata_results}.
This table gathers, for each commodity and model, the obtained log-likelihood, the AIC, BIC, the value taken by the statistic $D1$ of the first likelihood ratio test,
its $p$-value as well as the AIC difference and Akaike weight.
In addition, it provides the log-likelihood obtained after estimating the nested versions of the considered models obtained when setting $\lambda=0$,
as well as the value taken by the statistic $D2$ of the second likelihood ratio test and its $p$-value.

In Table \ref{tab:altdata_models_ranking} in Appendix \ref{a:altdata_results}, we provide the model rankings for each commodity.
These rankings are based on the AIC differences described and discussed in Section \ref{ss:MaximumLikelihoodEstimation}.
Note that the likelihood scores obtained when estimating the models with the alternative dataset are higher for corn and sugar,
and lower for cotton, soybeans and wheat.

The main conclusions obtained above are confirmed by the results obtained with the alternative dataset.
For each commodity, the five seasonal models perform better than the non-seasonal one.
The associated likelihood ratio tests confirm that the need to model seasonality in volatility is statistically significant.
More specifically, the values of the statistic $D1$ are bigger than those obtained with the initial dataset.
In addition, the $\Delta_{aic}$ between the best seasonal model and the non-seasonal model are clearly higher than those obtained with the initial dataset.
These two remarks indicate that the alternative dataset contains a greater statistical evidence in favour of the seasonality of volatility.

In addition, the estimation with the alternative datasets clearly confirms the conclusion of a strong statistical evidence in favour of the Samuelson effect.

The model rankings obtained with the alternative dataset confirm the main conclusions obtained above.
The exp-sinusoidal specification is always ranked first or second and the non-seasonal one is always the last.
Importantly, the first and last ranked models remain the same for all five commodities.
The second best model also remains the same, except for sugar.

For corn, cotton, soybeans, and wheat, the new values obtained for the parameter $\lambda$ are higher than those obtained with the initial dataset.
In line with the intuition one can have that the Samuelson effect will be stronger in the alternative datasets.
For sugar, the new values of $\lambda$ are slightly lower than with the initial dataset.
We do not report the detailed list of estimated parameter values obtained with the alternative dataset for the sake of brevity, but it is available upon request.

The results obtained with the alternative dataset clearly confirm the results presented above;
in fact, the statistical evidence obtained in favour of a seasonal volatility is even stronger than with the initial dataset.
We therefore conclude that our results are robust with respect to which type of time series - concatenated or constant-maturity - we use.



\section{Conclusion}
\label{s:Conclusion}

We introduce a multi-factor stochastic volatility model for futures contracts that is capable of reproducing seasonality and the Samuelson effect simultaneously.
The model can accommodate general specifications of the seasonality functions, including piece-wise continuous ones.
We suggest five different seasonality functions, some of which are familiar from the literature, and provide details of how to incorporate these into the model.
Our model can be used both for option pricing in the risk-neutral measure, and filtering, smoothing and predicting via the Kalman filter in the physical measure.
In an empirical section, we estimate our model for time-series data from corn, cotton, soybean, sugar and wheat futures markets.
Our results confirm that it is important to account for seasonality in the volatility as well as for the Samuelson effect.

A seasonal pattern that seems to fit agricultural markets well in most cases is the exponential-sinusoidal one;
this pattern also shows robust behaviour during parameter estimation.
Some non-agricultural commodity markets also display seasonality in futures prices and their volatility:
natural gas, gasoline, heating oil and fuel oil futures markets.
Seasonality patterns in these markets, which can also be viewed through the term-structure of implied volatilities in option markets,
can be of a different type, and we conjecture that other seasonality patterns than the exponential-sinusoidal one could be better adapted to these markets.
This, however, is a question for future research.


\pagebreak

\appendix


\section{Proof of Proposition \ref{Prop:CIR_SDE_Solution}}
\label{a:ProofCIR}

In this appendix we prove Proposition \ref{Prop:CIR_SDE_Solution}.

\begin{MyProof}{of Proposition \ref{Prop:CIR_SDE_Solution}.}
\begin{enumerate}
\item
The drift function $b(t, v_t) := \kappa (\theta(t) - v(t))$ in \eqref{CIR_SDE_TimeDependentTheta} is Lipschitz continuous w.r.t the second argument, i.e.
$$
| b(t, x) - b(t, y) | \leq K | x - y |,
$$
where we can choose $K = \kappa$, since $| b(t, x) - b(t, y) | = \kappa | x - y |$.
Then Proposition 5.2.13 (Yamada and Watanabe) of \citet{KaratzasShreve1988} guarantees the existence of a unique strong solution to \eqref{CIR_SDE_TimeDependentTheta}
with continuous sample paths.
\item
The comparison result given in Proposition 5.2.18 of \citet{KaratzasShreve1988} establishes $v_t \geq \tilde{v}_t$ a.s. for all $t \geq 0$
under the hypothesis that the drift function $b(t, v_t)$ is continuous.
Now, if $\theta$ has a discontinuity at time $t_1$, we know from this argument applied to the interval $[0, t_1[$ that $\tilde{v}_{t} \leq v_{t} \forall t \in [0, t_1[$ (a.s.).
It then follows from the continuity of the sample paths that $\tilde{v}_{t_1} \leq v_{t_1}$ (a.s.),
and we can apply the argument again to the interval $]t_1, t_2[$ to obtain $\tilde{v}_{t} \leq v_{t} \forall t \in ]t_1, t_2[$ (a.s.).
Since by assumption the set $\mathcal{T}$ of times where $\theta$ has discontinuities has no limit points,
we can proceed in this manner to cover all of $\mathbb{R}_0^+$.
\item
The Feller condition $\sigma^2 < 2 \kappa \theta_{min}$ for $\theta_{min}$ implies the strict positivity a.s. of $\tilde{v}$.
The strict positivity of $v$ itself therefore follows immediately from (ii).
\end{enumerate}
\end{MyProof}


\section{Proofs of Proposition \ref{Prop:JointCharacteristicFunction}}
\label{a:ProofCF}

In this appendix we prove Proposition \ref{Prop:JointCharacteristicFunction}.

\begin{MyProof}{of Proposition \ref{Prop:JointCharacteristicFunction}.}
The proof is an extension of the proof of Proposition 2.1 of \citet{SchneiderTavin2018} to the case where the variance mean-reversion level $\theta$ is time-dependent.
Going from $\theta$ to $\theta(t)$ leads to changes in two places.
The first is in Lemma A.1 of \citet{SchneiderTavin2018}, which needs to be modified as follows.
\begin{lemma}
\label{Lemma:StochasticIntegral}
Let $\theta: \mathbb{R}_0^+ \to \mathbb{R}^+$ be the seasonal mean-reversion level function,
and let
$$
\hat{\theta}_T(\lambda) := \int_0^T e^{\lambda t} \theta(t) dt
$$
be its transform. Then
\begin{equation}
\label{StochasticIntegral}
\sigma \int_0^T f_1(t) \sqrt{v(t)} d\tilde{B}(t)
=
\left[ f_1(t) v(t) \right]_0^T - f_1(0) \kappa \hat{\theta}_T(\lambda) + (\kappa - \lambda) \int_0^T f_1(t) v(t) dt.
\end{equation}
\end{lemma}
\begin{proof}
Multiplying equation \eqref{VarianceSDE_Q} by $f_1(t)$ and then integrating from $0$ to $T$ gives
\begin{equation}
\label{dv-integral1}
\int_0^T f_1(t) dv(t) = \int_0^T f_1(t) \kappa (\theta(t) - v(t)) dt + \sigma \int_0^T f_1(t) \sqrt{v(t)} d\tilde{B}(t).
\end{equation}
Using It\^{o}-integration by parts (see \cite{Oksendal2003}), we also have
\begin{align}
\int_0^T f_1(t) dv(t)
&= \left[ f_1(t) v(t) \right]_0^T - \int_0^T v(t) df_1(t)
\nonumber
\\
&= \left[ f_1(t) v(t) \right]_0^T - \lambda \int_0^T f_1(t) v(t) dt.
\label{dv-integral2}
\end{align}
Equating the right hand sides of equations \eqref{dv-integral1} and \eqref{dv-integral2} gives
\begin{align*}
\sigma \int_0^T f_1(t) \sqrt{v(t)} d\tilde{B}(t)
&= \left[ f_1(t) v(t) \right]_0^T - \lambda \int_0^T f_1(t) v(t) dt - \int_0^T f_1(t) \kappa (\theta(t) - v(t)) dt
\\
&= \left[ f_1(t) v(t) \right]_0^T - \kappa \int_0^T f_1(t) \theta(t) dt + (\kappa - \lambda) \int_0^T f_1(t) v(t) dt
\\
&= \left[ f_1(t) v(t) \right]_0^T - f_1(0) \kappa \int_0^T e^{\lambda t} \theta(t) dt + (\kappa - \lambda) \int_0^T f_1(t) v(t) dt
\\
&= \left[ f_1(t) v(t) \right]_0^T - f_1(0) \kappa \hat{\theta}_T(\lambda) + (\kappa - \lambda) \int_0^T f_1(t) v(t) dt,
\end{align*}
which proves the lemma.
\end{proof}

The second change in the proof is due to the appearance of $\theta$ in the generator of the process $v$.
As in \citet{SchneiderTavin2018}, let the function $h$ be given by
\begin{equation*}
h(t,v) = {\mathbb E} \left[ \exp \left( i \frac{\rho}{\sigma} f_1(T) v(T) + \int_t^T q(s) v(s) ds \right) \right].
\end{equation*}
Now $h$ satisfies the PDE
\begin{equation}
\label{hPDE}
\frac{\partial h}{\partial t} (t,v)
+ \kappa (\theta(t) - v(t)) \frac{\partial h}{\partial v} (t,v)
+ \frac{1}{2} \sigma^2 v(t) \frac{\partial^2 h}{\partial v^2} (t,v)
+ q(t) v(t) h(t,v)
= 0,
\end{equation}
with terminal condition
\begin{equation*}
 h(T,v) = \exp \left( i \frac{\rho}{\sigma} f_1(T) v(T) \right).
\end{equation*}

Again, we know from \citet{DuffiePanSingleton2000} that $h$ has affine form
\begin{equation}
\label{hGuess}
h(t,v) = \exp \left( A(t,T) v(t) + B(t,T) \right),
\end{equation}
with $A(T,T) = i \frac{\rho}{\sigma} f_1(T), B(T,T) = 0.$
Putting \eqref{hGuess} in \eqref{hPDE} gives
\begin{equation*}
B_t + A_t v + \kappa (\theta(t) - v) A + \frac{1}{2} \sigma^2 v A^2 + q v = 0,
\end{equation*}
and collecting the terms with and without $v$ leads to the two ODEs
\begin{align}
\label{A_RiccatiEquation}
A_t - \kappa A + \frac{1}{2} \sigma^2 A^2 + q &= 0,
\\
\label{A_Primitive}
B_t + \kappa \theta(t) A &= 0.
\end{align}
This completes the proof of the proposition.
\end{MyProof}

Note that $\theta$ only appears in the second ODE \eqref{A_Primitive},
and that therefore the closed-form expression previously given for $A$ in \citet{SchneiderTavin2018} can still be used.
Only the function $B$ changes due to the time-dependence of $\theta$.


\section{Transforms of the Seasonality Functions}
\label{a:Transforms}

In this appendix we show how the integral transforms $\hat{\theta}_T(\lambda)$
can be calculated for the sinusoidal \eqref{pattern:sinusoidal}, sawtooth \eqref{pattern:sawtooth} and triangle \eqref{pattern:triangle} patterns.
To the best of our knowledge, there are no closed-form expressions for the transforms of the exponential-sinusoidal \eqref{pattern:exponential-sinusoidal}
and spiked \eqref{pattern:spiked} patterns.

\begin{proposition}
\label{Prop:TransformsSeasonalityFunctions}
The transform of the sinusoidal pattern \eqref{pattern:sinusoidal} is given by
\begin{align}
\hat{\theta}_T(\lambda)
&= \frac{b e^{\lambda T}}{\lambda^2+4\pi^2}\left(2 \pi \sin {\left(2 \pi (T-t_0)\right)} + \lambda \cos{\left(2 \pi (T-t_0)\right)} \right) \nonumber \\
& + \frac{b}{\lambda^2+4\pi^2}\left(2 \pi \sin{\left(2\pi t_0 \right) - \lambda \cos {\left(2\pi t_0 \right)}}  \right) + \frac{a}{\lambda}\left(e^{\lambda T}-1\right).
\label{transform:sinusoidal}
\end{align}

The transform of the sawtooth pattern pattern \eqref{pattern:sawtooth} is given by
\begin{align}
\hat{\theta}_T(\lambda)
&= \frac{1}{\lambda}\left(a+b\left(\frac{1}{\lambda}-t_0 \right) \right)-\frac{e^{-\lambda T}}{\lambda}\left(a+b\left(T+\frac{1}{\lambda}-t_0 \right) \right) \nonumber \\
& -\frac{be^{\lambda t_0}}{\lambda} \left(\left\lfloor T-t_0 \right\rfloor e^{\lambda(T-t_0)}
- \left(\sum^{\left\lfloor T-t_0 \right\rfloor}_{k=1}{e^{\lambda k}}\right)\mathbb{I}_{\left\{ T \geq t_0\right\}}
+ \mathbb{I}_{\left\{ T<t_0\right\}}+e^{-\lambda t_0}-1 \right),
\label{transform:sawtooth}
\end{align}
where $\left\lfloor . \right\rfloor$ denotes the floor function, $\mathbb{I}$ is the indicator function and,
by convention, $\sum^{p}_{k=1}{e^{\lambda k}}=0$ if $p<1$.

The transform of the triangle pattern pattern \eqref{pattern:triangle} is given by
\begin{align}
& \hat{\theta}_T(\lambda)
= \frac{a}{\lambda}\left(e^{\lambda T}-1 \right) + \frac{b e^{\lambda t_0}}{\lambda}\left[ \left(z_2 + \left(\frac{2}{\lambda}e^{-\frac{\lambda}{2}}+e^{-\lambda t_0}(z_2-t_0) \right)\mathbb{I}_{\left\{ t_0 > \frac{1}{2} \right\}} - e^{-\lambda t_0}(z_2-t_0)\mathbb{I}_{\left\{ t_0 \leq \frac{1}{2} \right\}} \right)\right. \nonumber \\
& + \left(\left(\frac{2}{\lambda} e^{\frac{\lambda}{2}} + z_2 e^{\lambda}-z_1 \right)\sum^{n-1}_{k=0}{e^{\lambda k}}+ e^{\lambda n}\left(\left(\frac{2}{\lambda}e^{\frac{\lambda}{2}}-z_3 e^{\lambda \alpha} \right)\mathbb{I}_{\left\{ \alpha > \frac{1}{2} \right\}} + z_3 e^{\lambda \alpha}\mathbb{I}_{\left\{ \alpha \leq \frac{1}{2} \right\}} -z_1 \right) \right) \mathbb{I}_{\left\{ T \geq t_0\right\}} \nonumber \\
& + \left(e^{\lambda(T-t_0)}\left(z_2+T-t_0 \right)-z_2 \right)\mathbb{I}_{\left\{T-t_0 \in [-\frac{1}{2},0[ \right\}} \nonumber \\
& - \left. \left(\frac{2}{\lambda}e^{-\frac{\lambda}{2}} + e^{\lambda(T-t_0)}\left(z_2+T-t_0 \right) + z_2 \right)\mathbb{I}_{\left\{T-t_0 \in [-1,-\frac{1}{2}[ \right\}} \right],
\label{transform:triangle}
\end{align}
with $n = \left\lfloor T-t_0 \right\rfloor$, $\alpha = T- t_0 - \left\lfloor T-t_0 \right\rfloor$,
$z_1 = \frac{1}{2} + \frac{1}{\lambda}$, $z_2 = \frac{1}{2} - \frac{1}{\lambda}$ and $z_3 = z_1 - \alpha$,
and with the convention $\sum^{p}_{k=0}{e^{\lambda k}}=0$ if $p<0$.
\end{proposition}

\begin{proof}
The proof of \eqref{transform:sinusoidal} is straightforward.
The proofs of \eqref{transform:sawtooth} and \eqref{transform:triangle} are tedious and omitted here for brevity, but available from the authors upon request.
\end{proof}


\section{Estimation Results with Alternative Dataset}
\label{a:altdata_results}

\begin{table}[H]
  \centering
\footnotesize 
     \begin{tabular}{lrrrrrr}
     \toprule
	& Sinusoidal & Exp-sinusoidal & Triangle & Sawtooth & Spiked & Non-seasonal \\
    \midrule
\multicolumn{4}{l}{\textit{Corn:}} \\
LL & $103307.96$ & $103319.56$ & $103302.36$ & $103298.78$ & $103319.26$ & $103279.41$ \\
AIC & $-206577.92$ & $-206601.11$ & $-206566.71$ & $-206559.57$ & $-206600.53$ & $-206524.83$ \\
BIC & $-206467.04$ & $-206490.23$ & $-206455.83$ & $-206448.68$ & $-206489.64$ & $-206425.62$ \\
D1 (LR test 1) & $57.09$ & $80.28$ & $45.88$ & $38.74$ & $79.7$ & $-$ \\
p-value (test 1) & $0.0000$ & $0.0000$ & $0.0000$ & $0.0000$ & $0.0000$ & $-$ \\
$\Delta_{aic}$ & $23.19$ & $0.0000$ & $34.4$ & $41.54$ & $0.58$ & $76.28$ \\
$\omega_i$ & $0.0000$ & $0.5422$ & $0.0000$ & $0.0000$ & $0.4578$ & $0.0000$ \\
LL w/o $\lambda$ & $100449.77$ & $100468.94$ & $100451.29$ & $100465.52$ & $100463.55$ & $100344.6$ \\
D2 (LR test 2) & $5716.38$ & $5701.24$ & $5702.13$ & $5666.53$ & $5711.43$ & $5869.62$ \\
p-value (test 2) & $0.0000$ & $0.0000$ & $0.0000$ & $0.0000$ & $0.0000$ & $0.0000$ \\

\multicolumn{4}{l}{\textit{Cotton:}} \\
LL & $90447.33$ & $90453.91$ & $90444.4$ & $90436.31$ & $90461.58$ & $90420.07$ \\
AIC & $-180856.66$ & $-180869.83$ & $-180850.8$ & $-180834.61$ & $-180885.15$ & $-180806.15$ \\
BIC & $-180745.78$ & $-180758.95$ & $-180739.93$ & $-180723.74$ & $-180774.28$ & $-180706.94$ \\
D1 (LR test 1) & $54.51$ & $67.68$ & $48.66$ & $32.46$ & $83$ & $-$ \\
p-value (test 1) & $0.0000$ & $0.0000$ & $0.0000$ & $0.0000$ & $0.0000$ & $-$ \\
$\Delta_{aic}$ & $28.49$ & $15.32$ & $34.35$ & $50.54$ & $0.0000$ & $79$ \\
$\omega_i$ & $0.0000$ & $0.0000$ & $0.0000$ & $0.0000$ & $1$ & $0.0000$ \\
LL w/o $\lambda$ & $89507.22$ & $89503.73$ & $89502.67$ & $89499.2$ & $89519.33$ & $89486.45$ \\
D2 (LR test 2) & $1880.22$ & $1900.38$ & $1883.46$ & $1874.22$ & $1884.5$ & $1867.25$ \\
p-value (test 2) & $0.0000$ & $0.0000$ & $0.0000$ & $0.0000$ & $0.0000$ & $0.0000$ \\

\multicolumn{4}{l}{\textit{Soybeans:}} \\
LL & $139908.14$ & $139917.14$ & $139914.31$ & $139904.93$ & $139937.82$ & $139891.53$ \\
AIC & $-279772.29$ & $-279790.29$ & $-279784.61$ & $-279765.85$ & $-279831.64$ & $-279743.07$ \\
BIC & $-279643.89$ & $-279661.9$ & $-279656.22$ & $-279637.46$ & $-279703.25$ & $-279626.35$ \\
D1 (LR test 1) & $33.22$ & $51.22$ & $45.55$ & $26.78$ & $92.57$ & $-$ \\
p-value (test 1) & $0.0000$ & $0.0000$ & $0.0000$ & $0.0000$ & $0.0000$ & $-$ \\
$\Delta_{aic}$ & $59.35$ & $41.35$ & $47.03$ & $65.79$ & $0.0000$ & $88.57$ \\
$\omega_i$ & $0.0000$ & $0.0000$ & $0.0000$ & $0.0000$ & $1$ & $0.0000$ \\
LL w/o $\lambda$ & $138765$ & $138769.34$ & $138770.59$ & $138771.53$ & $138794.35$ & $138746.48$ \\
D2 (LR test 2) & $2286.28$ & $2295.6$ & $2287.43$ & $2266.79$ & $2286.94$ & $2290.12$ \\
p-value (test 2) & $0.0000$ & $0.0000$ & $0.0000$ & $0.0000$ & $0.0000$ & $0.0000$ \\

\multicolumn{4}{l}{\textit{Sugar:}} \\
LL & $68720.5$ & $68723.13$ & $68715.48$ & $68721.87$ & $68719.02$ & $68698.83$ \\
AIC & $-137409$ & $-137414.26$ & $-137398.96$ & $-137411.74$ & $-137406.05$ & $-137369.67$ \\
BIC & $-137315.64$ & $-137320.89$ & $-137305.59$ & $-137318.37$ & $-137312.68$ & $-137287.97$ \\
D1 (LR test 1) & $43.34$ & $48.59$ & $33.29$ & $46.07$ & $40.38$ & $-$ \\
p-value (test 1) & $0.0000$ & $0.0000$ & $0.0000$ & $0.0000$ & $0.0000$ & $-$ \\
$\Delta_{aic}$ & $5.25$ & $0.0000$ & $15.3$ & $2.52$ & $8.21$ & $44.59$ \\
$\omega_i$ & $0.0000$ & $0.9597$ & $0.0000$ & $0.0402$ & $0.0000$ & $0.0000$ \\
LL w/o $\lambda$ & $67334.5$ & $67334.66$ & $67333.47$ & $67330.44$ & $67338.26$ & $67323.57$ \\
D2 (LR test 2) & $2772.01$ & $2776.94$ & $2764.02$ & $2782.86$ & $2761.52$ & $2750.54$ \\
p-value (test 2) & $0.0000$ & $0.0000$ & $0.0000$ & $0.0000$ & $0.0000$ & $0.0000$ \\

\multicolumn{4}{l}{\textit{Wheat:}} \\
LL & $100968.54$ & $100964.04$ & $100961.88$ & $100961.87$ & $100961.6$ & $100952.36$ \\
AIC & $-201899.07$ & $-201890.07$ & $-201885.77$ & $-201885.74$ & $-201885.19$ & $-201870.72$ \\
BIC & $-201788.19$ & $-201779.19$ & $-201774.88$ & $-201774.86$ & $-201774.31$ & $-201771.51$ \\
D1 (LR test 1) & $32.35$ & $23.35$ & $19.05$ & $19.02$ & $18.47$ & $-$ \\
p-value (test 1) & $0.0000$ & $0.0000$ & $0.0001$ & $0.0001$ & $0.0001$ & $-$ \\
$\Delta_{aic}$ & $0.0000$ & $9$ & $13.31$ & $13.33$ & $13.88$ & $28.35$ \\
$\omega_i$ & $1$ & $0.0000$ & $0.0000$ & $0.0000$ & $0.0000$ & $0.0000$ \\
LL w/o $\lambda$ & $98247.1$ & $98248.85$ & $98252.2$ & $98244.83$ & $98252.33$ & $98230.53$ \\
D2 (LR test 2) & $5442.87$ & $5430.37$ & $5419.36$ & $5434.08$ & $5418.53$ & $5443.65$ \\
p-value (test 2) & $0.0000$ & $0.0000$ & $0.0000$ & $0.0000$ & $0.0000$ & $0.0000$ \\
    \bottomrule
    \end{tabular}
\caption{Summary of the results obtained with the considered models estimated with the alternative dataset described in Section \ref{ss:AlternativeEstimationProcedure}:
log-likelihood, AIC, BIC, Statistics D1 and D2 of the likelihood ratio tests and their $p$-values, AIC differences $\Delta_{aic}$ and Akaike weights $\omega_i$.}
\label{tab:altdata_summary_results}
\end{table}

\begin{table}[H]
  \centering
     \begin{tabular}{ccccccccccc}
     \toprule
\multicolumn{3}{c}{\textit{Corn:}} & & \multicolumn{3}{c}{\textit{Cotton:}} & & \multicolumn{3}{c}{\textit{Soybeans:}} \\
\cmidrule(lr){1-3} \cmidrule(lr){5-7} \cmidrule(lr){9-11}
rank & $\Delta_{aic}$ & model & & rank & $\Delta_{aic}$ & model & & rank & $\Delta_{aic}$ & model \\
\cmidrule(lr){1-3} \cmidrule(lr){5-7} \cmidrule(lr){9-11}
$1$ & $0$ & Exp-sinusoidal & & $1$ & $0$ & Spiked & & $1$ & $0$ & Spiked  \\
$2$ & $0.58$ & Spiked & & $2$ & $15.32$ & Exp-sinusoidal & & $2$ & $41.35$ & Exp-sinusoidal  \\
$3$ & $23.19$ & Sinusoidal & & $3$ & $28.49$ & Sinusoidal & & $3$ & $47.03$ & Triangle  \\
$4$ & $34.4$ & Triangle & & $4$ & $34.35$ & Triangle & & $4$ & $59.35$ & Sinusoidal  \\
$5$ & $41.54$ & Sawtooth & & $5$ & $50.54$ & Sawtooth & & $5$ & $65.79$ & Sawtooth  \\
$6$ & $76.28$ & Non-seasonal & & $6$ & $79$ & Non-seasonal & & $6$ & $88.57$ & Non-seasonal  \\ \\
\multicolumn{3}{c}{\textit{Sugar:}} & & \multicolumn{3}{c}{\textit{Wheat:}} & & & & \\
\cmidrule(lr){1-3} \cmidrule(lr){5-7}
rank & $\Delta_{aic}$ & model & & rank & $\Delta_{aic}$ & model & &  &  & \\
\cmidrule(lr){1-3} \cmidrule(lr){5-7}
$1$ & $0$ & Exp-sinusoidal & & $1$ & $0$ & Sinusoidal & & & & \\
$2$ & $2.52$ & Sawtooth & & $2$ & $9$ & Exp-sinusoidal & & & & \\
$3$ & $5.25$ & Sinusoidal & & $3$ & $13.31$ & Triangle & & & & \\
$4$ & $8.21$ & Spiked & & $4$ & $13.33$ & Sawtooth & & & & \\
$5$ & $15.3$ & Triangle & & $5$ & $13.88$ & Spiked & & & & \\
$6$ & $44.59$ & Non-seasonal & & $6$ & $28.35$ & Non-seasonal & & & & \\
    \bottomrule
    \end{tabular}
\caption{Model rankings for each commodity.
Models are estimated with the alternative dataset described in Section \ref{ss:AlternativeEstimationProcedure}.
The rankings are based on AIC.
For each model, the difference with respect to the smallest AIC, $\Delta_{aic}$, is provided along with the model rank and name.}
\label{tab:altdata_models_ranking}
\end{table}


\bibliographystyle{plainnat}

\bibliography{articles,books,websites}

\end{document}